\documentclass[11pt, a4paper]{article}
\title{Epstein curves and holography of the Schwarzian action}
\usepackage[T1]{fontenc}

\usepackage{mathtools}
\usepackage{lmodern}
\usepackage{amsthm,amsmath,amsfonts,amssymb}
\usepackage{graphicx}
\usepackage{enumerate,enumitem}
\usepackage{authblk}
\usepackage{cite,url}
\usepackage[normalem]{ulem}
\usepackage{comment}
\usepackage[margin=1cm]{caption}

\usepackage{hyperref}
\hypersetup{
    colorlinks=true, 
    linkcolor=black, 
    urlcolor=black, 
    citecolor=black,
    linktoc=all 
}

\allowdisplaybreaks

\setlist[enumerate]{topsep = 1ex, leftmargin=.5cm, itemsep= -3pt}
\setlist[itemize]{topsep = 1ex, leftmargin=.5cm, itemsep= -3pt}

\let\OLDthebibliography\thebibliography
\renewcommand\thebibliography[1]{
  \OLDthebibliography{#1}
  \setlength{\parskip}{1pt}
  \setlength{\itemsep}{2pt}
}

\newtheorem{thm}{Theorem}[section]
\newtheorem{cor}[thm]{Corollary}
\newtheorem{lem}[thm]{Lemma}
\newtheorem{prop}[thm]{Proposition}

\newtheorem{definition}[thm]{Definition}

\theoremstyle{definition} 

\newtheorem{ex}[thm]{Example}

\newtheorem{remark}[thm]{Remark}

\def\weld{\varphi}
\def\cayley{\mathfrak c}

\numberwithin{equation}{section}

\usepackage{longtable,booktabs}

\usepackage[font=normal]{caption}
\usepackage{floatrow}

\usepackage{mathtools}

\usepackage[dvipsnames]{xcolor}

\global\long\def\ii{\mathfrak{i}}

\newcommand{\abs}[1]{\left\lvert #1 \right \rvert}

\newcommand{\norm}[1]{\lVert #1 \rVert}

\newcommand{\mc}[1]{\mathcal{#1}}
\newcommand{\m}[1]{\mathbb{#1}}

\def\ie{i.e.,\,}

\renewcommand\Re{\operatorname{Re}}
\renewcommand\Im{\operatorname{Im}}

\def\PSL{\operatorname{PSL}}
\def\PSU{\operatorname{PSU}}

\def\Diff{\operatorname{Diff}}
\def\QS{\operatorname{QS}}

\def\WP{\operatorname{WP}}

\newcommand{\Ep}{\operatorname{Ep}}

\def\a{\alpha}

\def\g{\gamma}

\def\t{\theta}

\def\k{\kappa}

\def\s{\sigma}
\def\S{\Sigma}

\def\o{\omega}
\def\O{\Omega}
\def\vare{\varepsilon}

\def\Chat{\hat{\m{C}}}
\def\Rhat{\hat{\m{R}}}

 \def\AdS{\operatorname{AdS}}
  \def\dS{\operatorname{dS}}
  \def\CFT{\operatorname{CFT}}

\def\dS{\operatorname{dS}}

 \def\Isom{\operatorname{Isom}}

\def\dd{\mathrm{d}}

\def\1{\mathbf{1}}

\def\ee{\mathrm{e}}

\usepackage[dvipsnames]{xcolor}

\def \Isch{I_{\operatorname{Sch}}}
\newcommand{\emu}{{\vare \mu}}

\def \Mobn{\operatorname{M\ddot{o}b}_n(\m S^1)}

\author{Franco Vargas Pallete\thanks{\protect\url{fevargas@asu.edu} School of Mathematical and Statistical Sciences, Arizona State University, Tempe, AZ, USA} \qquad
Yilin Wang\thanks{\protect\url{yilin.wang@math.ethz.ch} Eidgen\"ossische Technische Hochschule, Z\"urich, ZH, Switzerland} \qquad Catherine Wolfram\thanks{\protect\url{catherine.wolfram@yale.edu} or \protect\url{cwolfram@math.ethz.edu}, Yale University, New Haven, CT, USA and Eidgen\"ossische Technische Hochschule, Z\"urich, ZH, Switzerland}}
\begin{document}

\maketitle

\begin{abstract}
We apply Epstein's construction of hypersurfaces in the hyperbolic disk $\mathbb D$ to prove identities between the Schwarzian action on $\operatorname{PSL}_2(\mathbb R)\backslash \mathrm{Diff}^3 (\mathbb S^1)$, the length of the corresponding Epstein curve in $\mathbb D$, and the area enclosed by the Epstein curve.  These results are inspired by the holographic duality between Jackiw--Teitelboim gravity and Schwarzian field theory. We also show that the horocycle truncation used in the construction of the Epstein curve defines a renormalized length of hyperbolic geodesics in $\mathbb D$, which coincides with the logarithm of the bi-local observable of Schwarzian field theory. The construction of the Epstein curve also extends to the coadjoint orbits $\operatorname{PSL}_2^{(n)}(\mathbb R)\backslash \mathrm{Diff}^3 (\mathbb S^1)$, and we obtain the same identities for the analog of the Schwarzian action on these coadjoint orbits.

Furthermore, we show that the Schwarzian action is the derivative of the Loewner energy of the welded Jordan curve. This energy is the action functional of Schramm--Loewner evolutions and holographically expressed as a renormalized volume in hyperbolic $3$-space. As a by-product of these relations, we obtain two immediate proofs of the non-negativity of the Schwarzian action using the isoperimetric inequality and the monotonicity of the Loewner energy.

\end{abstract}

\tableofcontents

\section{Introduction}

We write $\m D = \{z \in \m C \colon |z| < 1\}$, $\m D^* = \{z \in \Chat = \m C \cup \{\infty\} \colon |z| > 1\}$, and $\m S^1 = \partial \m D  = \partial \m D^*$. We endow $\m D$ with the complete hyperbolic metric, and $\m S^1$ is its boundary at infinity. Let  $\varphi : \m S^1 \to \m S^1$ be a $C^3$ diffeomorphism. 
The \emph{Schwarzian action} of $\varphi$ is defined as
\begin{equation}\label{eq:def_Sch_action}
    \Isch (\varphi) : = \int_{0}^{2\pi} \ee^{2\ii \t} \mc S [\varphi] (\ee^{\ii \t}) \,\dd \t
\end{equation}
where $\mc S [\varphi] = (\varphi''/\varphi')' - (\varphi''/\varphi')^2/2$ is the Schwarzian derivative of $\varphi$. This is the action of the Schwarzian field theory, the same as considered in, e.g.,  \cite{BLW_schwarzian,Stanford_Witten}, up to a multiplicative and an additive constant. See Remark~\ref{rem:positivity_Losev} and Remark~\ref{rem:action_form_k*}. 

In this paper, we relate the Schwarzian action of $\varphi: \m S^1\to \m S^1$ to the following geometric quantities in 2D: 
\begin{enumerate}
    \item it is the signed hyperbolic area enclosed by the Epstein curve associated with the pushforward metric $\varphi_* \dd \t$ on $\m S^1$;
    \item it is the asymptotic isoperimetric excess of the corresponding Epstein foliation of $\m D$;
    \item it is the infinitesimal variation of the Loewner energy of the Jordan curve, whose welding homeomorphism is given by $\varphi$, along its equipotential foliation; 
    \item it is the asymptotic change in  hyperbolic area under a conformal distortion extending $\varphi$ in a neighborhood of $\m S^1$, when $\varphi$ is analytic.
\end{enumerate}
Using the second and third relations, we obtain two new proofs of the non-negativity of $\Isch$ (also shown in, e.g., \cite{BLW_schwarzian,LosevLDP}) for circle diffeomorphisms, one using the isoperimetric inequality in the hyperbolic disk and the other using the monotonicity of the Loewner energy.  Furthermore, we extend Relation 1 to the dual Epstein curve in the de Sitter space $\operatorname{dS}_2$, to the Schwarzian action associated with the Virasoro coadjoint orbits $\operatorname{M\ddot{o}b}_n (\mathbb S^1) \backslash \operatorname{Diff}(\mathbb S^1)$, and to piecewise M\"obius circle diffeomorphisms which are only $C^{1,1}$.

Although our viewpoint is purely geometric, our work is particularly motivated by the proposed holographic relationship between Schwarzian field theory \cite{BLW_schwarzian,Belokurov_exact,Stanford_Witten} and Jackiw--Teitelboim (JT) gravity \cite{JACKIW1985343,Teitelboim1983GravitationAH,Maldacena_Stanford_Yang,Ferrari_JT,EngelsoyAdS,JensenChaos}. We comment on this further at the end of the introduction, see Section~\ref{sec:intro_comments}.

Another key object in the Schwarzian field theory is the \textit{bi-local observable} $\mc O(\varphi; u,v)$ associated with the circle diffeomorphism $\varphi$ and any two distinct points $u,v\in \m S^1$. 
Bi-local observables are predicted to correspond to boundary-anchored Wilson lines for JT gravity \cite{Blommaert}.
It was also predicted in physics \cite{MertensEtAl} and verified mathematically \cite{LosevCorrelations, Belokurov_corr} that the correlation functions of bi-local observables characterize the Schwarzian field theory.
Here, we observe that bi-local observables can be expressed as the length of hyperbolic geodesics using the same truncation as in the construction of Epstein curves, and the bi-local observables along the edges of any ideal triangulation of $\m D$ fully determine the diffeomorphism up to M\"obius transformations.

\subsection{Relation to Epstein curve}\label{subsec:intro_ep}

We now give a conceptual description of the Epstein curve used in our first main result. 
The construction of 1-dimensional Epstein curves in the disk is inspired by the analogous construction in higher dimensions. 

The most studied setup is for hyperbolic $3$-manifolds. In fact, a complete convex co-compact hyperbolic $3$-manifold $M$ determines the conformal structure on its conformal boundary $\partial_\infty M$ at infinity.  
This property appeals to theoretical physicists, providing a framework for describing an $\AdS_3/\CFT_2$ holographic correspondence \cite{KrasnovSchlenker_CMP,krasnov2000holography,TT_Liouville,Maldacena}, and is also interesting from the point of view of hyperbolic geometry and Teichm\"uller theory  \cite{McMullen2000, Brock_Canary_Minsky_ending_lamination}. 

Epstein~\cite{epstein-envelopes} 
proposed a way to further associate with each conformal metric on $\partial_\infty M$ a truncation of $M$ by horospheres, where the ``size'' of the horosphere $H_z$ centered at each $z \in \partial_\infty M$ is determined by the metric tensor at $z$. The Epstein hypersurface associated with the conformal metric is the hypersurface in $M$ that is tangent to all horospheres and is called ``the envelope'' of the horospheres.  See Section~\ref{subsec:epstein_general} for more details and Figure~\ref{fig:horocycle} for an illustration. 
Moreover, for hyperbolic $3$-manifolds, the renormalized volume enclosed by the Epstein surface can be viewed as a renormalized Einstein--Hilbert action of $M$, and is shown to have the same conformal anomalies as the Liouville action defined for the conformal metrics on $\partial_\infty M$ and has deep links to Teichm\"uller theory \cite{TT_Liouville,KrasnovSchlenker_CMP}.

\begin{figure}
    \centering
    \includegraphics[width=.45\textwidth]{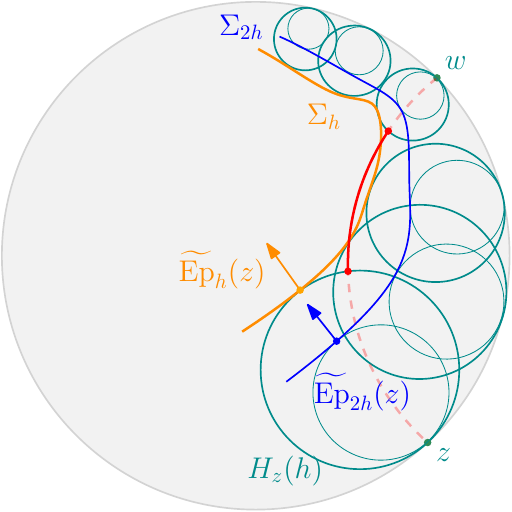}
    \caption{Illustration of part of the Epstein curves associated with a metric $h$ (orange) and $2h$ (blue), where $H_z$ is a horocycle determined by $h$ at $z$, the Epstein curve is the envelope of the family of corresponding horocycles $(H_z)_{z \in \m S^1}$, and $\widetilde{\Ep}_h (z) \in T^1 \m D$ is the outward unit normal with base point $\Ep_h (z) \in \m D$. The hyperbolic length of the red geodesic segment gives a positive renormalized length between $z, w \in \m S^1$ as in Proposition~\ref{prop:intro_length}.}
    \label{fig:horocycle}
\end{figure}

Although hyperbolic $3$-manifolds are by far the most studied, Epstein's construction is valid in hyperbolic spaces of any dimension. Given that the Schwarzian field theory shall be obtained from a dimensional reduction of the Liouville theory \cite{Mertens}, it appears natural to us to consider the Epstein construction in the two-dimensional hyperbolic disk $\m D$, whose conformal boundary is identified with $\m S^1$. 
Our first result is the following.

\begin{thm}[See  Theorem~\ref{thm:SchArea}]\label{thm:Ep_Sch}
    Let $\varphi : \m S^1 \to \m S^1$ be a $C^3$ diffeomorphism. Let $\dd \t$ be the Euclidean arc-length on $\m S^1$ with total length $2\pi$ and $h = \varphi_* \dd \t$ the pushforward of $\dd \t$ by $\varphi$. Let $A(\Ep_h)$ denote the signed hyperbolic area enclosed by the Epstein curve associated with  $h$ and $L(\Ep_h)$ the signed hyperbolic length of the Epstein curve, then we have
    \begin{align}\label{eq:Sch_L_A}
        \Isch(\varphi) = L(\Ep_h) =  -A(\Ep_h).
    \end{align}
\end{thm}
See Definitions~\ref{def:length_curv} and \ref{def:signed_A} for the notion of \emph{signed} length and enclosed area of Epstein curves.

In Section~\ref{subsec:dS}, we describe \emph{dual Epstein curves} in de Sitter space $\dS_2$ and express the Schwarzian action in terms of geometric quantities in $\dS_2$ (Proposition~\ref{prop:dualISch}).

We also consider piecewise M\"obius circle diffeomorphisms which have only $C^{1,1}$ regularity in Section~\ref{subsec:pw}. In this case, the Schwarzian action is interpreted in the distributional sense, the Epstein map sends each piece of the circular arc in $\m S^1$ to a point in $\m D$, and we join the points by arcs of horocycles and show that Theorem~\ref{thm:Ep_Sch} is still valid with this completion. 
This class of circle homeomorphisms has been studied extensively \cite{Penner1993UniversalCI,SWW_shears,Wang_optimization,RW} for its link to shear coordinates on the universal Teichm\"uller space  (see next section)  and to real projective structures on the punctured sphere \cite{BJRW_pwg}. 

\bigskip 

More generally, the Epstein map $\Ep_h \colon \m S^1 \to \m D$ is well-defined for any metric $h = \ee^{\s} \dd \t$ on $\m S^1$ where $\s \in C^1 (\m S^1, \m R)$. See Section~\ref{subsec:eps_1}. If $\s$ is $C^2$, then the curvature of the Epstein curve is well-defined.  
In other words, Theorem~\ref{thm:Ep_Sch} studies the special case where the metric $h$ has total length $2\pi$. We also consider metrics of the form $h_t = \ee^t h$, where $t \in \m R$, which then encompasses metrics of arbitrary total length. Moreover, for large enough $t_0$, the Epstein curves associated with the family $(h_t = \ee^t h)_{t \ge t_0}$ form an equidistant foliation in a neighborhood of $\m S^1$, whose leaves converge to $\m S^1$ as $t \to \infty$. See 
Section~\ref{sec:isoperimetric} for details and Figure~\ref{fig:example_1} for an example (see Appendix~\ref{sec:examples} for additional examples). 

\begin{figure}[ht]
    \centering
    \includegraphics[width=\textwidth]{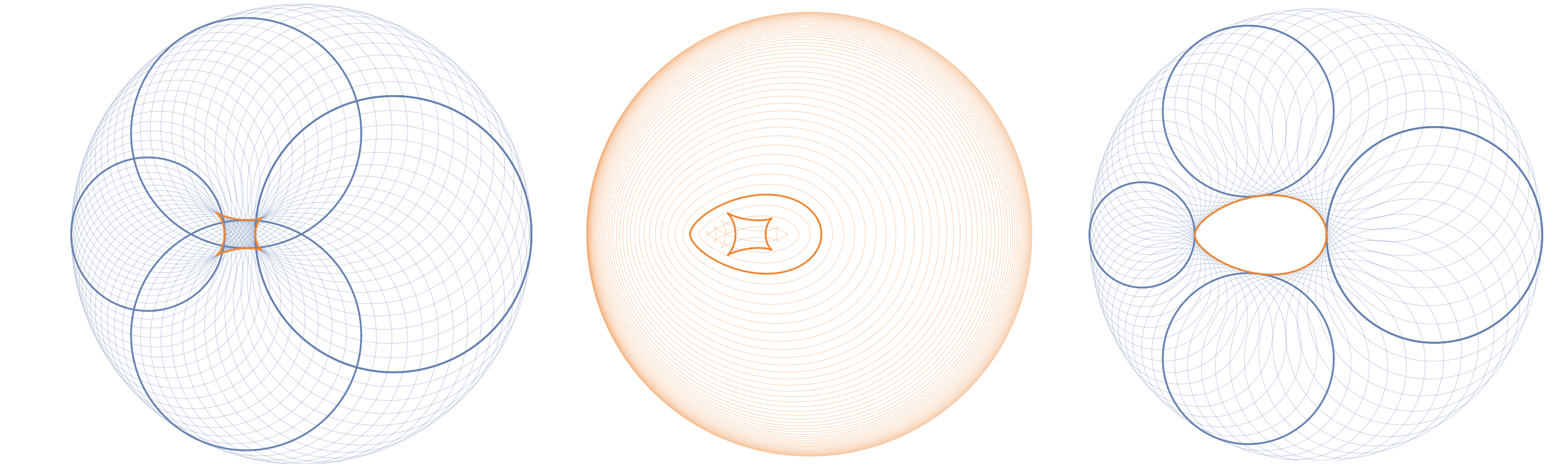}
    \caption{Left: Epstein curve (orange) and the horocycles (blue) associated with the metric $h = \varphi_* \dd \t$, where $\varphi(\theta) = \frac{1}{2}\sin(\theta) + \theta$. Right: Epstein curve and the horocycles associated with the metric $e^{1/2} h $. Middle: the family of Epstein curves associated with the metrics $\{h_t = \ee^{t} h\}_{t \ge 0}$ which form a foliation near $\m S^1$.}
    \label{fig:example_1}
\end{figure}

We obtain the following result:

\begin{lem}[See Lemma~\ref{lem:asymptoticLA}]  \label{lem:intro_t}
If $h$ is a $C^2$ metric of total length $2\pi$ and $h_t : = \ee^t h$ for $t \in \m R$, then
    \begin{enumerate}[label=(\alph*)]
        \item  $L(\Ep_h) + A(\Ep_h) = 0$;
        \item $L(\Ep_{h_t}) = 2\pi\sinh(t) + \ee^{-t} L(\Ep_h)$;
        \item $A(\Ep_{h_t}) = 2\pi(\cosh(t)-1) +\ee^{-t}A(\Ep_h)$.
    \end{enumerate}
\end{lem}

When $h = \dd \t$ and $t > 0$, then  $L(\Ep_h) = A(\Ep_h) = 0$ and the image of $\Ep_{h_t}$ is the circle of hyperbolic radius $t$ centered at $0 \in \m D$ which has length $L(\Ep_{h_t}) = 2\pi\sinh(t)$ and bounds a disk of area $A(\Ep_{h_t}) = 2\pi(\cosh(t)-1)$.  See Example~\ref{ex:Ep_point}.

\begin{cor}\label{cor:fol}
  Let $\varphi : \m S^1 \to \m S^1$ be a $C^3$ diffeomorphism and $h = \varphi_* \dd \t$. We have
  $$\Isch (\varphi) = \lim_{t \to \infty} \ee^t \left(L(\Ep_{h_t}) - 2\pi \sinh (t)\right) = \lim_{t \to \infty} \ee^t \left( 2\pi (\cosh (t)-1) - A(\Ep_{h_t})\right).$$
\end{cor}
The \emph{isoperimetric profile} is the function $J : \m R_+ \to \m R_+$, 
$$J(A) = 2\pi\sinh\left(\cosh^{-1}\left(\frac{A}{2\pi}+1\right)\right),$$
such that any domain in $\m D$ with hyperbolic area $A$ has boundary length larger or equal to $J(A)$. The equality is realized only by round disks.

Our next result expresses the Schwarzian action as an asymptotic isoperimetric excess along the Epstein equidistant foliation. We obtain the non-negativity of $\Isch$ using the isoperimetric inequality for the hyperbolic disk. We believe similar proofs are known in the physics literature.
\begin{thm}[See Theorem~\ref{thm:excess} and Corollary~\ref{cor:non-negativeSch_isoperimetric}]
    Let $\varphi:\m S^1\rightarrow\m S^1$ be a $C^3$ diffeomorphism and $h_t=\ee^t \,\varphi_* \dd \t$, then
    \begin{align}
        2\Isch(\varphi) = \lim_{t\rightarrow+\infty} \ee^t\Big( L(\Ep_{h_t}) - J(A(\Ep_{h_t}))  \Big).
    \end{align}
   Moreover, $\Isch (\varphi)\ge 0$ with equality if and only if $\varphi$ is M\"obius.
\end{thm}

\bigskip 

We can also use the Epstein construction to give a geometric interpretation of bi-local observables, which are important observables in the Schwarzian field theory. See Section~\ref{sec:bi-local} for more background.
More precisely, for $\varphi\in \Diff^1(\m S^1)$, we let $(H_z)_{z\in \m S^1}$ denote the horocycles associated with $\varphi_*\dd \theta$. For $u,v\in \m S^1$ distinct, we define the \textit{renormalized length} $\mathrm{RL}_{\varphi}(u,v)$ to be the signed hyperbolic distance from $H_{\varphi(u)}$ to $H_{\varphi(v)}$ along the geodesic $(\varphi(u),\varphi(v))$. The sign is positive if $H_{\varphi(u)}, H_{\varphi(v)}$ are disjoint and negative otherwise. The bi-local observables of Schwarzian field theory correspond to the exponential of the renormalized length.

\begin{prop}[See Proposition~\ref{prop:bilocal_is_length}]\label{prop:intro_length}
    Fix $\varphi\in \Diff^1(\m S^1)$ and $u,v\in \m S^1$ distinct. 
    The bi-local observable and renormalized length are related by: 
    \begin{align*}
        \mc O(\varphi;u,v)^2 := \frac{|\varphi'(u)\varphi'(v)|}{|\varphi(u)-\varphi(v)|^2} =\frac{1}{4}\exp(-\mathrm{RL}_{\varphi}(u,v)). 
    \end{align*}
\end{prop}

Renormalized length (and bi-local observables) are naturally viewed as functions on the space of geodesics in $\m D$, which are parametrized by their endpoints $u,v\in \m S^1$. This is called the \textit{kinematic space} $\mc K = \m S^1\times \m S^1\setminus \{(u,u):u\in \m S^1\}$. Furthermore, each $\varphi$ gives a Lorentzian metric on $\mc K$
\begin{align}
    \mc O(\varphi;u,v)^2 \, \dd u \,\dd v = \frac{1}{4}\exp(-\mathrm{RL}_{\varphi}(u,v)) \,\dd u \,\dd v
\end{align}
 with constant curvature, hence solving the Liouville equation and isometric to the de Sitter space $\dS_2$.

In Section~\ref{subsec:inverse_pb}, we show that it suffices to know $\{\mc O(\varphi; u,v)\}$ where $(u,v)$ runs over the set of edges of any ideal triangulation of $\m D$ (see, e.g., Figure~\ref{fig:farey}) to fully recover $\varphi$ up to M\"obius transformations (Proposition~\ref{prop:ideal_triangulation}). This is closely related to the parametrization of the decorated universal Teichm\"uller space by log $\Lambda$-lengths by Penner \cite{PennerBook,Penner1993UniversalCI} and the diamond shears coordinates studied by the last two authors \cite{SWW_shears}.

\bigskip
Now let us turn to another variant of the Schwarzian action.  For $n \in \m Z_{>0}$, let
$$\Isch^n (\phi) : = \Isch (\phi) + \frac{1-n^2}{2} \int_0^{2\pi} |\phi'(\ee^{\ii \t})|^2 \, \dd \t,$$ 
which is invariant under post-composition by elements of the subgroup $\Mobn \subset \Diff(\m S^1)$ obtained by conjugating $\PSU(1,1)$ by the power map $z \mapsto z^n$.  In other words,  $\Mobn  \simeq \PSL^{(n)}(2,\m R)$ is the subgroup whose Lie algebra is spanned by $(L_{-n}, L_0, L_n)$, where $(L_k = \ii \ee^{\ii k \t} \dd /\dd \t)_{k\in \m Z}$ is a basis of the Lie algebra of $\Diff (\m S^1)$.
Therefore, $\Isch^n (\phi)$ is of particular interest in the study of the special Virasoro coadjoint orbits $\Mobn \backslash \Diff(\m S^1)$ \cite{witten88,Alekseev}. 

We show that this variant of the Schwarzian action also has a geometric interpretation in $\m D$.  
Let $\phi:\m S^1 \rightarrow \m S^1$ a $C^3$ diffeomorphism. Let $\varphi (z) : =\phi (z)^n$ be the $n$-fold cover of $\m S^1$, where $\phi (z)^n$ is the $n$-th power of $\phi(z)$ for multiplication in $\m S^1$ (as complex numbers). Let  $h=\varphi_*(\dd\t)$ be the associated multi-valued metric on $\m S^1$ and $\Ep_h$ the associated Epstein curve. 

\begin{thm}[See Theorem~\ref{thm:IschnLength}]\label{thm:intro_n_fold}
    We have the following identity:
    \[
    \Isch^n(\phi) = L(\Ep_h) = -A(\Ep_h)-2\pi(n-1).
    \]
\end{thm}
See Figure~\ref{fig:n_fold_cover} and Figure~\ref{fig:n_fold_cover2} for examples.

\begin{figure}[ht]
    \centering
    \includegraphics[width=\linewidth]{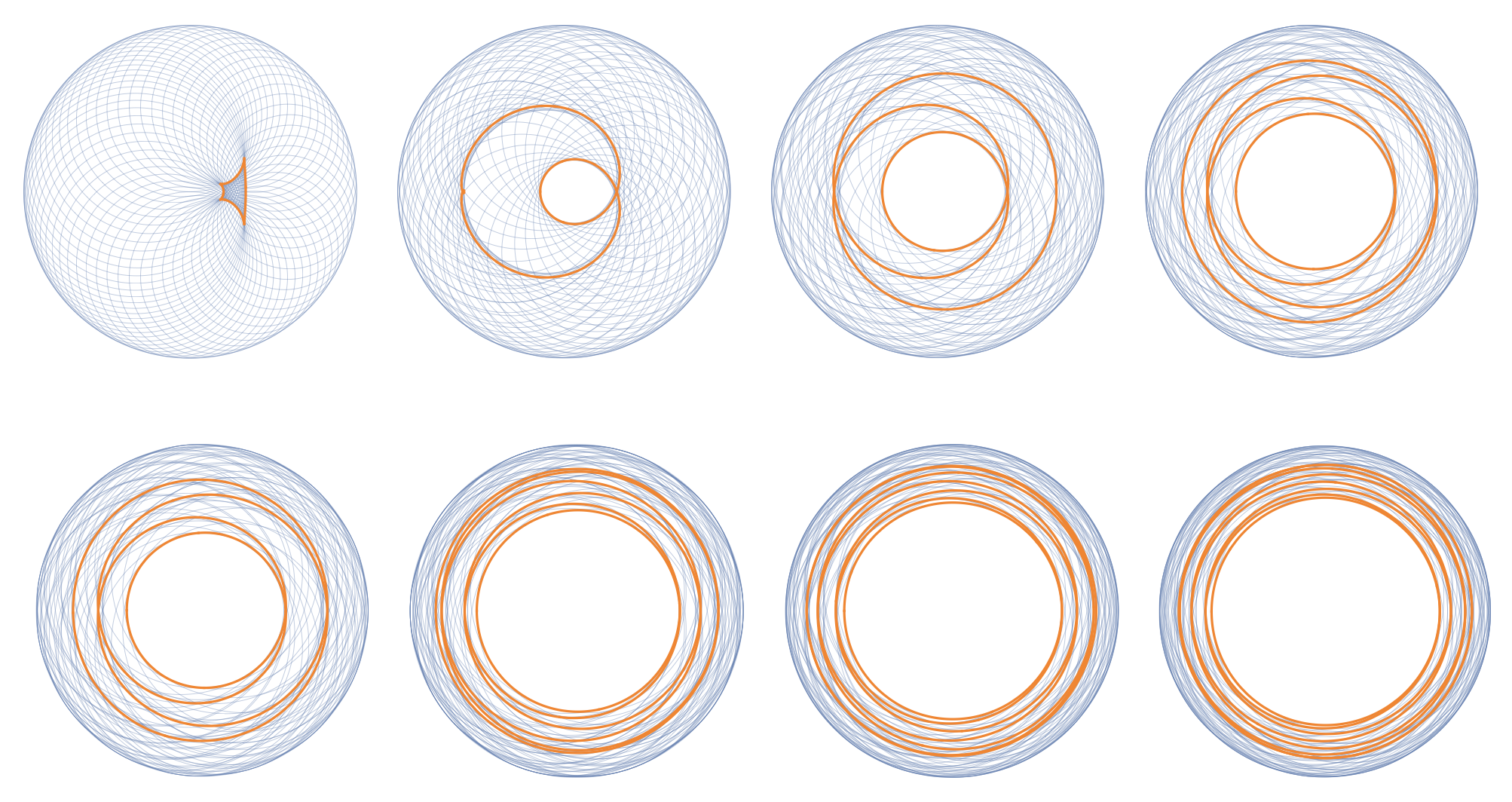}
    \caption{The horocycles and Epstein curve for $\varphi_*\dd \theta$, where $\varphi= (\phi(z))^n$, and $\phi^{-1}(\theta) = \frac{1}{2} \sin(\theta) + \theta$, for $n$ from $1$ to $8$.}
    \label{fig:n_fold_cover}
\end{figure}

\subsection{Relation to Loewner energy}
The Schwarzian action $\Isch$ is invariant under post-composition by M\"obius transformations in $\PSU(1,1)$, hence, it is defined on the space $\PSU(1,1)\backslash \Diff^3 (\m S^1)$. This is a subspace of the \emph{universal Teichm\"uller space} $T(1): = \PSU(1,1)\backslash \QS (\m S^1)$ which is modeled on the space of quasisymmetric homeomorphisms of the circle. The space $T(1)$ is further identified with the space of quasicircles via conformal welding:
$$ \PSL(2,\m C) \backslash \{\text{Quasicircles in } \Chat\} \xleftrightarrow[1:1]{\text{conformal welding}}  \PSU(1,1)\backslash \QS (\m S^1)/\PSU(1,1)$$
where $\PSL(2,\m C)$ denotes the group of M\"obius transformations of $\Chat$.

Another important quantity --- the Loewner energy $I^L$ which is also known as the universal Liouville action --- arises as a K\"ahler potential for the unique homogeneous K\"ahler metric on $T(1)$ and has been studied extensively \cite{TT06,bbvw,W2,michelat2021loewner,johansson_viklund,bishop-wp}. The Loewner energy is finite on the space of Weil--Petersson quasisymmetric circle homeomorphisms $\WP(\m S^1) \subset \QS (\m S^1)$, which contains all $C^{1.5+}$ regular circle homeomorphisms. Moreover, $I^L$ is bi-invariant under the action of $\PSU(1,1)$, and therefore is also considered a function on the space of Jordan curves modulo $\PSL(2,\m C)$.
See Section~\ref{subsec:Loewner} for more details.

Our next result shows that $\Isch$ is equal to a derivative of $I^L$.
For this, let $\g$ be a $C^{3,\alpha}$ curve for some $\a > 0$, $f$ (resp., $g$) be a conformal map $\m D \to \O$ (resp., $\m D^* \to \O^*$), where $\O$ (resp., $\O^*$) is the connected component of $\Chat \smallsetminus \g$ which does not contain $\infty$ (resp., which contains $\infty$). 
By Kellogg's theorem, $f$ and $g$ extend to $C^{3,\a}$ diffeomorphisms on the closures $\overline{\m D}$ and $\overline{\m D^*}$. In particular, the welding homeomorphism of $\g$, $\varphi_\g := g^{-1} \circ f |_{\m S^1}$ is a $C^{3,\a}$ diffeomorphism.

\begin{thm}[See Theorem~\ref{thm:var_IL_Sch}]\label{thm:var_IL_Sch_intro}
Let $\g$ be a $C^{3,\a}$ Jordan curve for some $\a > 0$. Using the notations above, we consider the family of \emph{equipotentials} $\Big(\g^\vare := f((1-\vare) \m S^1)\Big)_{0 \le \vare <1}$ of $\O$. 
We have 
    $$\frac{\dd I^{L} (\g^{\vare})}{\dd \vare}\Big|_{\vare = 0} = - \frac{2}{\pi} \, \Isch (\varphi_\g). $$
\end{thm}
 Note that the definition of equipotentials depends not only on $\g$, but also on the choice of the conformal map $f$. This is consistent with the fact that $\Isch$ is only left-invariant by M\"obius transformations of the circle.
\begin{remark}[Second proof of non-negativity of $\Isch$]\label{rem:non-negativity_loewner}
    It was shown in \cite[Cor.~1.5]{VW1} that $\vare \mapsto I^{L} (\g^{\vare})$ is strictly decreasing unless $\g$ is a circle (then $ I^{L} (\g^{\vare})$ is constantly zero). From this, we also see that $\Isch(\varphi) \ge 0$ for $\varphi\in \Diff^3(\m S^1)$.
\end{remark}

Inspired by the $\AdS_3/\CFT_2$ correspondence and using the same Epstein construction in hyperbolic $3$-space, Bridgeman, Bromberg, and the first two authors \cite{bbvw} showed that $I^L (\g)$ equals $4/\pi$ times the \emph{renormalized volume} 
\begin{equation}\label{eq:VR_IL}
V_R(\g) : = \operatorname{Vol} (N_\g) - \frac{1}{2} \int_{\partial N_\g} H \dd a
\end{equation}
of the $3$-manifold $N_\g$ bounded by the two Epstein surfaces associated with the hyperbolic metrics in $\O$ and $\O^*$, where $H$ is the mean curvature and $\dd a$ the area form induced by the hyperbolic metric of $\m H^3$.
Given Theorem~\ref{thm:Ep_Sch} and Theorem~\ref{thm:var_IL_Sch_intro}, we obtain that a variation of a three-dimensional renormalized volume equals a two-dimensional renormalized area. However, we do not have a direct geometric proof of this fact, as the correspondence between Jordan curves and circle homeomorphisms via conformal welding is rather implicit. We also mention that, as shown in \cite{Wang_optimization}, conformal welding is better viewed as a correspondence between curves on the conformal boundaries of the spaces $\m H^3$ and $\AdS^{2,1}$.

\subsection{Relation to conformal distortion}

We also show that the Schwarzian action appears asymptotically under conformal distortion of circles $\m S_r = r \m S^1$ where $r$ is close to $1$. Namely, fix a compact set $K\subset \m D$, and let $\varphi:\m D \setminus K\to \m D$ be a locally conformal map such that the analytic function $\varphi\mid_{\m S^1}$ maps $\m S^1$ to itself (we do not require $\varphi\mid_{\m S^1}$ to be injective). 
Let $\rho_{\m D}$ denote the hyperbolic metric on the disk.
\begin{thm}[See Theorem~\ref{thm:conformal_distortion}]
    Let $\rho = \varphi^\ast(\rho_{\m D})$. Let $\dd s, \dd s_{\rho}$ denote the length measure on $\m S_r$ for $\rho_{\m D},\rho$ respectively, and similarly let $k_{\rho_{\m D}},k_{\rho}$ denote the respective geodesic curvatures. Then
    \begin{equation}
        \lim_{r\to 1^-} \bigg[ \int_{\m S_r} k_{\rho}(s) \,\dd s_{\rho} - \int_{\m S_r} k_{\rho_{\m D}}(s)\, \dd s \bigg]= -\frac{2}{3} \int_0^{2\pi} \mc{S}[\varphi](\mathrm{e}^{\ii \theta}) \mathrm{e}^{2 \mathfrak{i} \theta} \, \dd \theta = -\frac{2}{3} \Isch(\varphi).
    \end{equation}
\end{thm}
Results of this nature are well-known in the physics literature, 
see, e.g., \cite{KitaevSuh,Mertens,blau2024}. 

We include this result, which follows essentially from the Taylor expansion of $\varphi$, in part for comparison with the results for Epstein curves. 
We note that the theorem above does not require that $\varphi: \m S^1\to \m S^1$ is a homeomorphism. 
However, even when $\varphi$ is a homeomorphism,
the Epstein foliation considered in Section~\ref{subsec:intro_ep} and the conformal foliation formed by $\varphi(\m S_r)$ are, in general, not the same. Such conformal foliation was also considered in \cite{VW2}, which observed the appearance of the Schwarzian action in the variation of Loewner--Kufarev energy of the foliation. By the above theorem and the Gauss--Bonnet theorem, we again have a relationship between the Schwarzian action and hyperbolic area enclosed by $\varphi(\m S_r)$.

\begin{cor}[See Corollary~\ref{cor:conformal_distortion_area}]
    Suppose that $\varphi: \m D \to \m D$ is quasiconformal with Beltrami coefficient supported in the compact set $K\subset \m D$. 
    Then 
    \begin{align}
        \lim_{r\to 1^-} \bigg[A(\varphi(\m S_r)) - A(\m S_r)\bigg] = -\frac{2}{3}\Isch(\varphi),
    \end{align}
    where $A(\gamma)$ denotes the hyperbolic area of the region enclosed by the curve $\gamma$.
\end{cor}

We may compare to Corollary~\ref{cor:fol} to see that the asymptotic behavior of the hyperbolic area bounded by the conformal foliation is different from that of the Epstein foliation.

\subsection{Comments and further directions}\label{sec:intro_comments}

The Schwarzian field theory has been studied extensively in physics and was recently formulated in a mathematically rigorous way using tools from probability theory \cite{BLW_schwarzian,Belokurov_exact,LosevLDP,LosevCorrelations}. 
In physics, the Schwarzian field theory has also been linked to the Sachdev--Ye--Kitaev (SYK) random matrix model \cite{Maldacena_Stanford}, low dimensional black holes \cite{ KitaevSuh}, a large $c$ limit of two-dimensional Liouville CFT \cite{MertensEtAl}, and Virasoro coadjoint orbits \cite{Alekseev,Stanford_Witten}. 

While our viewpoint is geometric, this work is motivated by the proposed holographic relationship between Schwarzian field theory \cite{BLW_schwarzian,Belokurov_exact,Stanford_Witten} and Jackiw--Teitelboim (JT) gravity on the disk with negative curvature \cite{JACKIW1985343,Teitelboim1983GravitationAH,Maldacena_Stanford_Yang,Ferrari_JT,EngelsoyAdS,JensenChaos}. See, e.g., \cite{Mertens} and the references therein. 

JT gravity on the disk is a model for random metrics of a fixed constant curvature, but unlike Schwarzian field theory, it does not yet have a rigorous mathematical formulation. The action of JT gravity is 
\begin{align*}
   I_{\operatorname{JT}} (g, \Phi) =  -\frac{1}{4\pi} \int_D \Phi (R+2) \sqrt{g} \,\dd x^2 - \frac{1}{2\pi} \Phi_b \int_{\partial D}  k \,\dd s,
\end{align*}
where $R$ is the curvature of the metric $g$ on $D$, and $k$ is geodesic curvature. Fixing the boundary condition $\Phi$ to be a constant $\Phi_b$ and setting $R = -2$, the action reduces to $\int_{\partial D} k \,\dd s$ up to a multiplicative constant. By the Gauss--Bonnet formula, for the disk, this total curvature is (up to an additive constant) also equal to the total area of $D$ under the metric $g$. A number of models in physics have been proposed to study random JT disks, many of which are defined as disks bounded by random curves. See, e.g., \cite{KitaevSuh} for an early work on JT random metrics on regions enclosed by Brownian loops, \cite{StanfordYang-selfavoiding} on random embedded disks bounded by self-avoiding walk loops, and \cite{Ferrari_JT} for random disks given by immersions of the disk into the spaces of constant curvature and self-overlapping polygons. In mathematics, motivated by the model introduced in \cite{Ferrari_JT}, self-overlapping polygons immersed in $\m Z^2$ were also studied in \cite{budd2025discreteflatdisksrigid}.

The Epstein map on the circle has an extension to a differentiable map of the disk in into the hyperbolic plane (see Definition~\ref{def:signed_A}), and provides a unified geometric framework to understand the Schwarzian action (Theorem~\ref{thm:Ep_Sch} equates the action of Schwarzian field theory with the action of JT gravity for the metric on the immersed disk given by extending the Epstein map), its analogs for Virasoro coadjoint orbits (Theorem~\ref{thm:intro_n_fold}), and the correlation functions of Schwarzian field theory (Proposition~\ref{prop:intro_length}). 

However, the Epstein construction does not apply directly in the random setting, because it requires more regularity. As is typical in random models, the Schwarzian action itself is infinite for the random diffeomorphisms of Schwarzian field theory, which have regularity only $C^{3/2-}$. Correspondingly, the Epstein curve for the random diffeomorphism is not defined without regularization. In Section~\ref{subsec:pw}, we extend the construction of the Epstein curve to $C^{1,1}$ piecewise M\"obius circle homeomorphisms. The regularization needed to extend classical constructions to the random setting often requires new ideas (e.g.\ as in the case of Liouville quantum gravity \cite{DuplantierSheffield}), and the more substantial regularization needed to understand this construction in the random setting will be addressed in a subsequent work.


The relationship between the Schwarzian action and the Loewner energy derived in Theorem~\ref{thm:var_IL_Sch_intro} also suggests that there is a link between SLE and Schwarzian field theory, which is under investigation but far from being clear.  The Loewner energy is the action of SLE loop measures as shown in \cite{carfagnini2023onsager}. SLE loop measures are indexed by a parameter $\k \in (0,8)$ and are rigorously constructed in \cite{zhan2020sleloop}.  They describe the interfaces appearing in the scaling limit of critical 2D lattice models, where the parameter depends on the specific model and encodes the central charge of the corresponding CFT.

\bigskip

\noindent\textbf{Acknowledgments:}
We thank Fredrik Viklund for useful discussions about Theorem~\ref{thm:var_IL_Sch_intro}, Jérémie Toulisse for the suggestion of looking into the dual Epstein curve in de Sitter space, Timothy Budd, Nicolas Delporte, Frank Ferrari, Ilya Losev, and Zhenbin Yang for discussions about Schwarzian action and JT gravity, and Anton Alekseev for discussions about Virasoro coadjoint orbits. This work is funded by European Union (ERC, RaConTeich, 101116694)\footnote{Views and opinions expressed are however those of the author(s) only and do not necessarily reflect those of the European Union or the European Research Council Executive Agency. Neither the European Union nor the granting authority can be held responsible for them.}, the Swiss State Secretariat for Education, Research and Innovation (SERI): MB25.00004, and grants from the Simons Foundation International [SFI-MPS-PP-00012621-19]. C.W. thanks IHES for their hospitality during an early phase of this work and is partially supported by NSF grant DMS-2401750.

\section{Preliminaries}

In this section, we collect previously known results about the Schwarzian action, define the Loewner energy, and describe Epstein maps in arbitrary dimension.

\subsection{Schwarzian action}

Let $\m S^1 = \partial \m D \subset \m C$ be the unit circle centered at $0$. 
Let  $\varphi : \m S^1 \to \m S^1$ be a $C^3$ diffeomorphism. The Schwarzian action of $\varphi$ is given by 
$$
    \Isch (\varphi) : = \int_{\m S^1} 
   - \ii z \, \mc S[\varphi](z) \,\dd z = \int_{0}^{2\pi} \ee^{2\ii \t} \mc S [\varphi] (\ee^{\ii \t}) \,\dd \t
$$
where $\mc S [\varphi] = (\varphi''/\varphi')' - (\varphi''/\varphi')^2/2$ is the Schwarzian derivative, $\dd z$ the contour integral on $\m S^1$, and for $z = \ee^{\ii \t} \in \m S^1$, 
$$\varphi'(z) = \frac{\dd \varphi (z)}{\dd \theta} / \frac{\dd z}{\dd \theta} = \frac{1}{\ii z} \frac{\dd \varphi (z)}{\dd \theta} .$$  

\begin{remark}\label{rem:phi_complex_derivative}
    If $\varphi$ extends to a holomorphic immersion in a neighborhood of $\m S^1$, then the above convention of $\mc S[\varphi]$ coincides with the Schwarzian derivative of a holomorphic immersion where $\varphi'$ is the usual complex derivative.
\end{remark}

\begin{remark}
    In Section~\ref{subsec:pw}, we will discuss a generalization to $C^{1,1}$ diffeomorphisms, where we interpret the Schwarzian derivative in the distributional sense.
\end{remark}

The Schwarzian derivative satisfies a chain rule:
$$\mc S[f\circ g] = \mc S [f] \circ g (g')^2 + \mc S[g]$$
and $\mc S[f] \equiv 0$ if and only if $f$ is a M\"obius map in $\PSL(2,\m C)$, i.e., $f(x) = (ax+b)/(cx+d)$ for some $a,b,c,d \in \m C$ and $ad -bc \neq 0$.

\begin{remark} \label{rem:positivity_Losev}
    If we conjugate $\varphi$ by the map $[0,1] \to \m S^1: \theta \mapsto \ee^{2 \pi \ii \t}$, we obtain a diffeomorphism $\tilde{\varphi}: [0,1]_{/0\sim 1} \to [0,1]_{/0\sim 1}$, such that $\varphi(\ee^{2\pi \ii \t}) = \ee^{2\pi \ii \tilde{\varphi} (\t)}$. Then we have
$$\mc S[\varphi](\ee^{2\pi \ii \t}) \ee^{4\pi \ii \t} = - \frac{\mc S[\tilde{\varphi}] (\t)}{4\pi^2} - \frac{ [\tilde{\varphi}'(\t)]^2 }{2} + \frac{1}{2}.$$
Hence,
\eqref{eq:def_Sch_action} can be rewritten as
\begin{equation}\label{eq:action_Losev}
    2 \pi \Isch (\varphi) = -\int_0^1 \left(\mc S[\tilde{\varphi}] (\t) + 2\pi^2 [\tilde{\varphi}'(\t)]^2 \right)  \,\dd \t + 2\pi^2 = \mc I(\tilde{\varphi}) + 2\pi^2
\end{equation}
where $\mc I(\cdot)$ is the Schwarzian action used in, e.g., \cite{BLW_schwarzian}, which is also shown to be larger or equal to $-2\pi^2$ in \cite{LosevLDP}. See also Remark~\ref{rem:action_form_k*}.
\end{remark}

Let $\a : \m S^1 \to \m S^1$ be a M\"obius map preserving $\m S^1$, namely, 
$$\a (x) = \frac{a x + b}{\bar b x + \bar a}, \qquad \text{for some } a, b \in \m C, \quad |a|^2 - |b|^2 = 1,$$
which we denote as
$\a\in \PSU(1,1)$. Then we have 
$$\mc S[\varphi] = \mc S[ \a \circ \varphi]. $$
Summarizing the above, we obtain:
\begin{lem}\label{lem:sch_real}
   The Schwarzian action $\Isch$ is non-negative on $\Diff^3(\m S^1)$ and invariant under the left action of $\PSU(1,1)$. 
\end{lem}
We will provide two different proofs of the non-negativity using our descriptions of the Schwarzian action.

The next lemma follows from the chain rule of the Schwarzian and expresses the Schwarzian action of a diffeomorphism in terms of its inverse, which will be used later.

\begin{lem}\label{lem:Schwarzian_inverse}
    Let $\varphi:\m S^1 \rightarrow\m S^1$ be a $C^3$ diffeomorphism with inverse $\phi=\varphi^{-1}$. Then
    \begin{align}
        \Isch(\varphi) = \int_{\m S^1} 
|\phi'(\zeta)|^{-1}\,\ii  \zeta \,\mc S[\phi](\zeta)\, \dd \zeta = -\int_0^{2\pi} |\phi'(\ee^{\ii\theta})|^{-1} \ee^{2\ii \t} \mc S [\phi] (\ee^{\ii \t}) \,\dd \t.
    \end{align}
\end{lem}

\subsection{Loewner energy and welding homeomorphisms} \label{subsec:Loewner}
We recall the definition of Loewner energy.
Let $\g$ be an oriented Jordan curve, $f$ (resp., $g$) be a conformal map $\m D \to \O$ (resp., $\m D^* \to \O^*$), where $\O$ (resp., $\O^*$) is the connected component of $\Chat \smallsetminus \g$ on the left of $\g$ (resp., on the right of $\g$). The Carath\'eodory theorem implies that $f$ and $g$ extend to a homeomorphism on the closure of $\m D$ and $\m D^*$. We call the circle homomorphism $\varphi:= g^{-1} \circ f|_{\m S^1}$ a \emph{welding} of $\g$.

Characterizing all the circle homeomorphisms that arise as welding is open and difficult \cite{Bishop_CR_hard,rodriguez2025}. However, if we restrict to the group of quasisymmetric circle homeomorphisms $\QS(\m S^1)$, then we obtain a one-to-one correspondence:
\begin{align*}
    \PSL(2,\m C) \backslash \{\text{Oriented quasicircles in } \Chat\} & \xleftrightarrow[1:1]{\text{conformal welding}} \PSU(1,1)\backslash \QS (\m S^1)/\PSU(1,1),\\
    [\g] & \xleftrightarrow[]{\hspace{60pt}}  [\weld  = g^{-1} \circ f|_{\m S^1}].
\end{align*}

Considering the equivalence classes above is natural. On the right-hand side, one may pre-compose $f$ and $g$ by M\"obius transformations preserving $\m S^1$, \ie $\PSU(1,1)$, this does not change the curve $\g$ and gives an equivalent welding. On the left-hand side, replacing $(\g; f,g)$ by $(\a \circ \g; \a \circ f, \a \circ g)$, where $\a \in \PSL(2, \m C)$ is a conformal automorphism of $\Chat$, shows that $\g$ and $\a \circ \g$ share the same welding.

 The \emph{Loewner energy} of a Jordan curve $\g$ is first defined as the Dirichlet energy of its Loewner driving function in \cite{RW}. It was shown to be $\PSL(2,\m C)$-invariant and independent of its orientation. The second author \cite{W2} also showed that it coincides with the universal Liouville action in \cite{TT06}, which can be computed by choosing the orientation of $\g$ such that $\O$ does not contain $\infty$, and $g : \m D^* \to \O^*$ fixing $\infty$:
 \begin{equation}\label{eq:def_IL}
     I^L(\g) = \frac{1}{\pi} \left(\int_{\m D} \abs{\frac{f''}{f'}}^2 |\dd z|^2 + \int_{\m D^*} \abs{\frac{g''}{g'}}^2 |\dd z|^2 \right) + 4 \log \abs{\frac{f'(0)}{g'(\infty)}} \in [0,\infty],
 \end{equation}
 where $|\dd z|^2$ denotes the Euclidean area measure. 
 Since $I^L(\g)$ is $\PSL(2,\m C)$-invariant, it defines a bi-invariant function on the space of weldings.
 \begin{thm}[See \cite{shen13}]
    A Jordan curve $\g$ satisfies $I^L (\g) < \infty$ if and only if any welding of $\g$ belongs to $\WP(\m S^1)$, 
    where 
    $\WP(\m S^1)$ is the space of circle homeomorphisms $\varphi$ that are absolutely continuous and
    $\log |\varphi'|$ is in the Sobolev $H^{1/2}(\m S^1)$ space.
 \end{thm}
 From this, we have the inclusions:
 \begin{lem}
      For all $\vare > 0$, 
 $$C^{1.5+\vare} (\m S^1) \subset \WP(\m S^1) \subset \QS(\m S^1).$$
 In particular, the Loewner energy is well-defined, finite, and bi-invariant in the space of $C^3$ circle homeomorphisms.
 \end{lem}

\subsection{Epstein map in any dimension} \label{subsec:epstein_general}
We describe briefly the general construction of Epstein maps introduced in \cite{epstein-envelopes}. This construction will be used to define renormalized area and renormalized length.

Let $ n \ge 1$. We use the ball model of the hyperbolic $n+1$ space
$$\m H^{n+1} = \left\{x \in \m R^{n+1} \, |\, \norm{x} < 1\right\},$$ whose conformal boundary is identified with the Euclidean unit ball $\m S^n$.
Let $h$ be a smooth Riemannian metric on $\m S^n$ conformal to the round metric (induced by the Euclidean metric of $\m R^{n+1}$). For $x \in  \m H^{n+1}$ in the ball model, we let $\nu_x$ be the hyperbolic visual metric on the unit sphere  $\m S^n = \partial \m H^{n+1}$ seen from $x$ (namely, the pull-back of the round metric on $\m S^n$ by any orientation-preserving isometry of $\m H^{n+1}$ sending $x$ to the origin $0$), then for $z\in \m S^n$, we define
\begin{equation}\label{eq:def_horo}
    H_z = H_z (h) = \{ x\in \m H^{n+1} \,|\, \nu_x(z) = h(z)\}.
\end{equation}
It is straightforward to check:
\begin{lem}\label{lem:horo_t}
    The set $H_z (h)$ is a horosphere centered at $z$ and depends only on the metric tensor $h (z)$. Let $t \in \m R$, $H_z(\ee^{2t} h)$ is the horosphere centered at $z$ at signed distance $t$ from $H_z (h)$. The sign is such that if $t > 0$, $H_z(\ee^{2t} h)$ is contained in the horoball bounded by $H_z (h)$.
\end{lem}
The \emph{Epstein map} $\Ep_h$ is the solution to the envelope equation of these horospheres. More precisely, there exists a (unique) continuous map 
\begin{equation}\label{eq:tilde_Ep}
    \widetilde{\Ep}_h\colon \m S^n \to T^1\m H^{n+1},
\end{equation}
where $T^1\m H^{n+1}$ denotes the unit tangent bundle of $\m H^{n+1}$,
 such that $\widetilde{\Ep}_h(z)$ is an outward pointing normal to the horoball bounded by $H_z$, 
$$\Ep_h\colon \m S^n \to \m H^{n+1}$$
is the composition of $\widetilde\Ep_h$ with the projection $ T^1\m H^{n+1} \to \m H^{n+1}$, and the image of the differential of $\Ep_h$ at $z$ is orthogonal to $\widetilde\Ep_h(z)$. We call the image of $\Ep_h$ the \emph{Epstein hypersurface associated with $h$} and denote it by $\S_h$. (We sometimes also do not distinguish between $\Ep_h$ and $\S_h$.)
In particular, $\Ep_h(z)$ is a tangent point between the horosphere $H_z$ and the Epstein hypersurface.

\begin{ex}\label{ex:round}
    If $h$ is the round metric on $\m S^n$, then for all $z \in \m S^n$, $H_z$ is the horosphere passing through $0$ and centered at $z$. The function $\widetilde{\Ep}_h$ maps $z \in \m S^n$ to $-z \in T^1_0 \m H^{n+1}$ and the Epstein hypersurface degenerates to the origin. 
\end{ex}

We record some basic facts about Epstein maps.

\begin{lem}[Naturality of Epstein map]\label{lem:naturality}
    If $x\in\mathbb{H}^{n+1}$, $\a \in \Isom_+(\mathbb{H}^{n+1})$ then $\nu_{\a(x)} = \a_*(\nu_x)$, where $\a_*$ is the pushforward by $\a$. 
It follows that for $z \in \m S^{n}$,
\begin{equation}\label{eq:Eps_invariant}
\a (H_z (h)) = H_{\a(z)} (\a_* h) \quad \text{ and } \quad  \a \circ \Ep_{h} (z) = \Ep_{\a_* h} (\a (z)).
\end{equation}
\end{lem}

\begin{thm}[Scaling property of Epstein map \cite{KrasnovSchlenker_CMP,BBB}] \label{thm:basic_epstein}
Let $h$ be a $C^2$ conformal metric on $\m S^n$. Let $h_t := \ee^{2t} h$ for $t \in \m R$.
\begin{enumerate}
    \item The value of $\widetilde{\Ep}_h(z)$ is determined by $h$ and its first derivatives at $z$.

\item Let $\mathfrak g_t: T^1\m H^{n+1} \rightarrow T^1\m H^{n+1}$ be time $t$ geodesic flow.  Then $\widetilde\Ep_{h_t} = \mathfrak{g}_{-t} \circ \widetilde\Ep_h$. 
    \item Let $\mathfrak g_{-\infty}: T^1 \m H^{n+1} \rightarrow \m S^n$ be the {\em hyperbolic Gauss map} sending a tangent vector to the endpoint of the associated geodesic ray as $t \to -\infty$. Then $\mathfrak g_{-\infty} \left(\widetilde \Ep_h (z) \right) = z$.

    \item For all $z \in \m S^n$, there are at most two values of $t$ where $\Ep_{h_t}$ is not an immersion at $z$.
    \item There exists $t_0$ such that for all $t \ge t_0$, $\Ep_{h_t}$ is an embedding. Let $\operatorname{I}_t$ be the pullback metric by $\Ep_{h_t}  \colon \m S^n \to \S_{h_t}$, then $4 \ee^{-2t} \operatorname{I}_t \to h$ as $t \to \infty$.
\end{enumerate}
\end{thm}

\begin{remark}\label{rem:equidistant_foliation}
    In particular, the conformal metric $h$ defines an equidistant foliation  $(\S_{h_t})_{t \ge t_0}$ in a neighborhood of $\m S^n$.
\end{remark}

The Epstein map is defined locally, therefore, we may extend its definition to a conformal metric defined on an open domain $\O \subset \m S^n$. In Equation~\eqref{eq:VR_IL}, $n = 2$ and $N_\g$ is the domain bounded by two Epstein surfaces $\S_{\rho_\O}$ and $\S_{\rho_{\O^*}}$, where $\rho_\O$ and $\rho_{\O^*}$ are respectively the hyperbolic metric on $\O$ and $\O^*$.

\section{Epstein curves}

\subsection{Epstein map in one dimension} \label{subsec:eps_1}

In this section, we consider Epstein maps when $n=1$. 
With slight abuse of notation, we use the length element of the metric in this section.  
Let $h = \ee^{\sigma(\theta)} \dd\theta$ be a $C^2$ metric on $\m S^1$. We use $\s_\t$ to denote the derivative of $\s$ with respect to $\t$ (and only for $\theta$).

The \emph{Epstein curve} associated with $h$ is the curve (envelope) in $\m D =  \m H^2$ tangent to the
$1$-parameter family of horocycles $$H_{\ee^{\ii\theta}} = \left\{ \frac{\ee^{\sigma(\theta)}}{\ee^{\sigma(\theta)}+1} \ee^{\ii\theta} + \frac{1}{\ee^{\sigma(\theta)}+1} Y\,\Big|\, Y\in \m S^1  \right\}$$ 
which follows from \eqref{eq:def_horo}.
Explicitly, the formula of the envelope derived in \cite{epstein-envelopes} showed that it is described by the parametrized curve $\Ep_h: \m S^1 \rightarrow \mathbb{D}$
\begin{align}\label{eq:Epstein_point}
\Ep_h(\ee^{\ii\theta}) = \frac{\sigma_{\theta}^2 + (\ee^{2\sigma}-1)}{\sigma_{\theta}^2+(\ee^{\sigma}+1)^2}\ee^{\ii\theta} + \frac{2\sigma_{\theta}}{\sigma_{\theta}^2+(\ee^{\sigma}+1)^2} \ii \ee^{\ii\theta}
\end{align}
with normal given by 
$\widetilde{\Ep}_h: \mathbb{S}^1 \rightarrow T^1\mathbb{D}$
\begin{align}\label{eq:Epsteinnormal}
    \widetilde{\Ep}_h(\ee^{\ii\theta}) = \frac{2\ee^{\sigma}\left(\sigma_{\theta}^2 - (\ee^{\sigma}+1)^2\right)}{\left(\sigma_{\theta}^2+(\ee^{\sigma}+1)^2\right)^2}\ee^{\ii\theta} + \frac{4\sigma_{\theta}\ee^{\sigma}(\ee^\sigma+1)}{\left(\sigma_{\theta}^2+(\ee^{\sigma}+1)^2\right)^2} \ii\ee^{\ii\theta}.
\end{align}

For a constant rescaling of the metric $h_t = \ee^th = \ee^t \ee^\sigma\dd \t$ ($t\in\mathbb{R}$), we check that $\widetilde{\Ep}_{h_t}(\ee^{\ii\theta})$ is obtained by flowing the vector $\widetilde{\Ep}_{h}(\ee^{\ii\theta})$ distance $-t$ along the geodesic flow, as stated in Theorem~\ref{thm:basic_epstein}. See Appendix~\ref{sec:examples} for a few examples.

\begin{ex}\label{ex:Ep_point}
    For $h=\dd \theta$, the Epstein curve $\Ep_h$ is the constant curve at the origin $(0,0)\in\mathbb{D}$, while $\widetilde\Ep_h$ is an arc-length parametrization of the unit vectors at $(0,0)$. For $h_t=\ee^th$, the Epstein curve $\Ep_{h_t}$ is a parametrization of the circle $\{z \in \m D \colon |z| = \tanh(|t|/2)\}$ of hyperbolic radius $|t|$ around $(0,0)\in\mathbb{D}$. 
    Let $T (\ee^{\ii \t})$ be the tangent vector $D_\t \Ep_h$, then when $t > 0$, $(T,\widetilde{\Ep}_h)$ is positively oriented, and when $t < 0$, $(T,\widetilde{\Ep}_h)$ is negatively oriented. 
\end{ex}

\begin{ex} \label{ex:horocycle}
    For $h = 2 \,\dd \theta / (1-\cos\theta)=\dd \theta / \sin^2 (\t/2)$, the Epstein curve $\Ep_h$ is the horocycle $H\in\mathbb{D}$ tangent at $1$ and passing through $(0,0)$, and $\widetilde{\Ep}_h$ parametrizes the inner unit normal vectors of the horocycle $H$. For $h_t=\ee^th$, the Epstein curve $\Ep_{h_t}$ is the horocycle at signed distance $t$ from $H$.
\end{ex}

\begin{remark}
    While we refer to $\Ep_h$ as an Epstein \emph{curve}, it is truly a parametrized path as in Theorem~\ref{thm:basic_epstein}. This path might fail to be immersed (Example~\ref{ex:Ep_point} as the extreme scenario), but it has a continuous choice of unit normal given by $\widetilde{\Ep}_h$. We will use this choice of normal to appropriately extend geometric quantities to Epstein curves.
\end{remark}

\begin{definition}\label{def:length_curv}
    Given $h$ a $C^2$ metric on $\m S^1$ with Epstein curve $\Ep_h$, we define the \emph{signed arc-length} $\dd\ell$ as the standard arc-length $1$-form of $\S_h$ 
    times the sign of the orientation of the pair $(T=D_\t \Ep_h,\widetilde{\Ep}_h)$ \textnormal(taken to be $0$ if $T$ vanishes\textnormal). At points where $\Ep_h$ is an immersion, we define its \emph{geodesic curvature} $k$ as the geodesic curvature of the curve $\S_h$ parametrized by $\Ep_h$ with respect to $\widetilde{\Ep}_h$.
\end{definition}

The following lemma is a useful way to compute the arc-length and geodesic curvature of the Epstein curve in terms of the conformal factor at infinity.

\begin{lem}\label{lem:length_curv_local}
    For a metric $h=\ee^\sigma \dd \theta$ in $\m S^1$, the signed arc-length $\dd \ell$ and geodesic curvature $k$ of the Epstein curve $\Ep_h$ can be expressed as
    \begin{align}
        \dd\ell &= \left(\ee^{-\sigma}\left( \sigma_{\theta\theta}-\frac{\sigma_{\theta}^2}{2}\right)+\sinh(\sigma)\right) \dd\theta,\\\nonumber
        k &= \frac{1-\ee^{-2\sigma}(2\sigma_{\theta\theta}-\sigma_{\theta}^2-1)}{1+\ee^{-2\sigma}(2\sigma_{\theta\theta}-\sigma_{\theta}^2-1)}.
    \end{align}
\end{lem}
\begin{proof}
Observe that the tangent vector $T (\ee^{\ii \t})= \partial_\t \Ep_h(\ee^{\ii \t})$ to the Epstein curve is expressed as
\begin{align*}
    T(\ee^{\ii\theta}) 
    &= 
    \frac{2\sigma_{\theta}(\ee^\sigma+1)(2\sigma_{\theta\theta}-\sigma_{\theta}^2+\ee^{2\sigma}-1)}{(\sigma_{\theta}^2+(\ee^{\sigma}+1)^2)^2}
    \ee^{\ii\theta} 
    +
    \frac{((\ee^\sigma+1)^2-\sigma_{\theta}^2)(2\sigma_{\theta\theta}-\sigma_{\theta}^2 +\ee^{2\sigma}-1)}{(\sigma_{\theta}^2+(\ee^{\sigma}+1)^2)^2}
    \,\ii\ee^{\ii\theta}\\
    &=-\ii\widetilde{\Ep}_h(\ee^{\ii\theta})\left(\ee^{-\sigma}\left( \sigma_{\theta\theta}-\frac{\sigma_{\theta}^2}{2}\right)+\sinh(\sigma)\right).
\end{align*}
It follows that the signed arc-length $\,\dd \ell$ (according to the convention in Definition~\ref{def:length_curv}) of the Epstein curve is given by
\begin{align}
    \dd\ell = \left(\ee^{-\sigma}\left( \sigma_{\theta\theta}-\frac{\sigma_{\theta}^2}{2}\right)+\sinh(\sigma)\right) \dd\theta.
\end{align}

From the expression \eqref{eq:Epsteinnormal} for the normal vector to the Epstein curve, we can explicitly calculate
\begin{align}
    \langle-\nabla_T \widetilde{\Ep}_h, T \rangle = \left(\ee^{-\sigma}\left( \sigma_{\theta\theta}-\frac{\sigma_{\theta}^2}{2}\right)+\sinh(\sigma)\right) \left(\ee^{-\sigma}\left( -\sigma_{\theta\theta}+\frac{\sigma_{\theta}^2}{2}\right)+\cosh(\sigma)\right),
\end{align}
where $\nabla, \langle.,.\rangle$ denote the connection and metric of the hyperbolic plane. Hence it follows that the geodesic curvature $k(x)$ of the Epstein curve at $\Ep_h(x)$ is given by
\begin{align}
    k = \frac{\left(\ee^{-\sigma}\left( -\sigma_{\theta\theta}+\sigma_{\theta}^2/2\right)+\cosh(\sigma)\right)}{\left(\ee^{-\sigma}\left( \sigma_{\theta\theta}-\sigma_{\theta}^2/2\right)+\sinh(\sigma)\right)} = \frac{1-\ee^{-2\sigma}(2\sigma_{\theta\theta}-\sigma_{\theta}^2-1)}{1+\ee^{-2\sigma}(2\sigma_{\theta\theta}-\sigma_{\theta}^2-1)}
\end{align}
as claimed.
\end{proof}

\begin{remark}\label{rem:non-immersed}
    The previous lemma justifies Definition~\ref{def:length_curv} in the sense that under our choice, the signed arc-length vanishes but does not lose regularity at points where $\Ep_h$ is not immersed. Similarly, while $k$ blows up at points where $\Ep_h$ is not immersed, the form $k\,\dd\ell$ can be continuously and naturally extended to these points as 
    \begin{align*}
    k\,\dd\ell : = \left(\ee^{-\sigma}\left( -\sigma_{\theta\theta}+\frac{\sigma_{\theta}^2}{2}\right)+\cosh(\sigma)\right) \dd\theta.
    \end{align*}
\end{remark}

\begin{ex}\label{ex:horocycle_2}
   We check that for the horocycle in Example~\ref{ex:horocycle}, we have $k \equiv 1$.
\end{ex}

\begin{definition}\label{def:total_length_curv}
    Given $h=\ee^\sigma \dd\theta$ a $C^2$ metric on $\m S^1$ with Epstein curve $\Ep_h$, we define its total length $L(\Ep_h)$ as
    \begin{align}
        L(\Ep_h) : = \int_{\m S^1} \dd\ell = \int_{\m S^1} \left(\ee^{-\sigma}\left( \sigma_{\theta\theta}-\frac{\sigma_{\theta}^2}{2}\right)+\sinh(\sigma)\right) \dd\theta.
    \end{align}
    Similarly, we define its total geodesic curvature as
    \begin{align*}
        \int_{\m S^1} k\,\dd\ell : = \int_{\m S^1} \left(\ee^{-\sigma}\left( -\sigma_{\theta\theta}+\frac{\sigma_{\theta}^2}{2}\right)+\cosh(\sigma)\right) \dd\theta.
    \end{align*}
\end{definition}

We now extend the notion of enclosed area for Epstein curves.
\begin{definition}\label{def:signed_A}
    Let $h=\ee^\sigma \dd\theta$ be a $C^2$ metric on $\m S^1$. We define its \emph{signed area} $A(\Ep_h)$ as
    \[
    A(\Ep_h) : = \int_{\overline{\m D}} \Phi^*(\dd a_{\rm hyp}) 
    \]
    where $\Phi:\overline{\m D}\rightarrow  \m D$ is any piecewise $C^1$ extension of $\Ep_h$ and $\dd a_{\rm hyp}$ denotes the area form of the hyperbolic disk.
\end{definition}
Observe that since for any two piecewise $C^1$ extensions $\Phi_0$, $\Phi_1$ of $\Ep_h$ there is a piecewise $C^1$ homotopy $\lbrace \Phi_t \rbrace_{0\leq t\leq 1}$ satisfying $\Phi_t(\ee^{\ii\theta}) = \Phi_0(\ee^{\ii\theta}) = \Phi_1(\ee^{\ii\theta})$, it follows by Stokes theorem that
\begin{align*}
    \int_{\overline{\m D}} \Phi_1^*(\dd a_{\rm hyp}) -  \int_{\overline{\m D}} \Phi_0^*(\dd a_{\rm hyp})=0.
\end{align*}
Hence $A(\Ep_h)$ is well-defined.

The following Lemma shows that the Gauss--Bonnet theorem extends to signed area $A(\Ep_h)$ and signed length $\dd\ell$.

\begin{lem}[Gauss--Bonnet for signed length and area] \label{lem:Gauss_Bonnet}
    Let $\hat{h}_0 = \ee^{\sigma_0}\dd\theta$, $\hat{h}_1 = \ee^{\sigma_1}\dd\theta$ be any $C^2$ metrics in $\m S^1$. Then
    \begin{align}\label{eq:AreavsCurv}
        A(\Ep_{\hat{h}_1}) - A(\Ep_{\hat{h}_0}) = \int_{\m S^1} \hat{k}_1\,\dd\hat{\ell}_1 - \int_{\m S^1} \hat{k}_0\,\dd\hat{\ell}_0,
    \end{align}
    where $\hat{k}_i, \dd\hat{\ell}_i$ are the curvature and signed arc-length of $\Ep_{\hat{h}_i}$ $(i=0,1)$. 
    In particular, it follows that Epstein curves satisfy the Gauss--Bonnet theorem. Namely, for any $C^2$ metric $h$ in $S^1$ we have
    \begin{align}\label{eq:Epstein_Gauss-Bonnet}
        A(\Ep_h) = \int_{\m S^1} k\,\dd\ell -2\pi.
    \end{align}
\end{lem}
\begin{proof}
    Let $\sigma=\sigma_1-\sigma_0$ and $\sigma_t=\sigma_0+t\sigma$. Consider the $1$-parameter family of metrics $\hat{h}_t = \ee^{\sigma_t} \dd \theta$. Their respective Epstein curves define a parametrized cylinder between $\hat{h}_0, \hat{h}_1$ by $\Psi:\m S^1\times [0,1]\rightarrow \m D$ given as
    \begin{align*}
\Psi(\ee^{\ii\theta}, t) : = \Ep_{\hat{h}_t}(\ee^{\ii\theta}). 
    \end{align*}
    As we can get an extension of $\Ep_{\hat{h}_1}$ by doing a domain reparametrization of a concatenating between $\Psi$ and an extension of $\Ep_{\hat{h}_0}$, it follows that
    \begin{align*}
        A(\Ep_{\hat{h}_1}) - A(\Ep_{\hat{h}_0}) = \int_{\m S^1\times [0,1]} \Psi^*(\dd a_{\rm hyp}).
    \end{align*}
    Denote by $\hat{k}_t, \dd\hat{\ell}_t$ the curvature and signed arc-length of $\Ep_{\hat{h}_t}$. Observe that
    \begin{align*}
        \Psi^*(\dd a_{\rm hyp}) = -\sigma\,\dd \hat{\ell}_t\wedge \dd t = \sigma \left(\ee^{-\sigma_t}\left( (\sigma_t)_{\theta\theta}-\frac{((\sigma_t)_{\theta})^2}{2}\right)+\sinh(\sigma_t)\right) \dd t\wedge\dd \theta
    \end{align*}
    which can be directly computed from \eqref{eq:Epstein_point} and Lemma~\ref{lem:length_curv_local}, or seen by observing that the component of $D_t\Psi$ in $\widetilde{\Ep}_h$ is given by the rate of change of distance between horospheres with respect to $\widetilde{\Ep}_h$, namely $-\sigma$.

    In $\m S^1\times [0,1]$, define the $1$-form
    \begin{align*}
        \eta : &= \hat{k}_t\dd\theta - \sigma_\theta \ee^{-\sigma_t} \dd t \\\nonumber&=\left(\ee^{-\sigma_t}\left( -(\sigma_t)_{\theta\theta}+\frac{((\sigma_t)_{\theta})^2}{2}\right)+\cosh(\sigma_t)\right) \dd\theta - \sigma_\theta \ee^{-\sigma_t} \dd t.
    \end{align*}
    By direct computation, we find 
    \begin{align*}
        \dd\eta = \sigma \left(\ee^{-\sigma_t}\left( (\sigma_t)_{\theta\theta}-\frac{((\sigma_t)_{\theta})^2}{2}\right)+\sinh(\sigma_t)\right) \dd t\wedge\dd \theta  = \Psi^*(\dd a_{\rm hyp}).
    \end{align*}
    Applying Stokes theorem and Lemma~\ref{lem:length_curv_local} we obtain
    \begin{align*}
        \int_{\m S^1\times [0,1]} \Psi^*(\dd a_{\rm hyp}) &= \int_{\m S^1\times \lbrace 1\rbrace} \eta - \int_{\m S^1 \times \lbrace 0\rbrace} \eta\\\nonumber
        &= \int_{\m S^1}\left(\ee^{-\sigma_1}\left( (\sigma_1)_{\theta\theta}-\frac{((\sigma_1)_{\theta})^2}{2}\right)+\sinh(\sigma_1)\right) \dd\theta \\\nonumber & \quad-  \int_{\m S^1}\left(\ee^{-\sigma_0}\left( (\sigma_0)_{\theta\theta}-\frac{((\sigma_0)_{\theta})^2}{2}\right)+\sinh(\sigma_0)\right) \dd\theta\\\nonumber
        &= \int_{\m S^1} \hat{k}_1\dd\hat{\ell}_1 - \int_{\m S^1} \hat{k}_0\dd\hat{\ell}_0,
    \end{align*}
    so \eqref{eq:AreavsCurv} follows.

    In order to conclude \eqref{eq:Epstein_Gauss-Bonnet}, observe that by Example~\ref{ex:Ep_point} this equation follows for $\tilde{h}=\ee^t\dd\theta$ and $t>0$, as it can be computed directly while also being equivalent to the standard Gauss--Bonnet theorem for the associated Epstein curve. Now, for any given $h$, we can apply \eqref{eq:AreavsCurv} for $\hat h_1=h$ and $\hat h_0=\tilde{h}$. By adding this equation to \eqref{eq:Epstein_Gauss-Bonnet} for $\hat h_0$ we obtain \eqref{eq:Epstein_Gauss-Bonnet} for $h$.
\end{proof}

\begin{definition}\label{def:curvature_infty}
    Let $h=\ee^\sigma \dd\theta$ be a $C^2$ metric on $\m S^1$. The term \[k^*=\ee^{-2\sigma}(2\sigma_{\theta\theta}-\sigma_{\theta}^2-1)\] is called the \emph{curvature at infinity} of $h$.
\end{definition}

By Lemma~\ref{lem:length_curv_local}, we have
\begin{align}\label{eq:k_k*}
    k=\frac{1-k^*}{1+k^*}.
\end{align}
This also characterizes the points where $\Ep_h$ fails to be an immersion as the points where $k^*=-1$. It follows from the formulae that if for constant $t$ we denote by $k_t$, $k^*_t$ the curvature of the Epstein curve and the curvature at infinity (respectively) of the metric $h_t=\ee^t h$, we can easily verify
\begin{align*}
    k^*_t&=\ee^{-2t}k^*\\\nonumber
    k_t&=\frac{1-k^*_t}{1+k^*_t} = \frac{1-\ee^{-2t}k^*}{1+\ee^{-2t}k^*}
\end{align*}
and subsequently see that $\lim_{t\rightarrow+\infty} \ee^{2t}(1-k_t)/2 = k^*$. This justifies the term curvature at infinity for $k^*$, as it appears as an asymptotic limit of the curvature of the equidistant curves.

\begin{ex}
    For the standard metric $\dd \theta$ in $\m S^1$, $\sigma$ is identically $0$. Hence, $\dd \ell\equiv0$ and $k^*\equiv-1$, while the blow-up of $k$ corresponds to the Epstein curve collapsing to a point. It clearly has signed area $A(\Ep_h)=0$. Finally, $k_t=\coth(t)$, which is the geodesic curvature of the geodesic ball of radius $t$ with respect to the choice of normal given by $\widetilde{\Ep}_{h_t}$.
\end{ex}

\begin{remark}\label{rem:fol_circle}
    Given $h=\ee^\sigma\dd\theta$ a $C^2$ metric on $\m S^1$ and $h_t=\ee^t h$, it follows by \eqref{eq:Epstein_point} that as $t\rightarrow +\infty$ the Epstein curves $\Ep_{h_t}$ converge in $C^1$ to the identity map of $\m S^1$. Hence, for $t$ sufficiently large, $\Ep_{h_t}$ is a parametrized Jordan curve with $\widetilde{\Ep}_{h_t}$ given by the inner unit normal. Hence, $\dd\ell$, $k$, and $A(\cdot)$ are equal to the standard geometric concepts.
\end{remark}

We collect a few useful expressions for the arc-length and curvature of an equidistant Epstein foliation, both in terms of data at infinity for the initial Epstein curve.

\begin{lem}\label{lem:asymptoticforms}
    Let $h=\ee^\sigma \dd\theta$ be a metric on $\m S^1$, $k^*$ its curvature at infinity, 
    $k$ the geodesic curvature of the Epstein curve. For $t\in\mathbb{R}$, let $h_t=\ee^th$. Denote by $\dd\ell_t$ and $k_t$ the signed arc-length and geodesic curvature of the Epstein curve $\Ep_{h_t}$.
    The following identities hold:
    \begin{enumerate}
        \item $\dd\ell_t = \frac12 (\ee^t+\ee^{-t}k^*)h$
        \item $k_t\,\dd \ell_t = \frac12 (\ee^t-\ee^{-t}k^*)h$
        \item $\dd\ell_t = \cosh(t)\,\dd\ell + \sinh(t) k\,\dd\ell$
        \item $k_t\,\dd \ell_t = \sinh(t)\,\dd\ell + \cosh(t) k\,\dd\ell$.
    \end{enumerate}
\end{lem}
\begin{proof}
Let us start by verifying the first two items for $\dd \ell$ and $k\,\dd \ell$. Recall that from the definition of $k^*$, Lemma~\ref{lem:length_curv_local},   and the equality $k=(1-k^*)/(1+k^*)$,
we have
\begin{align}\label{eq:dl_k*}
    \dd\ell &= \frac12\left(1+k^*\right)\ee^{\sigma} \dd\t = \frac12\left(1+k^*\right)h, \\
    k\,\dd \ell &= \frac12\left(\frac{1-k^*}{1+k^*}\right) (1+k^*)h = \frac12(1-k^*)h.
    \label{eq:kdl_k*}
\end{align}
Applying this for $h_t=\ee^th$, and recalling that its curvature at infinity is $k^*_t = \ee^{-2t}k^*$, we now obtain the first two items
\begin{align}\label{eq:equidistant_length_curv}
    \dd\ell_t&= \frac12\left(1+k^*_t\right)h_t = \frac12(\ee^t+\ee^{-t}k^*)h, \\\nonumber
    k_t\,\dd \ell_t &= \frac12 \left(1-k^*_t\right)h_t = \frac12 (\ee^t - \ee^{-t}k^*) h.
\end{align}
In particular, the total signed length and total curvature of $\Ep_{h_t}$ can be expressed as a function of $t$ and the metric $h$ by
\begin{align}\label{eq:equidistanttotal}
    \int_{\m S^1} \dd\ell_t &= \frac{\ee^t}{2} \int_{\m S^1} h + \frac{\ee^{-t}}{2} \int_{\m S^1}  k^*h,\\\nonumber
    \int_{\m S^1} k_t\,\dd \ell_t &= \frac{\ee^t}{2} \int_{\m S^1}  h - \frac{\ee^{-t}}{2} \int_{\m S^1}  k^*h.
\end{align}
To express $\dd\ell_t$ and $k_t\,\dd \ell_t$ as integrals in the Epstein curve it will be useful to have the following inversion formulas, which follow easily from the previous statements
\begin{align}\label{eq:inversionformulas}
    k^*&=\frac{1-k}{1+k}, \\\nonumber
    h&= \left(\frac{2}{1+k^*}\right)\,\dd \ell =(1+k)\,\dd \ell.
\end{align}
Combining then \eqref{eq:equidistant_length_curv} and \eqref{eq:inversionformulas} we conclude the last two items as follows
\begin{align}\label{eq:length_t_from0}
    \dd\ell_t = \frac{\ee^t}{2}(1+k)\,\dd \ell + \frac{\ee^{-t}}{2} (1-k)\,\dd \ell = \cosh(t) \,\dd \ell + \sinh(t) k\,\dd \ell
\end{align}
and analogously
\begin{align}\label{eq:totalcurvature_t_from0}
    k_t\,\dd \ell_t = \frac{\ee^t}{2} (1+k)\,\dd \ell - \frac{\ee^{-t}}{2} (1-k)\,\dd \ell =  \sinh(t) \,\dd \ell + \cosh(t) k\,\dd \ell
\end{align}
which completes the proof.
\end{proof}

\subsection{Epstein curves from circle diffeomorphisms}\label{subsec:ep_diffeo}

The previous results establish the basic properties of the Epstein curve for any $C^2$ metric on the conformal boundary at infinity $\m S^1 =\partial_\infty \m H^2$, we now apply the results to metrics obtained by pushing forward $\dd \t$ by a $C^3$ diffeomorphism.

\begin{thm}\label{thm:SchArea}
        Let $\varphi : \m S^1 \to \m S^1$ be a $C^3$ diffeomorphism. Let $h = \varphi_* \dd \t$ be the pushforward of $\dd \t$ by $\varphi$. Let $A(\Ep_h)$ denote the signed hyperbolic area enclosed by the Epstein curve associated with  $h$ and $L(\Ep_h)$ the signed hyperbolic length of the Epstein curve, then we have
    \begin{align*}
        \Isch(\varphi) = L(\Ep_h) =  -A(\Ep_h).
    \end{align*}
\end{thm}

 \begin{proof}
Let $\varphi:\m S^1\rightarrow \m S^1$ be a $C^3$ diffeomorphism with inverse $\phi = \varphi^{-1}$ and 
$h = \varphi_* \dd \t =|\phi'(\ee^{\ii\theta})|\dd \theta$ the corresponding $C^2$ metric on $\m S^1$ of length $2\pi$. Writing $\s = \log |\phi'(e^{\ii\t})|$ so that $h=\ee^\sigma \dd \theta$, we have
\begin{align}\label{eq:sigma_jet}
    \sigma_{\theta}&=\Re\left( \frac{\ii\ee^{\ii\theta}\phi''(\ee^{\ii\theta})}{\phi'(\ee^{\ii\theta})}\right),\\
    \nonumber
    \sigma_{\theta\theta}&= \Re \left(-\ee^{2\ii\theta}\left( \frac{\phi'''(\ee^{\ii\theta})}{\phi'(\ee^{\ii\theta})} - \left(\frac{\phi''(\ee^{\ii\theta})}{\phi'(\ee^{\ii\theta})}\right)^2 \right) - \frac{\phi''(\ee^{\ii\theta})\ee^{\ii\theta}}{\phi'(\ee^{\ii\theta})} \right).
\end{align}
Recalling that $\phi$ is a diffeomorphism of $\m S^1$, one can verify that
\begin{align}\label{eq:phi''_identity}
    \Re\left( -\ee^{2\ii\theta} \left( \frac{\phi''(\ee^{\ii\theta})}{\phi'(\ee^{\ii\theta})} \right)^2\right) = \Bigg\vert\frac{\phi''(\ee^{\ii\theta})}{\phi'(\ee^{\ii\theta})}\Bigg\vert^2 + 4\Re\left(\frac{\phi''(\ee^{\ii\theta})\ee^{\ii\theta}}{\phi'(\ee^{\ii\theta})}\right) +2 -2 |\phi'(\ee^{\ii\theta})|^2.
\end{align}
We can then express $\sigma_{\theta\theta}-\sigma_{\theta}^2/2$ in terms of the first three derivatives of $\phi$ using the relation $(\Re x)^2 = (\Re (x^2) + |x|^2)/2$, \eqref{eq:sigma_jet}, and \eqref{eq:phi''_identity} to obtain
\begin{align*}
    \sigma_{\theta\theta}-\frac12\sigma_{\theta}^2&=  -\Re\left(\ee^{2\ii\theta}\left( \frac{\phi'''(\ee^{\ii\theta})}{\phi'(\ee^{\ii\theta})} - \left(\frac{\phi''(\ee^{\ii\theta})}{\phi'(\ee^{\ii\theta})}\right)^2 \right)\right) - \Re\left(\frac{\phi''(\ee^{\ii\theta})\ee^{\ii\theta}}{\phi'(\ee^{\ii\theta})}\right) \\\nonumber &\quad -\frac14\Re\left( -\ee^{2\ii\theta} \left( \frac{\phi''(\ee^{\ii\theta})}{\phi'(\ee^{\ii\theta})}\right)^2\right) -\frac14\Bigg\vert\frac{\phi''(\ee^{\ii\theta})}{\phi'(\ee^{\ii\theta})}\Bigg\vert^2\\\nonumber
    &= -\Re \left(\ee^{2\ii\theta}\left( \frac{\phi'''(\ee^{\ii\theta})}{\phi'(\ee^{\ii\theta})} - \frac32\left(\frac{\phi''(\ee^{\ii\theta})}{\phi'(\ee^{\ii\theta})}\right)^2 \right)\right) +\frac{1-|\phi'(\ee^{\ii\theta})|^2}{2}\\\nonumber
    &= -\ee^{2\ii\theta}\mc S[\phi](\ee^{\ii \theta}) +  \frac{1-\ee^{2\sigma}}{2}
\end{align*}
where in the last step we used that $|\phi'(\ee^{\ii\theta})| = \ee^{\sigma}$.

In particular, by Lemma~\ref{lem:length_curv_local} and \eqref{eq:dl_k*} 
we have
\begin{align} \label{eq:k*_phi}
    -\ee^{-\sigma} \ee^{2\ii\theta}\mc S[\phi](\ee^{\ii \theta})  \,\dd \t = \dd \ell = \frac12 (1+k^*)\ee^{\sigma} \dd \t. 
\end{align}
Using Lemma~\ref{lem:Schwarzian_inverse} to express the Schwarzian action in terms of the inverse diffeomorphism and $|\phi'(\ee^{\ii\theta})|^{-1} = \ee^{-\sigma}$, we obtain
\begin{align*}
    \Isch(\varphi) = \int_{0}^{2\pi} -\ee^{-\sigma} \ee^{2\ii\theta}(\mc S[\phi](\ee^{\ii\theta})) \dd \theta =\int_{\Ep_h} \dd \ell 
\end{align*}
so it follows that $\Isch(\varphi)$ is equal to the signed length of the Epstein curve $\Ep_h$ induced by $h=\varphi_* \dd \t$.

As the total length of $h=\ee^\sigma \dd \theta$ is equal to $2\pi$, we can rewrite the previous equation and use \eqref{eq:kdl_k*} to express the Schwarzian action as
\begin{align*}
    \Isch(\varphi) = \int_{0}^{2\pi} \frac12 (1+k^*)\ee^{\sigma} \dd \theta = 2\pi - \int_{0}^{2\pi} \frac12(1-k^*)\ee^\sigma \dd \theta = 2\pi-\int_{\m S^1} k\,\dd \ell.
\end{align*}
It follows from Lemma~\ref{lem:Gauss_Bonnet} that the Schwarzian action is also equal to minus the signed area enclosed by the Epstein curve.
\end{proof}

As we prove later in two ways (Corollary~\ref{cor:non-negativeSch_isoperimetric} or Remark~\ref{rem:non-negativity_loewner}), $\Isch(\varphi)\geq 0$ if $\varphi$ is a $C^3$ circle diffeomorphism. This non-negativity can also be used to show that the Epstein curves of diffeomorphisms must have a non-immersed point. 
\begin{prop}\label{prop:non-immersed}
    If $\varphi:\m S^1\to \m S^1$ is a $C^3$ diffeomorphism and $h = \varphi_* \dd \theta$, then $\Ep_h$ has a non-immersed point.
\end{prop}
\begin{proof}
    Let $\phi=\varphi^{-1}$. By Remark~\ref{rem:non-immersed} and \eqref{eq:k_k*}, the set of non-immersed points on $\Ep_h$ is the locus of points where $k^*=-1$. By \eqref{eq:k*_phi}, this is equivalent the set $\{\theta: \mc S[\phi](\ee^{\ii \theta}) = 0\}$.
    
    Since both $\varphi$ and $\phi$ are diffeomorphisms, $\Isch(\varphi)\geq 0$ and $\Isch(\phi)\geq 0$. In particular, $\mc S[\phi](\ee^{\ii \theta}) \ee^{2\ii \theta}$ cannot be strictly negative (we already know it is real by Remark~\ref{rem:positivity_Losev}). It also cannot be strictly positive as by Lemma~\ref{lem:Schwarzian_inverse} that would imply $\mc S[\varphi](\ee^{\ii \theta}) \ee^{2\ii \theta}$ is strictly negative. Therefore since $\mc S[\phi](\ee^{\ii \theta}) \ee^{2\ii \theta}$ is continuous it must have a zero.
\end{proof}

\subsection{Isoperimetric inequality}\label{sec:isoperimetric}

We note that a $C^3$ diffeomorphism $\varphi$ defines a $C^2$ metric $h = \varphi_* \, \dd \t$ with total length $2\pi$. Conversely, any $C^2$ metric $h = \ee^\s \, \dd \t$ with total length $2\pi$ arises this way. In fact, let $\t_0 \in [0,2\pi)$, define 
$$\phi (\ee^{\ii \t}) = \exp \left(\ii \int_{\t_0}^{\t} \ee^{\s (s)} \, \dd s \right),$$
we have $\phi \in \Diff^3(\m S^1)$ and $\phi (\ee^{\ii \t_0}) = 1$. Let $\varphi = \phi^{-1}$, then we have $\varphi_* \dd \t = h$. 
Choosing a different $\t_0$ amounts to precomposing $\varphi$ by a rotation, and if $\varphi_1$ and $\varphi_2$ give rise to the same metric $h$, then $\varphi_1$ is obtained from $\varphi_2$ by a rotation.

In other words, we have a one-to-one correspondence between
\begin{equation}\label{eq:11_2pi}
\{C^2 \text{ metrics in } \m S^1 \text{ with total length } 2 \pi\} \xleftrightarrow[]{1 \,\colon 1} \Diff^3 (\m S^1)/\m S^1.
\end{equation}

We also consider metrics of the form $h_t = \ee^t h$, where $t \in \m R$, which then encompasses metrics in $\m S^1$ with arbitrary but finite total length. 
Recall from Theorem~\ref{thm:basic_epstein} and Remark~\ref{rem:equidistant_foliation} that the family of Epstein curves $(\Ep_{h_t})_{t \ge t_0}$ for some large $t_0$ forms an equidistant foliation, which converges in $C^1$ to $\m S^1$ as $t \to \infty$ (Remark~\ref{rem:fol_circle}).

\begin{lem}\label{lem:asymptoticLA}
    Let $h$ be a $C^2$ metric on $\m S^1=\partial_\infty\mathbb{H}^2$. Let $h_t=\ee^th$, while $L(\cdot)$, $A(\cdot)$ denote the signed length and signed area enclosed by an Epstein curve. Then, if $h$ is a metric of total length equal to $2\pi$, we have
    \begin{enumerate}[label=(\alph*)]
        \item\label{item:L+A} $L(\Ep_h) + A(\Ep_h) = 0$,
        \item\label{item:L_t} $L(\Ep_{h_t}) = 2\pi\sinh(t) - \ee^{-t}A(\Ep_h)$,
        \item\label{item:A_t} $A(\Ep_{h_t}) = 2\pi(\cosh(t)-1) +\ee^{-t}A(\Ep_h)$.
    \end{enumerate}
\end{lem}
\begin{proof}
    The first item follows from Theorem~\ref{thm:SchArea} as a consequence of Gauss--Bonnet, \eqref{eq:dl_k*}, and \eqref{eq:kdl_k*}:
\begin{align}\label{eq:2pilengthatinfinity}
        L(\Ep_h) + A(\Ep_h) = \int_{\m S^1} \,\dd \ell + \int_{\m S^1}k\,\dd \ell - 2\pi = \int_{\m S^1} h -2\pi = 0.
    \end{align}
    
    By \eqref{eq:length_t_from0}, \eqref{eq:totalcurvature_t_from0} and \eqref{eq:2pilengthatinfinity} we have that
    \begin{align}\label{eq:asympLAgrowth}
        L(\Ep_{h_t})=\int_{\m S^1} \,\dd \ell_t &= \cosh(t) \int_{\m S^1} \,\dd \ell + \sinh(t) \int_{\m S^1} k\,\dd \ell\\\nonumber
        &=-\cosh(t)A(\Ep_h) + \sinh(t)(A(\Ep_h)+2\pi)\\\nonumber
        &= 2\pi\sinh(t) - \ee^{-t}A(\Ep_h)\\\nonumber
        A(\Ep_{h_t})=\int_{\m S^1} k_t\,\dd \ell_t -2\pi &= \sinh(t) \int_{\m S^1} \,\dd \ell + \cosh(t) \int_{\m S^1} k\,\dd \ell -2\pi\\\nonumber
        &= -\sinh(t)A(\Ep_h) + \cosh(t)(A(\Ep_h)+2\pi) -2\pi\\\nonumber
        &= 2\pi(\cosh(t)-1) +\ee^{-t}A(\Ep_h)
    \end{align}
    which concludes \ref{item:L_t} and \ref{item:A_t}.
\end{proof}

A simple consequence of the asymptotic growth for length and area in Lemma~\ref{lem:asymptoticLA} is that for $t$ sufficiently large, both $L(\Ep_{h_t})$ and $A(\Ep_{h_t})$ are positive. This is relevant for the following theorem, where we express the Schwarzian action as the asymptotic isoperimetric excess of the Epstein curve $\Ep_{h_t}$ as $t$ goes to infinity. This should be compared to \cite[Thm.\,1.2]{VargasVianaisoperimetric}, where the renormalized volume of acylindrical convex co-compact manifolds was expressed as an asymptotic limit of the isoperimetric profile of the manifold.

The \emph{isoperimetric profile} in the hyperbolic disk is the function $J : \m R_+ \to \m R_+$, 
$$J(A) = 2\pi\sinh\left(\cosh^{-1}\left(\frac{A}{2\pi}+1\right)\right),$$
such that any (embedded) domain in $\m D$ with area $A$ has boundary length larger or equal to $J(A)$ where the equality is realized by disks. The disk with (hyperbolic) radius $t$ has boundary length $2\pi \sinh(t)$ and area $2\pi (\cosh (t)-1)$, which follows from, e.g., Lemma~\ref{lem:asymptoticLA}.

\begin{thm}[Schwarzian action as an asymptotic isoperimetric excess] \label{thm:excess}
    Let $\varphi :\m S^1\rightarrow\m S^1$ be a $C^3$ diffeomorphism, and $h_t=\ee^t \, \varphi_* \dd \t$, then
    \begin{align}\label{eq:excess}
        2\Isch(\varphi) = \lim_{t\rightarrow+\infty} \ee^t\Big( L(\Ep_{h_t}) - J(A(\Ep_{h_t}))  \Big).
    \end{align}
\end{thm}

\begin{proof}
    By items~\ref{item:L_t}, \ref{item:A_t} of Lemma~\ref{lem:asymptoticLA}, writing $A = A(\Ep_h)$, the right-hand side of \eqref{eq:excess} can be calculated as
    \begin{align*}
        &\lim_{t\rightarrow+\infty} \ee^t \left(  2\pi\sinh(t) - \ee^{-t}A - 2\pi\sinh\left(\cosh^{-1}\left( \cosh(t)+\ee^{-t}\frac{A}{2\pi}\right)\right) \right)\\\nonumber
        &=-A+ \lim_{t\rightarrow+\infty} \ee^t\left(  2\pi\sinh(t) - 2\pi\sinh\left(\cosh^{-1}\left( \cosh(t)+\ee^{-t}\frac{A}{2\pi}\right)\right) \right)\\\nonumber
        &=-A+ \lim_{t\rightarrow+\infty} \ee^t\left(  2\pi\sinh(t) - 2\pi\sqrt{\left( \cosh(t)+\ee^{-t}\frac{A}{2\pi}\right)^2-1}\right),
    \end{align*}
    where we are using the well-known identity $\sinh(\cosh^{-1}(x)) = \sqrt{x^2-1}$. By further simplifying the last expression, we obtain
    \begin{align*}
        &=-A+ \lim_{t\rightarrow+\infty} \ee^t\left(  2\pi\sinh(t) - 2\pi\sqrt{\sinh^2(t)+\ee^{-t}\cosh(t)\frac{A}{\pi} +\ee^{-2t}\frac{A^2}{4\pi^2}}\right)\\\nonumber
        &=-A+ \lim_{t\rightarrow+\infty} \ee^t \left( \frac{-4\pi^2\left( \ee^{-t}\cosh(t)\frac{A}{\pi} +\ee^{-2t}\frac{A^2}{4\pi^2}\right)}{2\pi\sinh(t) + 2\pi\sqrt{\sinh^2(t)+\ee^{-t}\cosh(t)\frac{A}{\pi} +\ee^{-2t}\frac{A^2}{4\pi^2}}}\right)\\\nonumber
        &=-A+ \lim_{t\rightarrow+\infty} \left( \frac{-4\pi^2\left( \ee^{-t}\cosh(t)\frac{A}{\pi} +\ee^{-2t}\frac{A^2}{4\pi^2}\right)}{\ee^{-t}\left(2\pi \sinh(t) + 2\pi\sqrt{\sinh^2(t)+\ee^{-t}\cosh(t)\frac{A}{\pi} +\ee^{-2t}\frac{A^2}{4\pi^2}}\right)}\right)\\\nonumber
        &=-A+ \frac{-2\pi A}{2\pi} = -2A
    \end{align*}
    which by Theorem~\ref{thm:SchArea} proves the desired identity.
\end{proof}

\begin{cor}
\label{cor:non-negativeSch_isoperimetric}
    The Schwarzian action of $\varphi \in \Diff^3 (\m S^1)$ is non-negative and vanishes if and only if $\varphi \in \PSU(1,1)$.
\end{cor}

\begin{proof}
    Consider the associated metric $h=\varphi_*\dd\theta$ on $\m S^1$ and consider $t$ sufficiently large so that the Epstein curve for $h_t=\ee^t h$ is embedded and $\Ep_{h_t}$ corresponds to the inner normal (Remark~\ref{rem:fol_circle}). Since by definition of the isoperimetric profile we have that $L(\Ep_{h_t}) - J(A(\Ep_{h_t})) \geq 0$, the non-negativity follows by Theorem~\ref{thm:excess}.

    If we had equality, then by Theorem~\ref{thm:SchArea}, it follows that $A(\Ep_h)=0$. Then for any $t$ so that $\Ep(h_t)$ is embedded, we will have by items \ref{item:L_t} and \ref{item:A_t} of Lemma~\ref{lem:asymptoticLA} that $\Ep_{h_t}$ is a solution of the isoperimetric problem in $\m D$. By the characterization of the solutions of the isoperimetric problem in $\m D$, we have then that $\Ep_{h_t}$ is the boundary of a ball $B_t$ of radius $t$ in $\mathbb{D}$. This implies that $\Ep_h$ is constant and equals the center of the ball $B_t$ by Theorem~\ref{thm:basic_epstein}. 
    Let $\a \in \PSU(1,1)$ be an M\"obius transformation sending the center of the ball to the origin, the naturality of the Epstein map (Lemma~\ref{lem:naturality}) shows that $\a_* h = \dd \t$ by \eqref{eq:11_2pi}. 
 We obtain that $\a \circ \varphi$ is a rotation.
 Hence, $\varphi \in \PSU(1,1)$. The converse follows by Example~\ref{ex:Ep_point} and the naturality of the Epstein map (Lemma~\ref{lem:naturality}).
\end{proof}





\begin{cor}
    Let $\varphi : \m S^1 \to \m S^1$ be a $C^3$ diffeomorphism. Let $h = \varphi_* \dd \t = \ee^\s \dd \t$ be the pushforward of $\dd \t$ by $\varphi$. Then
    \[
    \int_{\m S^1} k^*h =  2 \Isch (\varphi)  -2\pi \geq -2\pi,
    \]
    with equality if and only if $\varphi\in \PSU(1,1)$.
\end{cor}
\begin{remark}\label{rem:action_form_k*}
    The Schwarzian action $\mc I(\cdot)$ in e.g.\ \cite{BLW_schwarzian,Stanford_Witten} is written with a different normalization, which by Remark~\ref{rem:positivity_Losev} is related to ours by $\frac{1}{\pi}\mc I(\varphi) = 2 \Isch(\varphi) -2\pi$. The above gives a geometric interpretation for this form of the action as the total curvature at infinity. 
\end{remark}
\begin{proof}
We write $\phi = \varphi^{-1}$. 
It follows from \eqref{eq:k*_phi} that 
$$\int_{S^1} k^* h = -\int_{0}^{2\pi} 2 \ee^{- \s} \ee^{2\ii \t} \mc S[\phi] (\ee^{\ii \t})  + \ee^{\s}\,\dd \t = -\int_{0}^{2\pi} 2 |\phi'|^{-1} \ee^{2\ii \t} \mc S[\phi] (\ee^{\ii \t})  \,\dd \t - 2\pi$$
where we used that the total length of $h$ is $2\pi$.  Lemma~\ref{lem:Schwarzian_inverse} then implies 
$$\int_{S^1} k^* h =  2 \Isch (\varphi) - 2\pi$$
as claimed.
\end{proof}

\begin{remark}\label{rem:total_curvature_*}
    Observe that for  a $C^2$ metric $h=\ee^\sigma\dd\theta$ with curvature at infinity $k^*$, the integral $\int_{\m S^1} k^*h$ can be reduce to
\begin{align}
    \int_{\m S^1} k^*h &= \int_0^{2\pi} \ee^{-\sigma}(2\sigma_{\theta\theta}-\sigma_{\theta}^2-1) \dd\theta\\\nonumber
    &= \int_0^{2\pi} \ee^{-\sigma}(\sigma_{\theta}^2-1) \dd\theta
\end{align}
by an integration by parts.
Hence, we can extend the definition of \emph{total curvature at infinity} $\int_{\m S^1} k^*h$ to $C^1$ metrics in $\m S^1$ as 
    \begin{align}
        \int_{\m S^1} k^*h = \int_0^{2\pi} \ee^{-\sigma}(\sigma_{\theta}^2-1) \dd\theta
    \end{align}
    which then also extends the definition of $\Isch (\cdot)$ to $C^2$ diffeomorphisms. 
\end{remark}

\subsection{Schwarzian action for $n$-folds covers}\label{subsec:n-fold}
We discuss the Epstein curve associated with a degree $n$ local diffeomorphism of $\m S^1$ (which we also call an $n$-fold cover of $\m S^1$).
\begin{definition}
    The group $\Mobn \simeq \PSL^{(n)}(2, \m R)$ is the group of circle diffeomorphisms obtained by conjugating $\PSU(1,1)$ with the $n$-th power map $z\mapsto z^n$. That is, \[
    B\in \Mobn \iff \exists A \in \PSU(1,1),\, (B(z))^n = A(z^n).
    \]
\end{definition}

Observe that the Schwarzian action defined by \eqref{eq:def_Sch_action} extends to any local diffeomorphism $\varphi:\m S^1 \rightarrow \m S^1$. As these maps are $n$-fold covers of the circle, it is natural to study the Schwarzian action in $\Diff_n(\m S^1): =\lbrace \varphi:\m S^1 \rightarrow \m S^1\,|\, \varphi \text{ is a $n$-fold cover} \rbrace$. To describe this action, the following lemma will be useful.

\begin{lem}
    The map $^n:\Diff(\m S^1) \rightarrow \Diff_n(\m S^1)$ given by
    \[
    \phi \in \Diff(\m S^1) \mapsto (^n\phi)(z) : = (\phi(z))^n
    \]    
    is an $\m Z/n\m Z$ principal bundle with fiber action given multiplication by a $n$-root of unity.
\end{lem}

Using this correspondence, we can express the Schwarzian action on $\Diff_n(\m S^1)$ as a $\Mobn$ invariant functional on $\Diff(\m S^1)$.

\begin{prop}
    We define the map $\Isch^n: \Diff(\m S^1) \rightarrow \m R$ as the composition of $^n:\Diff(\m S^1) \rightarrow \Diff_n(\m S^1)$ and $\Isch$. We have $\Isch^n$ is invariant by the left action of $\Mobn$ and is given by
    \begin{align*}
    \Isch^n(\phi) & :=  \Isch(^n \phi) = \int_{\m S^1} - \ii z \mc S [^n\phi](z) \,\dd z\\
    &= \Isch(\phi) + \int_{\m S^1} \frac{n^2-1}{2}\left(\frac{\phi'(z)}{\phi(z)}\right)^2 \ii z \,\dd z\\
    &= \Isch(\phi) + \frac{1-n^2}{2}\int_0^{2\pi} |\phi'(e^{\ii\theta})|^2 \,\dd\theta.
    \end{align*}
\end{prop}
In particular, $\Isch^n$ is well defined on the coadjoint orbit $\Mobn \backslash \Diff(\m S^1)$ \cite{witten88,Alekseev}.

\begin{proof}
    Let us prove the $\Mobn$ invariance first. Let $\phi\in \Diff(\m S^1)$, $\varphi = \prescript{n}{}{\phi}$, $B\in \Mobn$ and $A\in \PSU(1,1)$ such that $(B(z))^n = A(z^n)$. Then
    \begin{align*}
        \Isch^n (B\circ\phi) &= \Isch(\prescript{n}{}{(B\circ\phi)}) = \Isch(A(\prescript{n}{}{\phi}))\\
        &= \Isch(\prescript{n}{}{\phi}) = \Isch^n (\phi).
    \end{align*}
    To get the formula for $\Isch^n$ we can use the chain rule for the Schwarzian derivative and that the Schwarzian derivative for $z\mapsto z^n$ is given by $(1-n^2)/(2z^2)$ to obtain
    \begin{align*}
        \mc S[\varphi](z) = \frac{1-n^2}{2\phi^2(z)}\left(\phi'(z)\right)^2 + \mc S[\phi](z)
    \end{align*}
    from which we obtain
    \begin{align*}
        \Isch(\varphi) = \int_{\m S^1}
   - \ii z \, \mc S[\varphi](z) \,\dd z = \int_{\m S^1} \frac{n^2-1}{2\phi^2(z)}\left(\phi'(z)\right)^2\ii z
   - \ii z \, \mc S[\phi](z) \,\dd z
    \end{align*}
    and the first identity follows. For the second identity it remains to observe that as $\phi:\m S^1 \rightarrow \m S^1$ we have that $\phi'(z)z = |\phi'(z)|\phi(z)$.
\end{proof}

Recall that the definition of Epstein curve \eqref{eq:Epstein_point} and its normal \eqref{eq:Epsteinnormal} are given by the $1$-jet of the metric $h$ on the circle. Hence we can extend these definitions to multi-valued metrics and define the Epstein curve associated to a $n$-fold cover of $\m S^1$.

\begin{definition}
    Let $\varphi:\m S^1 \rightarrow\m S^1$ be a $C^2$ $n$-fold cover. Then its associated metric in $\m S^1$ is given by the multi-valued $C^1$ metric $h=\varphi_*(\dd\theta)$. To this multi-valued metric we have its associated Epstein curve $\Ep_h$ and normal $\widetilde{\Ep}_h$, which are curves parametrized by
    \[
    \Ep_h(e^{\ii\theta}) :  = \Ep_{\varphi_*(\dd\t)}(\varphi(e^{\ii\t}))
    \]
    \[
    \widetilde{\Ep}_h(e^{\ii\theta}) :  = \widetilde{\Ep}_{\varphi_*(\dd\t)}(\varphi(e^{\ii\t}))
    \]
\end{definition}

See Figure~\ref{fig:n_fold_cover} and Figure~\ref{fig:n_fold_cover2} for examples.

\begin{figure}
    \centering
    \includegraphics[width=0.8\linewidth]{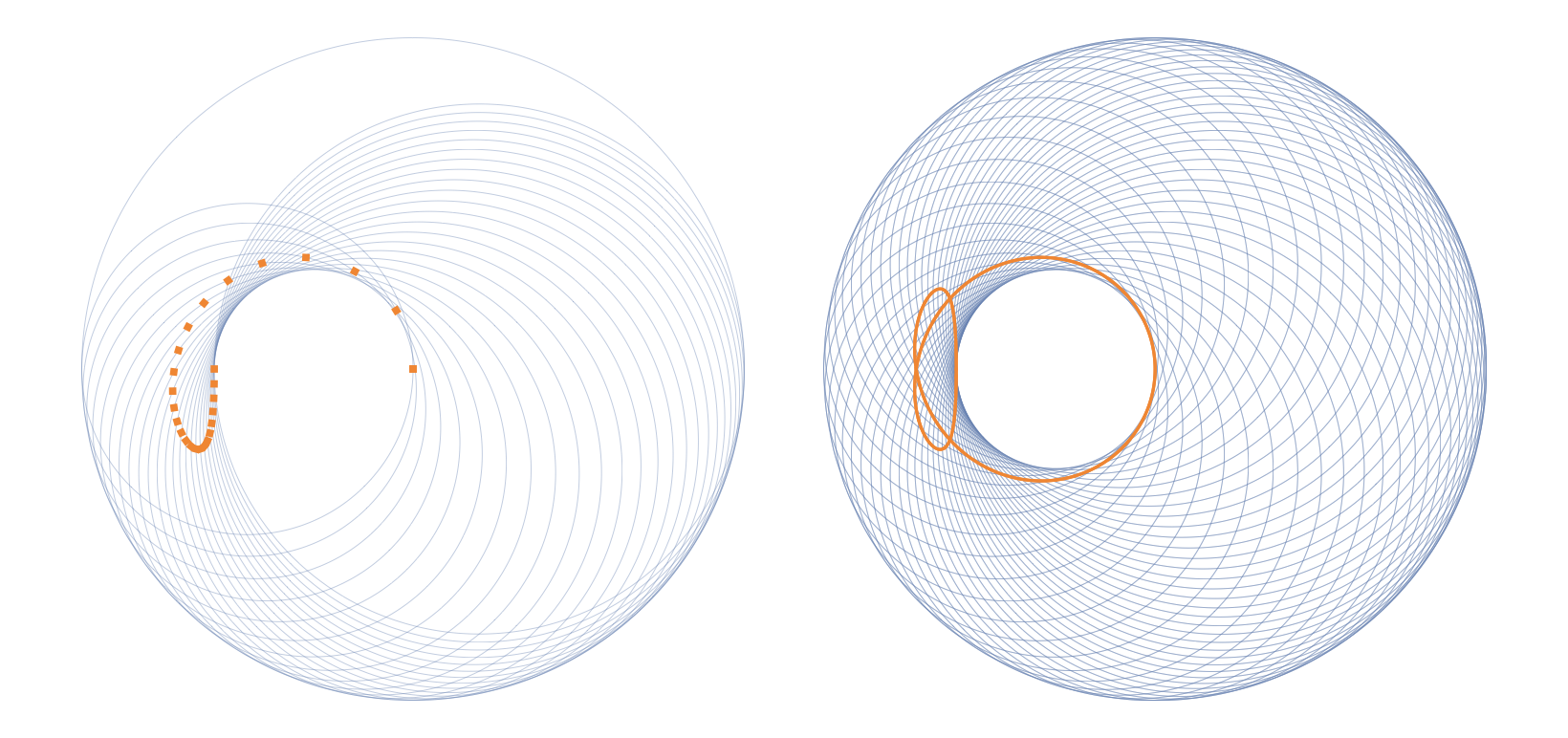}
    \caption{On the right are the Epstein curve and horocycles for the double cover corresponding to the metric $h^a$, $a=1/2$ in Example~\ref{ex:z3_conjugated} (i.e., $h^{1/2}$ pushed forward by $z^2$). On the left are twenty five evenly-space horocycles for only $\theta\in [0,\pi/2]$, with the point on the Epstein curve marked on them.}
    \label{fig:n_fold_cover2}
\end{figure}

Observe that for $n=1$ this definition gives the associated Epstein curve as in the previous sections, but after a right composition by (in this case diffeomorphism) $\varphi$. While this changes the Epstein curve and its normal as maps, it does so only through a reparametrization, so geometric quantities such as the total length and total curvature remain the same.

Local geometric concepts associated with circle diffeomorphism generalize in a straightforward fashion to $n$-fold covers. Such is  the case for the arc-length and geodesic curvature of Definition~\ref{def:length_curv} (including their extensions by Lemma~\ref{lem:length_curv_local} to non-immersion points), the total length and total geodesic curvature of Definition~\ref{def:total_length_curv}, and the geodesic curvature at infinity of Definition~\ref{def:curvature_infty} (as a parametrized multi-valued function). Moreover, since Lemma~\ref{lem:asymptoticforms} is a purely local statement, it directly extends to the multi-valued metric $h=\varphi_*(\dd\t)$.

We can also extend Definition~\ref{def:signed_A} for the signed area $A(\Ep_h)$, although here the statement of Lemma~\ref{lem:Gauss_Bonnet} needs some adjustment. Namely, \eqref{eq:AreavsCurv} still holds for any two multi-valued metrics with the same number of branches. On the other hand, equation \eqref{eq:Epstein_Gauss-Bonnet} required that we verify it for one metric. In this case, we can compute it for $e^t\dd\t$ as a $n$-valued metric, from which it follows:
\begin{align}\label{eq:nth-Epstein_Gauss-Bonnet}
        A(\Ep_h) = \int_{\m S^1} k\,\dd\ell -2\pi n.
\end{align}

We can also derive a geometric interpretation for $\Isch^n$ as follows.

\begin{thm}\label{thm:IschnLength}
    Let $\phi:\m S^1 \rightarrow \m S^1$ a $C^3$ diffeomorphism. Let $\varphi = \prescript{n}{}{\phi}$ be the associated $n$-fold cover with associated multi-valued metric $h=\varphi_*(\dd\t)$. Then $\Isch^n(\phi)$ can be computed as
    \[
    \Isch^n(\phi) = L(\Ep_h) = -A(\Ep_h)-2\pi(n-1).
    \]
\end{thm}
\begin{proof}
    As in the proof of Theorem~\ref{thm:SchArea} we can verify that the integrand of $\Isch^n(\phi) = \Isch(\varphi)$ is equal to the arc-length of $\Ep_h$ for the appropriate branch of the multi-valued metric $h$, so the first equality follows.

    For the second equality, as in Lemma~\ref{lem:asymptoticLA} we have that the total length and the total geodesic curvature of $\Ep_h$ add up to the total length of $\varphi_*(\dd\t)$, which is $2\pi$. Hence by using this and \eqref{eq:nth-Epstein_Gauss-Bonnet} the second equality follows.
\end{proof}

\section{Bi-local observables} \label{sec:bi-local}

For $\varphi\in \Diff^1(\m S^1)$,  define a $C^1$ function $\tilde{\varphi}:[0,2\pi]_{/0\sim 2\pi}\to [0,2\pi]_{/0\sim 2\pi}$ by $\varphi(\ee^{\ii \theta}) = \ee^{\ii \tilde{\varphi}(\theta)}$. Pick $u,v\in \m S^1$ distinct, and let $x,y\in [0,2\pi]$ be their arguments. The \textit{bi-local observable} of $\varphi$ for $u,v$ is the function:
\begin{align}
    \mc O(\varphi;u,v) := \frac{\sqrt{\tilde{\varphi}'(x)\tilde{\varphi}'(y)}}{2\sin(\frac{1}{2}(\tilde{\varphi}(x)-\tilde{\varphi}(y)))} = \frac{\;\;|\varphi'(u)\varphi'(v)|^{1/2}}{|\varphi(u)-\varphi(v)|}.
\end{align}

It is straightforward to see that these are invariant under the left action $\varphi\mapsto \a\circ \varphi$ for $\a\in \PSU(1,1)$, and hence descend to functions on $\PSU(1,1)\backslash \Diff^1 (\m S^1)$. In this section, we relate the bi-local observables to renormalized lengths of hyperbolic geodesics truncated by the horocycles in the construction of the Epstein curve (Proposition~\ref{prop:bilocal_is_length}). From this, we show that the bi-local observables along the edges of any ideal triangulation determine the circle diffeomorphism in $\PSU(1,1)\backslash\Diff^1(\m S^1)$ (Proposition~\ref{prop:ideal_triangulation}).

Although we do not deal directly with the Schwarzian field theory, let us comment on the role played by bi-local observables there. In field theories with a gauge symmetry (i.e.\ here the $\PSU(1,1)$ action), correlation functions of gauge-invariant observables often play a central role in determining the field theory. A collection of bi-local observables $\{\mc O(\varphi; u_j,v_j)\}_{j=1}^n$ is called \textit{non-crossing} if the lines (hyperbolic geodesics) in the disk connecting $(u_j,v_j)$ do not cross. Correlation functions for non-crossing bi-local observables were computed in the physics literature using a connection to two-dimensional CFT and predicted to determine the theory \cite{MertensEtAl}; in \cite{LosevCorrelations} these correlation functions were computed mathematically rigorously and shown to determine the measure for the theory using the probabilistic framework for Schwarzian field theory. In the conjectured holography duality between Schwarzian field theory and JT gravity in the disk, bi-local observables in the Schwarzian theory are predicted to correspond to \textit{boundary-anchored Wilson lines} for JT gravity \cite{Blommaert}.

\subsection{Bi-local observables as renormalized length}

The Epstein curve for $\varphi\in \Diff^1(\m S^1)$ is the envelope of horocycles $(H_z)_{z\in \m S^1}$ in $\m D$, where the size of $H_z$ is determined at each point by $\varphi'(u)$ and $\varphi(u)$ if $z = \varphi (u)$. We saw in Theorem~\ref{thm:SchArea} that $\Isch(\varphi)$ can be computed as the hyperbolic area enclosed by the Epstein curve, loosely speaking as a ``renormalized area'' where we cut off the disk using horocycles at all points $u\in \m S^1$. 
The main result of this section is Proposition~\ref{prop:bilocal_is_length} where we show that the bi-local observable $\mc O(\varphi; u,v)$ is related to the \textit{renormalized length} for $\varphi$ of the hyperbolic geodesic from $\varphi(u)$ to $\varphi(v)$, given by truncation with the same horocycles as the Epstein curve construction. 

Let $d_{\m D}$ denote the distance function for the hyperbolic metric on $\m D$.

\begin{definition}
  Fix $\varphi\in \Diff^1(\m S^1)$ and let $(H_z)_{z\in \m S^1}$ denote the horocycles associated with the metric $\varphi_*\dd \theta$ as in Section~\ref{subsec:eps_1}. For $u,v\in \m S^1$ distinct, let $(u,v)$ be the hyperbolic geodesic between them in $\m D$. 
  We define the \emph{renormalized length} for $\varphi$ from $u$ to $v$, denoted $\mathrm{RL}_{\varphi}(u,v)$, to be the signed hyperbolic distance from $H_{\varphi(u)}$ to $H_{\varphi(v)}$ along $(\varphi(u),\varphi(v))$. The sign is positive if $H_{\varphi(u)}$ and $H_{\varphi(v)}$ are disjoint and negative otherwise. 
\end{definition}

\begin{remark}\label{rem:log_lambda}
    Renormalized length is closely related to the $\log \Lambda$ lengths of decorated Teichm\"uller theory, see e.g. \cite{PennerBook}. These are also related to diamond shears studied by the last two authors \cite{SWW_shears}. The renormalized length corresponds to $\log \Lambda$ length for the particular choice of horocycles given by the Epstein construction. 
\end{remark}

For any $\varphi$, renormalized length is invariant under post-composition by M\"obius transformations by Lemma~\ref{lem:naturality}. As the first step in understanding renormalized length, we can compute it for $\varphi = \mathrm{Id}$. 
\begin{lem}\label{lem:identity_bilocal}
    For any $u,v\in \m S^1$ distinct, the renormalized length for $\varphi=\mathrm{Id}$ is 
    \begin{align}
        \mathrm{RL}_{\mathrm{Id}}(u,v) = -\log \frac{4}{|u-v|^2}.
    \end{align}
\end{lem}
\begin{proof}

Let $(H_z^0)_{z\in \m S^1}$ denote the horocycles associated with $\dd \theta$. By Example~\ref{ex:round}, for each $z$, $H_z^0$ is the horocycle based at $z$ containing $0$. This is invariant under rotation, so it suffices to compute $\mathrm{RL}_{\mathrm{Id}}(-1,v)$ for $v\in \m S^1$. 

We compute the desired length from $H_{-1}\cap (-1,v)$ to $H_v \cap (-1,v)$ using the halfplane model $\m H = \{\zeta\in \m C \,|\, \Im (\zeta) > 0\}$ of the hyperbolic plane. Let $\cayley:\m D\to \m H$ be the Cayley map given by $\cayley(z) = \ii (1-z)/(1+z)$. Since $\cayley(0) = \ii$, $(H_z^0)_{z\in \m S^1}$ is sent by $\cayley$ to the collection of horocycles $(\tilde{H}_x)_{x\in\Rhat}$ all of which contain $\ii$. Note that $\cayley(-1) = \infty$ and let $\lambda = \cayley(w)$.

Since $\tilde{H}_{\infty}$ is the horocycle of Euclidean height $1$, $\tilde{H}_\infty$ and $(\lambda,\infty)$ intersect at $\lambda+\ii$. For $\lambda\in \m R$, $\tilde H_\lambda$ must be centered at $\lambda+\ii r$ for some $r>0$, and must contain $\lambda$ and $\ii$. We use this to solve for $r$: 
\begin{align*}
    r^2 = \lambda^2 + (r-1)^2 \implies r = \frac{1}{2}(1+\lambda^2).
\end{align*}
Therefore the intersection of $\tilde{H}_\lambda$ and $(\lambda,\infty)$ is at $\lambda+(1+\lambda^2)\ii$. Note that $1+\lambda^2\geq 1$ for all $\lambda\in \m R$, corresponding to the fact that $\tilde{H}_\lambda$ and $\tilde{H}_\infty$ always intersect. Hence, 
\begin{align*}
    \mathrm{RL}_{\mathrm{Id}}(-1,w) &= \mathrm{RL}_{\mathrm{Id}}(\lambda,\infty) = \int_{1+\lambda^2}^{1} \frac{1}{y}\, \dd y = - \log(1+\lambda^2) \\
    &= -\log \frac{4v}{(1+v)^2} = -\log \frac{4}{|v+1|^2},
\end{align*}
which completes the proof.
\end{proof}

One can observe directly from the lemma that $\mc O(\mathrm{Id};u,v)^2=\frac{1}{4} \exp(-\mathrm{RL}_{\mathrm{Id}}(u,v))$. 

\begin{prop}\label{prop:bilocal_is_length}
    Fix $\varphi\in \Diff^1(\m S^1)$ and $u,v\in \m S^1$ distinct. 
    The bi-local observable and renormalized length are related by: 
    \begin{align}\label{eq:bilocal}
        \mc O(\varphi;u,v)^2 = \frac{|\varphi'(u)\varphi'(v)|}{|\varphi(u)-\varphi(v)|^2} =\frac{1}{4}\exp(-\mathrm{RL}_{\varphi}(u,v)). 
    \end{align}
\end{prop}

The key observation to prove Proposition~\ref{prop:bilocal_is_length} is the following lemma coming from the naturality of the Epstein construction.

\begin{lem}\label{lem:naturality2}
    Let $(H^0_z)_{z \in \m S^1}$ denote the horocycles associated with the metric $\dd \t$ as in Example~\ref{ex:round} 
    and $(H_z)_{z \in \m S^1}$ those associated with the metric $\varphi_*\dd \t$. Then $H_{\varphi (z)} = \a_z (H^0_z)$, where $\a_z \in \PSU(1,1) = \Isom_+(\m D)$ is any osculating M\"obius map at $z$ such that $\a_z(z) = \varphi (z)$ and $\a_z'(z) = \varphi'(z)$.
\end{lem}
\begin{proof}
    We use the notation as in 
    \eqref{eq:Eps_invariant}. Since $H_{\varphi (z)} := H_{\varphi (z)} (\varphi_* \dd \t)$ is a horocycle centered at $\varphi (z)$ which only depends on the metric $\varphi_* \dd \t$ at $\varphi (z)$,  we have
    $$H_{\varphi (z)}  = H_{\varphi (z)} (\varphi_* \dd \t) = H_{\a_z (z)} (\a_{z*} \dd \t) = \a_z (H_z^0)$$
    where the last equality follows from Lemma~\ref{lem:naturality}.
\end{proof}

\begin{proof}[Proof of Proposition~\ref{prop:bilocal_is_length}] 

Fix $u,v\in \m S^1$. As in Lemma~\ref{lem:naturality2}, let $(H_{z})_{z\in \m S^1}$ and $(H_z^0)_{z\in \m S^1}$ denote the horocycles for $\varphi$ and $\mathrm{Id}$ respectively. 

First consider the case where $\varphi$ fixes $u,v$.
By Lemma~\ref{lem:naturality2}, for any $z\in \m S^1$, $H_{\varphi(z)} = \a_z (H_z^0)$, where $\a_z \in \PSU(1,1)$ has $\a_z (z) = \varphi(z)$ and $\a_z'(z) = |\varphi'(z)|$. Since $\varphi(u)=u$, $H_{\varphi(u)}=H_u$ and $H_u^0$ are both horocycles at $u\in \m S^1$, and further $H_u$ and $H_u^0$ are at signed distance $-\log |\varphi'(u)|$, which is positive when $H_u$ is contained in the horoball bounded by $H_u^0$. This also follows directly from Lemma~\ref{lem:horo_t}.
The same statement holds for $v$. 
Hence, from the definition of renormalized length, 
$$\mathrm{RL}_\varphi (u,v) = \mathrm{RL}_{\mathrm{Id}} (u,v) - \log |\varphi'(u)\varphi'(v)| = -\log \frac{4 |\varphi'(u)\varphi'(v)|}{|u-v|^2},$$
which proves \eqref{eq:bilocal} when $\varphi$ fixes $u,v$.
  
If $\varphi$ does not fix $u,v$, then we can find a M\"obius transformation $\alpha:\m D\to \m D$ such that $\psi=\alpha\circ \varphi$ fixes $u,v$. Then using $\alpha'(z) \alpha'(w) = \frac{(\alpha(z)-\alpha(w))^2}{(z-w)^2}$,  we get  
  \begin{align*}
     \psi'(u) \psi'(v) = \alpha'(\varphi(u)) \varphi'(u) \,\alpha'(\varphi(v)) \varphi'(v) = \frac{(u-v)^2}{(\varphi(u)-\varphi(v))^2}\varphi'(u) \varphi'(v).
  \end{align*}
  As such, 
  \begin{align}
      \mathrm{RL}_{\varphi}(u,v) = \mathrm{RL}_{\psi}(u,v) = -\log \frac{4|\psi'(u)\psi'(v)|}{|u-v|^2} =  - \log \frac{4|\varphi'(u)\varphi'(v)|}{|\varphi(u)-\varphi(v)|^2}
  \end{align}
  which proves the result.
\end{proof}

\subsection{From bi-local observables to a diffeomorphism} \label{subsec:inverse_pb}

In this section, we explain how to reconstruct the diffeomorphism from its bi-local observables on any ideal triangulation. This is reminiscent of the coordinate system of $\log \Lambda$ length for decorated Teichm\"uller space mentioned in Remark~\ref{rem:log_lambda}. 

A geodesic ideal triangle in $\m D$ is a triangle with vertices in $\m S^1$ and hyperbolic geodesic edges. An \textit{ideal triangulation} (or \textit{ideal tessellation}) $\mc T$ of $\m D$ is a locally-finite collection of geodesic ideal triangles that cover $\m D$ and have non-overlapping interiors. We denote the vertices and edges of $\mc T$ by $(V,E)$. In particular, $V\subset \m S^1$ is dense, and $E$ is a collection of non-crossing geodesics in the disk. See Figure~\ref{fig:farey} for an example. 
\begin{figure}[ht]
    \centering
    \includegraphics[width=.5\textwidth]{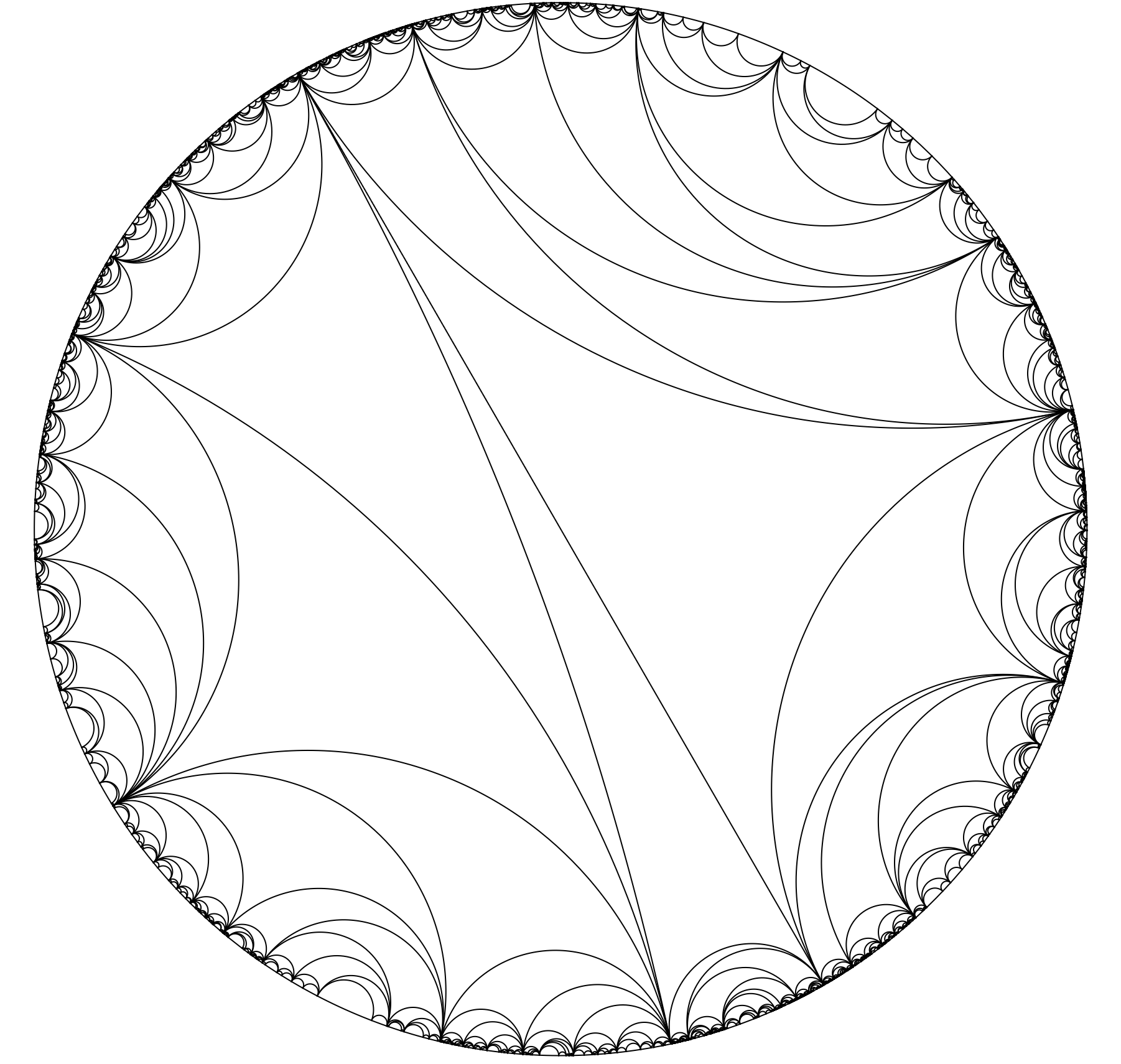}
    \caption{Example of an ideal triangulation, where only finitely many edges are drawn.}
    \label{fig:farey}
\end{figure}

\begin{prop} \label{prop:ideal_triangulation}
    Fix $[\varphi]\in \PSU(1,1)\setminus\Diff^1(\m S^1)$. For any ideal triangulation $\mc T = (V,E)$ of $\m D$, $\{ \mc O(\varphi; u,v)\}_{(u,v)\in E}$ determines $[\varphi]$.
\end{prop}

\begin{remark}
   For any $u,v\in \m S^1$, the observable $\mc O(\varphi; u,v)$ is determined just by the equivalence class $[\varphi]\in \PSU(1,1)\setminus \Diff^1(\m S^1)$, as for any $\alpha\in \PSU(1,1)$, $\mc O(\alpha\circ \varphi; u,v) = \mc O(\varphi; u,v)$. 
\end{remark}

\begin{remark}
    Since $V$ is dense in $\m S^1$, the data $\{\varphi(v)\}_{v\in V}$ determines $\varphi$. As we will see in the proof, $\{ \mc O(\varphi; u,v)\}_{(u,v)\in E}$ determines $\varphi(v)$ and $H_{\varphi(v)}$ (and hence $\varphi'(v)$) for all $v\in V$. We require that $\{ \mc O(\varphi; u,v)\}_{(u,v)\in E}$ come from a diffeomorphism in the statement since this over determination indicates that not all collections of numbers $(\mc O_e)_{e\in E}$ are the observables of a diffeomorphism.
\end{remark}

\begin{proof}[Proof of Proposition~\ref{prop:ideal_triangulation}]

Let $T$ denote the set of triangles in $\mc T$. The dual graph of $\mc T$, denoted $\mc T^*=(V^*,E^*)$, has vertices $V^* = T$ and edges $E^*$ in correspondence with $E$ (triangles $t_1,t_2$ are connected by an edge $e^* \in E^*$ if and only if they share an edge $e\in E$). In particular, since $\mc T$ is an ideal triangulation, $\mc T^*$ is an infinite trivalent tree.
We use this tree structure to iteratively recover the values of $\varphi (u)$ and the horocycle decoration $H_{\varphi(u)}$ for all $u \in V$ from $\{\mc O (\varphi; u,v)\}_{(u,v) \in E}$.

For this, fix a root edge $e_0=(u_0,v_0)\in E$ and assume without loss of generality that $\varphi$ fixes $u_0$ and $v_0$, and moreover that $H_{\varphi (u_0)} = H_{u_0}^0$.  Since $\mc O(\varphi;u_0,v_0)$ is the hyperbolic distance from $H_{u_0}$ to $H_{v_0}$ along $e_0$, this information determines $H_{v_0}$.

Now, take any ideal triangle $t=(u,v,w)\in T$. Let $a,b,c$ denote the bi-local observables for $(u,v)$, $(v,w)$ and $(w,u)$ respectively. Given the data $\{\varphi(u),\varphi(v),H_{\varphi(u)},H_{\varphi(v)}\}$, the position of $\varphi(w)$ is then determined by $\varphi (u)$, $\varphi (v)$, the relative distance $b-c$ to $H_{\varphi(u)}$ vs.\ to $H_{\varphi(v)}$, 
and the order (clockwise or counterclockwise) of $u,v,w$, and is independent of the horocycle $H_{\varphi(w)}$. Once $\varphi(w)$ is known, the value of $b$ (or $c$) then determines $H_{\varphi(w)}$.

Therefore knowing $\{\varphi(u),\varphi(v),H_{\varphi(u)},H_{\varphi(v)}\}$ for one edge $e=(u,v)$ of a triangle $t=(u,v,w)\in T$, the bi-local observables then determine $\{\varphi(w),H_{\varphi(w)}\}$. Iterating this for all $t\in T$ along the branches of the tree $\mc T^*$ starting from $e_0$, we see that the observables determine $\varphi(v)$ for all $v\in V$. Since $V\subset\m S^1$ is dense, this determines $\varphi$.
\end{proof}

\section{Schwarzian action and  Loewner energy} \label{sec:Sch_Loewner}

In this section, we prove Theorem~\ref{thm:var_IL_Sch} relating the Loewner energy to the Schwarzian action. Therefore, although some results in this section hold for curves with lower regularity, we restrict our consideration to Jordan curves that are at least $C^{3,\alpha}$ regular for some $\a > 0$ so that their welding homeomorphisms are at least $C^3$ regular.

 \subsection {Variation of Loewner energy}  \label{subsec:var_I_L}

We now recall briefly the variational formula of the Loewner energy proved in \cite{TT06}. 

Let $\g$ be a $C^{3,\alpha}$ curve, $f$ (resp., $g$) be a conformal map $\m D \to \O$ (resp., $\m D^* \to \O^*$), where $\O$ (resp., $\O^*$) is the connected component of $\Chat \smallsetminus \g$ which does not contain $\infty$ (resp., which contains $\infty$). By Kellogg's theorem, $f$ and $g$ extend to $C^{3,\a}$ diffeomorphisms on the closures of $\m D$ and $\m D^*$ respectively.
In particular, the \emph{welding homeomorphism} $\varphi : = g^{-1} \circ f|_{\m S^1}$ is a $C^{3,\alpha}$ diffeomorphism of $\m S^1$.

Let $\mu \in L^2 (\O^*, \rho_{\O^*}) \cap L^\infty (\O^*)$ be a Beltrami differential, where $\rho_{\O^*} (z)$ is the density of the hyperbolic metric on $\O^*$. 
For $\vare \in \m{R}$ small enough such that $\|\vare \mu\|_{\infty} < 1$, let $\omega^\emu: \Chat \to \Chat$ be any quasiconformal mapping solving the Beltrami equation
$$\bar \partial \o^\emu = \emu \,\partial \o^\emu.$$
In particular, $\o^\emu$ is conformal in $\O$. By the measurable Riemann mapping theorem, another solution to the Beltrami equation is of the form $\a \circ \o^\emu$ where $\a \in \PSL (2,\m C)$ is a M\"obius map of $\Chat$.

If $\o^\emu$ is \emph{normalized} such that $\o^\emu$ fixes $0$ and $\infty$ with $(\o^{\emu})' (0) = 1$, then the map $\o^\emu$ depends analytically on $\vare$ in a neighborhood of $0$. The vector field of the infinitesimal variation
$v = v_\mu: = \dd \o^{\emu}/ \dd \vare |_{\vare = 0} = \dot \o$ satisfies  
$\mu = \bar \partial v$, $v (z) = O (z^2)$ as $z \to 0$, and $v(z) = O(z)$ as $z \to \infty$. 

The following result was proved in \cite{TT06} (see also \cite{sung2023quasiconformal}).
\begin{thm}\label{thm:var_I} Let $\g$ be a $C^{3, \a}$ curve for some $\a > 0$ and $\mu \in L^2 (\O^*, \rho_{\O^*}) \cap L^\infty (\O^*)$.
    With the notations above,  let $\g^{\emu} =\omega^\emu(\g)$. 
     We have
      \begin{align} \label{eq:variation_Bloop}
    \frac{\dd}{\dd \vare}\bigg|_{\vare=0} I^L (\g^{\emu})& =   -\frac{4}{\pi} \, \mathrm{Re}  \left ( \int_{\O^*} \mu(z) \mc S [g^{-1}](z) \,\dd^2 z  \right)\\
    & = -\frac{2}{\pi} \, \mathrm{Im}  \left[\int_{\partial \O^*} v(z) \mc S [g^{-1}](z) \,\dd z\right], \label{eq:int_vector}
\end{align}
where $g : \m D^* \to \O^*$ are any conformal maps, $\dd^2 z$ is the Euclidean area measure, $\dd z$ is the contour integral.
\end{thm}

We note that $\mc S[g^{-1}] (z) = O(z^{-4})$ as $z \to \infty$ and the right-hand side of \eqref{eq:variation_Bloop} is convergent.

\begin{remark}\label{rem:mu_v}
The second identity \eqref{eq:int_vector} follows from 
$\mu = \bar \partial  v$ and  Stokes' formula:
\begin{align*}
     \int_{\O} \mu(z) \mc S [f^{-1}](z) \,\dd^2 z   & = \int_{\O} \mu(z) \mc S [f^{-1}](z) \,\dd^2 z \\
     & = \frac{1}{2 \ii} \int_{\O } \bar \partial \left(v(z) \mc S [f^{-1}](z)\right) \,\dd \bar z \wedge \dd z \\
     & = \frac{1}{2 \ii} \int_{\O} \dd \left( v(z) \mc S [f^{-1}](z) \, \dd z \right) \\
     & = \frac{1}{2 \ii} \int_{\partial \O}  v(z) \mc S [f^{-1}](z) \, \dd z.
\end{align*}
\end{remark}

\subsection{Derivation of the Schwarzian action}

Using the same notation as in the previous section,
we consider the family of \emph{equipotentials} $\Big(\g^\vare := f((1-\vare) \m S^1)\Big)_{0 \le \vare <1}$ in $\O$ bounded by  a $C^{3,\a}$ Jordan curve $\g$ for some $\a > 0$.

\begin{thm}\label{thm:var_IL_Sch}
We have 
    $$\frac{\dd I^{L} (\g^{\vare})}{\dd \vare}\Big|_{\vare = 0} = - \frac{2}{\pi} \, \Isch (\varphi), $$
    where $\varphi = g^{-1} \circ f |_{\m S^1}$ is the welding homeomorphism of $\g$.
\end{thm}

\begin{proof}
For $\vare \in [0,1)$, let 
$f_{\vare} (z) = f((1-\vare) z)/(1-\vare)$ such that $\g^\vare = (1-\vare) f_\vare (S^1)$. The vector field $v := \dd f_{\vare} / \dd \vare|_{\vare = 0} \circ f^{-1}$ is holomorphic  on $\O$. Since $f$ is $C^{3,\a}$ on  $\overline{\m D}$, $v$ is a holomorphic vector field on $\O$ and is $C^{2,\a}$ on $\overline{\O}$.

Since $f_\vare (0) = 0$, $f_\vare' (0) = 1$, we have $v (z) = O(z^2)$ as $z \to 0$. We now extend $v$ into a $C^2$ vector field such that $v (z) = O(z)$ as $z \to \infty$ and that $\mu : = \bar \partial v$ satisfies $\mu \in  L^2 (\O^*, \rho_{\O^*}) \cap L^\infty (\O^*)$.  In other words, $v = v_\mu$ as described in Section~\ref{subsec:var_I_L}.

Such an extension exists since we may extend the vector field $\tilde v (z): = \iota^* v (z) := v (1/z) z^2$ to the bounded domain $\iota (\O^*)$ as a $C^2$ vector field, which we also denote by $\tilde v$, such that $\tilde v(0) = 0$. We define $v$ in $\O^*$ as $\iota^* \tilde v$, then it satisfies the conditions above.
In fact, $\tilde \mu = \bar \partial \tilde v = \iota^* \mu \in  L^2 (\iota (\O^*), \rho_{\iota(\O^*)}) \cap L^\infty (\iota(\O^*))$, where 
$$\mu = \bar \partial v \qquad \text{and} \qquad \iota^* (\mu) (z) = \mu \left(\frac1z\right) \frac{z^2}{\bar z^2}.$$
As $\mu \mapsto \iota^* \mu$ preserves the $L^2$ and $L^\infty$ norms of Beltrami differentials, we also obtain this way that $\mu \in  L^2 (\O^*, \rho_{\O^*}) \cap L^\infty (\O^*)$.

Now, we apply Theorem~\ref{thm:var_I} to the vector field $v$.
Using 
\begin{align*}
    v(z) &= -w f'(w) + z \quad  \text { where }z = f(w) \text{ and }w \in \m S^1,\\
    \dd z & = f'(w) \dd w,\\
    \dd w & = \ii w \,\dd \t  \quad \text{ where } w = \ee^{\ii \t}\\
    \mc S[g^{-1}] (z) & = \mc S[\varphi] \circ f^{-1} (z)(f^{-1})'(z)^2 + \mc S[f^{-1}] (z), 
\end{align*}
we obtain
   \begin{align*} 
    \frac{\dd}{\dd \vare}\bigg|_{\vare=0} I^L (\g^{\vare}) & = \frac{\dd}{\dd \vare}\bigg|_{\vare=0} I^L \left(f_\vare (\m S^1) \right) \\
    &= -\frac{2}{\pi} \, \mathrm{Im}  \left[\int_{\partial \O^*} -w f'(w) \mc S [g^{-1}](z) \,\dd z\right] \\
    & = - \frac{2}{\pi} \, \mathrm{Im}  \left[\int_{\partial \O} w f'(w) \mc S [g^{-1}](z) \,\dd z\right] \\
    & = - \frac{2}{\pi} \, \mathrm{Im}  \left[\int_{\partial \O} w f'(w) \mc S[\varphi] \circ f^{-1} (z)(f^{-1})'(z)^2 \,\dd z\right]\\
    & = - \frac{2}{\pi} \, \mathrm{Re}  \left[\int_{0}^{2\pi} \ee^{2\ii \t}  \mc S[\varphi] (\ee^{\ii \t}) \,\dd \t\right].
\end{align*}
We used the M\"obius invariance of $I^L$ in the first equality above, Theorem~\ref{thm:var_I} and that $z \mc S[g^{-1}] (z)$ is holomorphic in $\O^*$ in the second, changed the orientation of $\g = \partial \O$ in the third, used the fact that $z\mapsto w f'(w) \mc S[f^{-1}] (z)$ is holomorphic in $\O$, and change of variables in the last equality. 

Since we knew that $\Isch (\varphi)$ is real from Lemma~\ref{lem:sch_real}, the ``real part'' above is redundant and the proof of Theorem~\ref{thm:var_IL_Sch} is completed.
\end{proof}

\section{Conformal distortion}

Here, we show that the Schwarzian action is also related to the change of area bounded by asymptotical circles under a conformal distortion. 
Let $\rho_{\m D} = 4(1-|z|^2)^{-2} |\dd z|^2$ denote the hyperbolic metric on the disk and let $\m S_r := \{ re^{\ii \theta} :\theta\in [0,2\pi]\}$ denote the standard foliation of the disk by circles. Let $K\subset \m D$ be compact, and suppose that $\varphi:\m D \setminus K \to \m D$ is locally conformal such that the analytic function $\varphi\mid_{\m S^1}$ sends the circle to itself (note however that we do not require $\varphi\mid_{\m S^1}$ to be injective).  

    Let $\rho = \varphi^\ast(\rho_{\m D})$ be the pullback of the hyperbolic metric by $\varphi$. 

\begin{thm}\label{thm:conformal_distortion}
    Let $\dd s, \dd s_{\rho}$ denote the length measure on $\m S_r$ for $\rho_{\m D},\rho$ respectively, and similarly let $k_{\rho_{\m D}},k_{\rho}$ denote the respective geodesic curvatures. Then
    \begin{equation}\label{eq:conformal_distortion}
        \lim_{r\to 1^-} \bigg[ \int_{\m S_r} k_{\rho}(s) \,\dd s_{\rho} - \int_{\m S_r} k_{\rho_{\m D}}(s)\, \dd s \bigg]= -\frac{2}{3} \int_0^{2\pi} \mc{S}[\varphi](\mathrm{e}^{\ii \theta}) \mathrm{e}^{2 \mathfrak{i} \theta} \, \dd \theta = -\frac{2}{3} \Isch(\varphi).
    \end{equation}
\end{thm}
By Remark~\ref{rem:phi_complex_derivative}, we may interpret the derivatives of $\varphi$ in the action as the complex derivative with respect to $z = \ee^{\ii \theta}$. Note that Theorem~\ref{thm:conformal_distortion} applies even when $\varphi$ is not a circle homeomorphism, e.g.\ it applies to $\varphi(z) = z^n$ for any $n\geq 1$. However, when $\varphi: \m S^1\to \m S^1$ is a homeomorphism, we have the following corollary.
\begin{cor}\label{cor:conformal_distortion_area}
    Suppose that $\varphi: \m D \to \m D$ is quasiconformal with Beltrami coefficient supported in the compact set $K\subset \m D$. Then 
    \begin{align}
        \lim_{r\to 1^-} \bigg[A(\varphi(\m S_r)) - A(\m S_r)\bigg] = -\frac{2}{3}\Isch(\varphi),
    \end{align}
    where $A(\gamma)$ denotes the hyperbolic area of the region enclosed by the curve $\gamma$.
\end{cor}
\begin{proof}
    Let $\rho=\varphi^*(\rho_{\m D})$. By the Gauss--Bonnet theorem, 
    \begin{align*} 
    2\pi &= - A(\m S_r) + \int_{\m S_r} k_{\rho_{\m D}}(s) \, \dd s\\ 
        2\pi &= - A(\varphi(\m S_r)) + \int_{\varphi(\m S_r)} k_{\rho_{\m D}}(s) \,\dd s =  - A(\varphi(\m S_r)) + \int_{\m S_r} k_{\rho}(s) \,\dd s_\rho 
    \end{align*}
Rearranging and applying Theorem~\ref{thm:conformal_distortion} gives the result.
\end{proof}

Combined with Theorem~\ref{thm:var_IL_Sch}, we also obtain another corollary. 

\begin{cor}
    Suppose that $\varphi: \m S^1\to \m S^1$ is an analytic homeomorphism that admits a quasiconformal extension $\m D\to \m D$ with Beltrami coefficient compactly supported in $\m D$. Let $\gamma$ be the solution of the welding problem for $\varphi$, and let $\gamma^{\varepsilon}$ be the equipotential foliation for $\gamma$ for some conformal map $\m D\to \Omega$, $\partial \Omega=\gamma$. Then 
    \begin{align}
        \frac{\dd}{\dd \varepsilon}\bigg\lvert_{\varepsilon=0} I^L(\gamma^{\varepsilon}) = \frac{3}{\pi} \lim_{\varepsilon\to 0} \bigg[\mathrm{Area}(\varphi(\m S_{1-\varepsilon})) - \mathrm{Area}(\m S_{1-\varepsilon}) \bigg].
    \end{align}
\end{cor}

\begin{proof}[Proof of Theorem~\ref{thm:conformal_distortion}]

Since $\varphi$ is conformal on $\m D\setminus K$, in this region the perturbed metric $\rho$ differs from $\rho_{\m D}$ by a Weyl scaling, i.e., there is $\sigma$ such that $\rho = \ee^{2\sigma} \rho_{\m D}$. Since $\rho = \varphi^*(\rho_{\m D})$, $\sigma$ is given explicitly by
\begin{align}
    \sigma(z) = \log \bigg[ |\varphi'(z)| \frac{1-|z|^2}{1-|\varphi(z)|^2}\bigg].
\end{align}
Further we have that $\dd s_{\rho} = \ee^{\sigma} \dd s$ and that the geodesic curvatures $k_{\rho}, k_{\rho_0}$ of a curve $\gamma$ are related by: 
\begin{align}\label{eq:geocurve}
    k_{\rho} = \ee^{-\sigma} (k_{\rho_{\m D}} + \partial_{n_0}\sigma),
\end{align}
where $\partial_{n_0}$ denotes the outer normal directional derivative to the curve $\gamma$. When $\gamma = \m S_r$, $\partial_{n_0} = \partial_r$ is the radial derivative. Hence,
\begin{align}
    \lim_{r\to 1^-} \bigg( \int_{\m S_r} k_{\rho}(s) \,\dd s_{\rho} - \int_{\m S_r} k_{\rho_{\m D}}(s)\, \dd s \bigg) = \lim_{r\to 1^-} \int_{\m S_r} \partial_r \sigma \, \dd s.
\end{align}
For $z\in \m S^1$, let $z_r = r z$ denote a point on $\m S_r$. Noting that $\partial_r = r^{-1} (z \partial_z + \overline{z} \partial_{\overline{z}})$, we find
\begin{align}
    \nonumber \partial_r\sigma(z_r) &= \partial_r \log|\varphi'(z_r)| + \partial_r \log (1-r^2) - \partial_r \log (1-|\varphi(z_r)|^2) \\
    &= \Re \bigg(\frac{\varphi''(z_r)}{\varphi'(z_r)}\frac{z_r}{r}\bigg) - \frac{2r}{1-r^2} + \frac{2 \Re(z_r \varphi'(z_r) \overline{\varphi(z_r)})}{r (1-|\varphi(z_r)|^2)}.
\end{align}
A simple calculation shows that for $\m S_r$, $\dd s = 2r (1-r^2)^{-1} \dd \theta$, for $\theta \in [0,2\pi]$.  With $z = \ee^{\ii \theta}, z_r = r z$, the left hand side of \eqref{eq:conformal_distortion} becomes 
\begin{align}\label{eq:integrand}
    \lim_{r\to 1^-}\frac{2r}{1-r^2} \int_0^{2\pi} \Re \bigg(\frac{\varphi''(z_r)}{\varphi'(z_r)}{z}\bigg) - \frac{2r}{1-r^2} + \frac{2 \Re(z \varphi'(z_r) \overline{\varphi(z_r)})}{(1-|\varphi(z_r)|^2)}\, \dd \theta.
\end{align}
For $\theta$ fixed, we expand the integrand around $z$. Note that $z_r-z = -(1-r) z$. It is further important to note that all three terms in the integrand are real-valued. For the first term, we have 
  \begin{align*}
\Re\bigg(\frac{\varphi''(z_r)}{\varphi'(z_r)}z\bigg) & = \Re\bigg(\frac{\varphi''(z)}{\varphi'(z)}z\bigg) - \bigg( \Re\bigg(\frac{\varphi'''(z)}{\varphi'(z)}z^2\bigg) - \Re\bigg(\frac{\varphi''(z)}{\varphi'(z)}z\bigg)^2 \bigg) (1-r) \\
& \quad + O((1-r)^2).
     \end{align*}
For the third term, we expand the numerator and denominator and then their ratio, using that both are real. To simplify both expressions, we repeatedly use that 
\begin{equation}\label{eq:phi_fact}
    |\varphi'(z)| = \frac{\varphi'(z)}{\varphi(z)} z \qquad \forall z\in \m S^1
\end{equation}
which follows from the fact that $\varphi$ fixes $\m S^1$. 
For the denominator of the third term, we compute 
\begin{align*}
    \nonumber 1 - |\varphi(z_r)|^2 = &|\varphi'(z)| \bigg[ 2(1-r) - \bigg(\frac{\varphi'(z)}{\varphi(z)}z +\Re \bigg( \frac{\varphi''(z)}{\varphi'(z)} z\bigg)\bigg)(1-r)^2
    \\
    &+ \bigg( \frac{1}{3} \Re \bigg(\frac{\varphi'''(z)}{\varphi'(z)} z^2\bigg) + \Re\bigg( \frac{\varphi''(z)}{\varphi(z)} z^2\bigg)\bigg) (1-r)^3 + O((1-r)^4)
    \bigg].
\end{align*}
For the numerator, we find
\begin{align*}
   \nonumber 2 \Re (z \varphi'(z_r) \overline{\varphi(z_r)}) = &|\varphi'(z)|\bigg[2 - \bigg( 2\frac{\varphi'(z)}{\varphi(z)}z + 2\,\Re\bigg(\frac{\varphi''(z)}{\varphi'(z)}z\bigg)\bigg)(1-r) \\
   &+ \bigg( \Re\bigg(\frac{\varphi'''(z)}{\varphi'(z)} z^2\bigg) + 3 \,\Re\bigg(\frac{\varphi''(z)}{\varphi(z)}z^2\bigg)\bigg)(1-r)^2 \bigg] + O((1-r)^3).
\end{align*}
Let $x = (1-r)$, $a=  \frac{\varphi'(z)}{\varphi(z)}z + \Re(\frac{\varphi''(z)}{\varphi'(z)}z)$ and $b = \frac{1}{3} \Re(\frac{\varphi'''(z)}{\varphi'(z)} z^2) + \Re(\frac{\varphi''(z)}{\varphi(z)}z^2)$.  With this notation, we have that 
     \begin{align*}
         \frac{2 \Re (z \varphi'(z_r) \overline{\varphi(z_r)})}{ 1-|\varphi(z_r)|^2} = \frac{2 - 2a x + 3b x^2 + o(x^2)}{2x - a x^2 + b x^3 + o(x^3)}.
     \end{align*}
     As the coefficients $a, b$ are real-valued for all $z\in \m S^1$ it follows that 
     \begin{align*}
         \frac{2 \Re (z \varphi'(z_r) \overline{\varphi(z_r)})}{ 1-|\varphi(z_r)|^2} = \frac{1}{x} - \frac{a}{2} + \frac{1}{4}\bigg(4 b - a^2\bigg) x + o(x).
     \end{align*}
 Note that $a^2 = \bigg(\frac{\varphi'(z)}{\varphi(z)}z\bigg)^2 + \bigg(\Re \frac{\varphi''(z)}{\varphi'(z)}z\bigg)^2 + 2 \Re\bigg(\frac{\varphi''(z)}{\varphi(z)} z^2\bigg)$. Hence Plugging in $x,a,b$ gives
 \begin{align*}
       \nonumber &\frac{2 \Re (z \varphi'(z_r) \overline{\varphi(z_r)})}{ 1-|\varphi(z_r)|^2} =   \frac{1}{1-r} - \frac{1}{2}\bigg(\frac{\varphi'(z)}{\varphi(z)}z + \Re\bigg(\frac{\varphi''(z)}{\varphi'(z)}z\bigg)\bigg)  \\
       &+\bigg( \frac{1}{3} \Re\bigg(\frac{\varphi'''(z)}{\varphi'(z)} z^2\bigg) + \frac{1}{2}\Re\bigg(\frac{\varphi''(z)}{\varphi(z)}z^2\bigg) - \frac{1}{4}\bigg(\frac{\varphi'(z)}{\varphi(z)}z\bigg)^2 - \frac{1}{4}\bigg(\Re \frac{\varphi''(z)}{\varphi'(z)}z\bigg)^2 \bigg) (1-r)\\ & + O((1-r)^2).
 \end{align*}
    Collecting terms by order in $(1-r)$, the integrand in \eqref{eq:integrand} is equal to
     \begin{align*}
         &-\frac{2r}{1-r^2} + \frac{1}{1-r}\\
         &+\Re\bigg(\frac{\varphi''(z)}{\varphi'(z)}z\bigg)-\frac{1}{2}\frac{\varphi'(z)}{\varphi(z)}z -\frac{1}{2} \Re\bigg(\frac{\varphi''(z)}{\varphi'(z)}z\bigg)\\ 
         &+\bigg[ \frac{1}{3} \Re\bigg(\frac{\varphi'''(z)}{\varphi'(z)} z^2\bigg) + \frac{1}{2}\Re\bigg(\frac{\varphi''(z)}{\varphi(z)}z^2\bigg) - \frac{1}{4}\bigg(\frac{\varphi'(z)}{\varphi(z)}z\bigg)^2 - \frac{1}{4}\bigg(\Re \frac{\varphi''(z)}{\varphi'(z)}z\bigg)^2 \\
         &-\Re \bigg(\frac{\varphi'''(z)}{\varphi'(z)}z^2\bigg) + \Re\bigg(\bigg(\frac{\varphi''(z)}{\varphi'(z)}z\bigg)^2\bigg) \bigg](1-r) + O((1-r)^2).
     \end{align*}
    For the first line, we note
    \begin{align*}
        -\frac{2r}{1-r^2} + \frac{1}{1-r} = \frac{1}{1+r}.
    \end{align*}
    By \eqref{eq:phi_fact}, $\log |\varphi'(z)|$ can be extended to a holomorphic function in a neighborhood of $\m S^1$. As such $\partial_\theta = \mathfrak{i} z\partial_z$ and we get the relation
    \begin{align}\label{eq:theta_derivative}
        \partial_\theta \log |\varphi'(z)| = \mathfrak{i} \bigg(\frac{\varphi''(z)}{\varphi'(z)} z + 1 - \frac{\varphi'(z)}{\varphi(z)}z\bigg).
    \end{align}
   Therefore, 
    \begin{align*}
        \int_0^{2\pi}\bigg(\frac{1}{2}\Re\bigg(\frac{\varphi''(z)}{\varphi'(z)}z\bigg)-\frac{1}{2}\frac{\varphi'(z)}{\varphi(z)}z\bigg)\, \dd \theta = - \pi. 
    \end{align*}
    Hence, the first two lines together contribute $\frac{1}{2(1+r)} (1-r)$ to the integrand, and all remaining terms are of order $(1-r)$. Collecting like terms again, up to higher order terms, the integrand in \eqref{eq:integrand} is $(1-r)$ times:
    \begin{align*}
       &  \frac{1}{2(1+r)} + \frac{1}{2}\Re\bigg(\frac{\varphi''(z)}{\varphi(z)}z^2\bigg) - \frac{1}{4}\bigg(\frac{\varphi'(z)}{\varphi(z)}z\bigg)^2 - \frac{1}{4}\bigg(\Re \frac{\varphi''(z)}{\varphi'(z)}z\bigg)^2 \\
         &-\frac{2}{3}\Re \bigg(\frac{\varphi'''(z)}{\varphi'(z)}z^2\bigg) + \Re\bigg(\bigg(\frac{\varphi''(z)}{\varphi'(z)}z\bigg)^2\bigg).
    \end{align*}
    We note that 
    \begin{align*}
        -\frac{2}{3}\frac{\varphi'''(z)}{\varphi'(z)}z^2 + \bigg(\frac{\varphi''(z)}{\varphi'(z)}z\bigg)^2 = -\frac{2}{3} \mc{S}[\varphi](z) z^2 \in \m R. 
    \end{align*}
    In particular, this is already real-valued. We make two observations to show that the remaining terms contribute zero to the integral. First, we note that 
    \begin{align*}
        \partial_\theta \bigg( \frac{\varphi'(z)}{\varphi(z)} z \bigg) = \mathfrak{i} \bigg( \frac{\varphi''(z)}{\varphi(z)} z^2 + \frac{\varphi'(z)}{\varphi(z)}z - \bigg(\frac{\varphi'(z)}{\varphi(z)} z\bigg)^2\bigg).
    \end{align*}
    Combined with the expression above for $\partial_\theta \log |\varphi'(z)|$ in \eqref{eq:theta_derivative}, we infer that 
    \begin{align*}
        \int_0^{2\pi} \frac{1}{4} - \frac{1}{4}\bigg(\frac{\varphi'(z)}{\varphi(z)} z\bigg)^2 + \frac{1}{4} \frac{\varphi''(z)}{\varphi(z)} z^2 + \frac{1}{4} \frac{\varphi''(z)}{\varphi'(z)}z \; \dd \theta = 0.
    \end{align*}
    Therefore under the integral,
    \begin{align*}
         &\int_0^{2\pi} \frac{1}{2(1+r)} + \frac{1}{2}\Re\bigg(\frac{\varphi''(z)}{\varphi(z)}z^2\bigg) - \frac{1}{4}\bigg(\frac{\varphi'(z)}{\varphi(z)}z\bigg)^2 - \frac{1}{4}\bigg(\Re \frac{\varphi''(z)}{\varphi'(z)}z\bigg)^2 \, \dd \theta \\
         =&\int_{0}^{2\pi} \frac{1}{4} \Re\bigg(\frac{\varphi''(z)}{\varphi'(z)}z\bigg) \bigg( \frac{\varphi'(z)}{\varphi(z)}z - \Re\bigg(\frac{\varphi''(z)}{\varphi'(z)}z\bigg) -1 \bigg) + \frac{1}{4(1+r)}(1-r)\, \dd \theta.
    \end{align*}
    And in sum we have reduced the right side of \eqref{eq:conformal_distortion} to the $r\to 1^-$ limit of:
    \begin{align*}
     \frac{2r}{1+r} \int_0^{2\pi} \frac{1}{4} \Re\bigg(\frac{\varphi''(z)}{\varphi'(z)}z\bigg) \bigg( \frac{\varphi'(z)}{\varphi(z)}z - \Re\bigg(\frac{\varphi''(z)}{\varphi'(z)}z\bigg) -1 \bigg) - \frac{2}{3} \mc{S}[\varphi](z) z^2 + O(1-r) \, \dd \theta 
    \end{align*}
    Finally we note that $\log |\varphi'(z)|$ is a function from $\m S^1\to \m R$, hence, $\partial_\theta \log |\varphi'(z)|\in \m R$. As such by \eqref{eq:theta_derivative},
    \begin{align*}
         \frac{\varphi'(z)}{\varphi(z)}z - \frac{\varphi''(z)}{\varphi'(z)}z -1 \in \mathfrak{i} \m R \qquad \forall z\in \m S^1,
    \end{align*}
    and hence $ \frac{\varphi'(z)}{\varphi(z)}z - \Re\bigg(\frac{\varphi''(z)}{\varphi'(z)}z\bigg) -1 = 0$ for all $z\in \m S^1$. This completes the proof.
\end{proof}

\section{Further discussions} \label{sec:discussions}

\subsection{Dual Epstein curves in de Sitter space}\label{subsec:dS}

Using the duality between hyperbolic and de Sitter spaces (see, for instance, \cite{Andreev70}, \cite{HodgsonRivin93}, \cite{Schlenker93}), we can describe the dual Epstein map in de Sitter space, which we now summarize briefly.

In Minkowski space $\mathbb{R}^{1,n}$ we consider $q(x) = -x_0^2+\ldots+x_{n-1}^2+x_{n}^2$, the canonical quadratic form of signature $(-,+,\ldots,+)$. Here, hyperbolic space $\mathbb{H}^n$ and de Sitter space $\dS_n$ are represented by the submanifolds $\lbrace x=  (x_0,\ldots, x_{n}) \,|\, q(x)=-1, x_{0}>0 \rbrace$ and $\lbrace x=  (x_0,\ldots, x_{n}) \,|\, q(x)=1 \rbrace$ (respectively) with the metrics induced by the Lorentzian ambient space $\m R^{1,n}$. Totally geodesic submanifolds of dimension $k$ of $\m H^n$ and $\dS_n$ are given by their intersections with $k+1$ linear subspaces of $\m R^{1,n}$. Moreover, an orientation in a $(k+1)$-space $V^{k+1}\subseteq\m R^{1,n}$ is in one-to-one correspondence with an orientation on its intersection $P$ with $\m H^n$ or $\dS_n$. More precisely, for $x\in V$ with $q(x)=\pm 1$, which is identified with a point in $\m H^n$ or $\dS_n$, an ordered basis $v_1,\ldots,v_k$ in its orthogonal complement $x^\perp\subset V$ is positive if and only if $x,v_1,\ldots,v_k$ is positive in $V$. This choice considers the position vector as the outer-normal to $\m H^n$ or $\dS_n$ in $\mathbb{R}^{1,n}$. In particular, this choice gives an orientation of the tangent space of $P$ at $x$, and $\m H^n$ and $\dS_n$ their standard orientations.

The duality is a one-to-one correspondence between $k$-planes in $\m H^n$ and $(n-k-1)$-planes in $\dS_n$. In fact, given an unoriented $k$-plane $P=\m H^n\cap V^{k+1}$ we define its \emph{dual} $P^\perp$ as the unoriented $(n-k-1)$-plane $P^\perp = \dS_n \cap (V^{k+1})^\perp$, where $(V^{k+1})^\perp$ is the $(n-k)$-dimensional orthogonal complement of $V^{k+1}$ in $\m R^{1,n}$. Observe that by the same process we can define the dual of unoriented $k$-planes in $\dS_n$ as unoriented $(n-k-1)$-planes in $\m H^n$. It is easy to verify that this process is an involution, and if $P\subset Q$ are totally geodesic planes in either $\m H^n$ or $\dS_n$, then $P^\perp\supset Q^\perp$. Moreover, we can also define duality between oriented planes, under the convention that $V\oplus V^\perp$ is a positive basis in $\m R^{1,n}$, where $V$ defines the plane in $\m H^n$ while $V^\perp$ defines the dual in $\dS_n$.

We now come back to the $n = 2$ case. Observe that given $h$ a $C^2$ metric in $\m S^1 = \partial_\infty \m H^2$, the Epstein curve $\Ep_h$ defines a natural map into the set of oriented geodesics in $\m H^2$. Indeed, we can consider the geodesic line at $\Ep_h$ orthogonal to $\widetilde{\Ep}_h$ (defined in \eqref{eq:tilde_Ep} and \eqref{eq:Epsteinnormal}) oriented so that the pair $(v,\widetilde{\Ep}_h)$ (where $v$ is the velocity of the geodesic at $\Ep_h$) is positively oriented in $\m H^2$. 
In this way, $\widetilde{\Ep}_h$ can be viewed as the $(2-1-1)$-plane (namely, point) in $\dS_2$  dual to the geodesic in $\m H^2$ spanned by $v$, and the Epstein map $\m S^1 \to T^1 \m H^2$ is an immersion $\m S^1 \to \m H^2 \times \dS_2$.

By definition of the duality, the oriented geodesic spanned by $v$ is the dual of $\widetilde{\Ep}_h$ and we denote it as $\widetilde{\Ep}_h^\perp$. Under our orientation conventions, the basis $\Ep_h, \widetilde{\Ep}_h^\perp, \widetilde{\Ep}_h$ has the standard orientation of $\mathbb{R}^3$.

\begin{definition}
    Given $h$ a $C^2$ metric in $\m S^1$, we define its \emph{dual Epstein curve} in $\dS_2$ as $\widetilde{\Ep}_h $. 
\end{definition}

\begin{lem}\label{lem:dual_tangent}
    Given $h$ a $C^2$ metric in $\m S^1$ and any $\theta\in \m S^1$, the tangent to the dual Epstein curve at $\widetilde{\Ep}_h(e^{\ii\t}) \in \dS_2$ is parallel to $\widetilde{\Ep}_h^\perp(e^{\ii\t})$.
\end{lem}

\begin{proof}
    As the dual Epstein curve is given by $\widetilde{\Ep}_h$, it suffices to prove that $\partial_\theta \widetilde{\Ep}_h$ is orthogonal with respect to the quadratic form $q$ to $\Ep_h$ and $\widetilde{\Ep}_h$. Denote by $\langle,\rangle$ the Lorentzian inner product associated to $q$.
    
    As $q(\widetilde{\Ep}_h)\equiv 1$ it follows that $\langle \partial_\theta \widetilde{\Ep}_h, \widetilde{\Ep}_h\rangle =0$. On the other hand, as $\widetilde{\Ep}_h^\perp$ is parallel to $\partial_\theta \Ep_h$ it follows that
    \[
    \langle \partial_\theta \widetilde{\Ep}_h, \Ep_h\rangle = -\langle \widetilde{\Ep}_h, \partial_\theta \Ep_h \rangle = 0
    \]
    which concludes the proof.
\end{proof}
By the previous Lemma we have that the tangent and normal lines of the dual Epstein curve $\widetilde{\Ep}_h$ are given by $\widetilde{\Ep}_h^\perp$ and $\Ep_h$, as at every point of the dual curve the position, tangent and normal vectors should be mutually orthogonal vectors in $\mathbb{R}^{1,2}$ with norms equal to $\pm 1$. In particular it follows from Lemma \ref{lem:dual_tangent} that the dual Epstein curve is space-like, as its tangent vector is parallel to $\widetilde{\Ep}_h^\perp$ with $q(\widetilde{\Ep}_h^\perp)=1$.

Given Lemma \ref{lem:dual_tangent}, we make the following choice of unit tangent and unit normal for the dual Epstein curve:

\begin{definition}
     Given $h$ a $C^2$ metric in $\m S^1$. At any point $\widetilde{\Ep}_h(e^{\ii\t})$ we denote by $-\widetilde{\Ep}^\perp_h(e^{\ii\t})$ and $\Ep_h (e^{\ii\t})$ as the unit tangent and unit normal of the curve, respectively.
\end{definition}

Finally, we will explain how to relate the total signed length and the total geodesic curvature for the Epstein curve in $\m H^2$ and its dual in $\dS_2$. Similar results are well-known in arbitrary dimension, we include them here as they are simple to describe explicitly in this special case.

For the Epstein curve in $\m H^2 \subset \m R^{1,2}$, denote $x(\t) : = \Ep_h(e^{\ii\t})$, $T(\t) : = \widetilde{\Ep}_h^\perp(e^{\ii\t})$, $N(\t) : = \widetilde{\Ep}_h(e^{\ii\t})$, which are respectively the position, oriented tangent and unit normal. As the geometry of $\m H^2$ is induced by the geometry of $\m R^{1,2}$, it follows that we can compute the signed length $\dd\ell$ and its product with its geodesic curvature $k\,\dd\ell$ by
\begin{align*}
    \dd\ell &= \langle x', T\rangle,\\
    k\,\dd\ell &= \langle -N', T\rangle
\end{align*}
where $'$ is the standard derivative in $\m R^{1,2}$ and $\langle \cdot, \cdot \rangle$ is the Lorentzian metric defined by the quadratic form $q$.

For the dual Epstein curve in $\dS_{2}$, our conventions for the orientation of position, unit tangent, and unit normal are given by $N$, $-T$, and $x$, respectively. It also follows that the signed length $\dd\ell^\perp$ and its product with its geodesic curvature $k^\perp\dd\ell^\perp$ are given by
\begin{align*}
    \dd\ell^\perp &: = \langle N', -T\rangle = k\,\dd\ell,\\
    k^\perp\dd\ell^\perp &: = 
    \langle -x', -T\rangle
    = \dd\ell.
\end{align*}
Hence, by Theorem~\ref{thm:SchArea} we can compute the Schwarzian action in terms of the total length and total curvature of the dual Epstein curve.
\begin{prop}\label{prop:dualISch}
    Let $\varphi : \m S^1 \to \m S^1$ be a $C^3$ diffeomorphism. Let $h = \varphi_* \dd \t$ be the pushforward of $\dd \t$ by $\varphi$. Let $k^\perp(\Ep_h)$ denote the total geodesic curvature of the dual Epstein curve associated with $h$ and $L^\perp(\Ep_h)$ the signed de Sitter length of the dual Epstein curve, then we have
    \begin{align*}
        \Isch(\varphi) = k^\perp(\Ep_h) = 2\pi-L^\perp(\Ep_h).
    \end{align*}
\end{prop}

\subsection{Piecewise M\"obius circle homeomorphisms} \label{subsec:pw}

In this section, we show an extension of Theorem~\ref{thm:Ep_Sch} to circle homeomorphisms of lower regularity, in the special case of piecewise M\"obius circle diffeomorphisms. 

\begin{definition}
A map $\varphi:\m S^1\to \m S^1$ is a $C^1$ \emph{piecewise M\"obius} circle diffeomorphism, if $\varphi$ is continuously differentiable and there exists $z_1, \ldots, z_n \in \m S^1$ in counterclockwise order, such that the restriction of $\varphi$ to the circular arc $[z_j, z_{j+1}]$ is in $\PSU(1,1)$. We call the points $z_1,\dots,z_n$ the \textit{breakpoints} of $\varphi$. 
\end{definition}

We note that a $C^1$ piecewise M\"obius circle diffeomorphism is automatically $C^{1,1}$, that is, its derivative is Lipschitz.  
This class of circle diffeomorphisms has been considered as a finite-dimensional approximation of the universal Teichm\"uller space, and can be described using finite shears and diamond shears \cite{SWW_shears,Penner1993UniversalCI}. They are also related to various optimization problems for the Loewner energy \cite{RW,BJRW_pwg,Wang_optimization}.

We can extend the definition of the Schwarzian derivative to piecewise M\"obius circle diffeomorphisms using the distributional derivative. For $z=\ee^{\ii \theta}\in \m S^1$, recall that $\varphi'(z)$ is defined so that $\ii z \varphi'(z) = \varphi_\theta(z)$, where $\varphi_\theta$ is the derivative of $\varphi$ with respect to the argument of $z$. The Schwarzian is $\mc S[\varphi](z) = (\varphi''/\varphi'(z))' - (1/2) (\varphi''/\varphi')^2$.

\begin{lem}\label{lem:schwarzian_piecewise}
Suppose that $\varphi$ is $C^1$ piecewise M\"obius with breakpoints $z_1,\dots,z_n$ in counterclockwise order. In the sense of distributions, 
  \begin{align*}
         \ee^{2\ii\theta} \mc S[\varphi](\ee^{\ii \theta}) = \sum_{j=1}^n \lambda_j  \delta_{z_j},
    \end{align*}
    where 
    $$\lambda_j = -\ii z_j \left(\frac{\varphi''}{\varphi'}(z_j+) - \frac{\varphi''}{\varphi'}(z_j-)\right) = \frac{\varphi_{\theta \theta}}{\varphi_{\theta}}(z_j-) - \frac{\varphi_{\theta\theta}}{\varphi_{\theta}}(z_j+)\in \m R.$$ 
    We therefore define the Schwarzian action of $\varphi$ by 
    \begin{align*}
    	\Isch(\varphi) = \sum_{j=1}^n \lambda_j. 
    \end{align*}
    \end{lem}
   \begin{proof}
         Let $\theta_j$ denote the argument of $z_j$. First note that since $\partial_\theta = \ii z \partial_z$, 
         $$ \frac{\varphi''}{\varphi'}(z) = \frac{1}{\ii z}\frac{\varphi_{\theta\theta}}{\varphi_\theta} -\frac{1}{z}.$$
         Since $\varphi$ is piecewise M\"obius, for any test function $g$ and any $\varepsilon>0$ sufficiently small, 
         \begin{align*}
              \int_{0}^{2\pi} g(\theta) \ee^{2\ii \theta} \mc S[\varphi](\ee^{\ii \theta}) \,\dd \theta &= \sum_{j=1}^n \int_{\theta_j-\varepsilon}^{\theta_j+\varepsilon} g(\theta) e^{2\ii \theta}\mc S[\varphi](\ee^{\ii \theta}) \, \dd \theta\\
              &=\sum_{j=1}^n \int_{\theta_j-\varepsilon}^{\theta_j+\varepsilon} g(\theta) \bigg( - \frac{\dd}{\dd \theta}\bigg( \frac{\varphi_{\theta\theta}}{\varphi_{\theta}}(\ee^{\ii \theta})\bigg) + \frac{1}{2}\bigg( \frac{\varphi_{\theta\theta}}{\varphi_\theta}(\ee^{\ii \theta})\bigg)^2 + \frac12  \bigg)\,\dd\theta  
         \end{align*}

          Integrating by parts, it follows that 
    \begin{align*}
       &\int_{\theta_j-\varepsilon}^{\theta_j+\varepsilon} g(\theta) \bigg( - \frac{\dd}{\dd \theta}\bigg( \frac{\varphi_{\theta\theta}}{\varphi_{\theta}}(\ee^{\ii \theta})\bigg) + \frac{1}{2}\bigg( \frac{\varphi_{\theta\theta}}{\varphi_\theta}(\ee^{\ii \theta})\bigg)^2 + \frac12  \bigg)\,\dd\theta  \\
       =& -g(\theta) \frac{\varphi_{\theta\theta}}{\varphi_\theta}(\ee^{\ii \theta})\bigg\lvert_{\theta=\theta_j-\vare}^{\theta=\theta_j+\vare} + \int_{\theta_j-\vare}^{\theta_j+\vare} g_\theta(\theta) \frac{\varphi_{\theta\theta}}{\varphi_\theta}(\ee^{\ii \theta}) + g(\theta) \bigg( \frac{1}{2}\bigg( \frac{\varphi_{\theta\theta}}{\varphi_\theta}(\ee^{\ii \theta})\bigg)^2 + \frac12 \bigg) \,\dd \theta.
    \end{align*}
    Since $\varphi_{\theta\theta}/\varphi_{\theta} = \ii z (\varphi''/\varphi')$ is piecewise smooth, taking $\varepsilon\to 0$ gives
    \begin{align*}
        \int_{0}^{2\pi} g(\theta) \ee^{2 \ii \theta}\mc S[\varphi](\ee^{\ii \theta}) \,\dd \theta = \sum_{j=1}^n g(\theta_j) \bigg( \frac{\varphi_{\theta\theta}}{\varphi_{\theta}}(\ee^{\ii \theta_j}-) - \frac{\varphi_{\theta\theta}}{\varphi_{\theta}}(\ee^{\ii \theta_j}+)\bigg).
    \end{align*}
    This completes the proof since $g$ is arbitrary.
    \end{proof}

For $\varphi$ which is $C^1$ piecewise M\"obius with breakpoints $z_1,\dots,z_n$, consider the metric $h = \varphi_* \dd \theta$. The horocycles $(H_{\varphi(z)})_{z\in \m S^1}$ corresponding to $h$ are well-defined since $\varphi$ is $C^1$, but since $\varphi$ is not $C^2$, the image of $\Ep_h$ in $\m D$ is disconnected. These both, however, have simple explicit descriptions, from which we can see that there is a natural way to construct an Epstein curve $\Sigma_h$ which extends the image of $\Ep_h$. 

Note that in previous sections, where we work with $C^2$ metrics, we have interchangeably used $\Ep_h$ to denote both the map and its connected image curve $\Sigma_h$. In this section we use $\Ep_h$ to refer specifically to the map and its image, and $\Sigma_h$ to refer to the curve we construct which extends the image of $\Ep_h$.

To construct $\Sigma_h$, let $(H_z^0)_{z\in \m S^1}$ denote the horocycles for the identity map and let $\alpha_1,\dots,\alpha_n\in \PSU(1,1)$ be the M\"obius transformations such that $\varphi\mid_{[z_j,z_{j+1}]}=\alpha_j$. By Lemma~\ref{lem:naturality}, we have
\begin{itemize}
	\item for any $z\in [z_j,z_{j+1}]$, $H_{\varphi(z)} = \alpha_j(H_z^0)$;
	\item the image of $\Ep_h$ is the collection of points $\{a_1,\dots, a_n\}\subset \m D$, where $a_j = \alpha_j(0)$. 
\end{itemize}
Since $\varphi$ is $C^1$, for each breakpoint $z_j$, the horocycle $H_{\varphi(z_j)} \ni \{a_{j-1}, a_{j}\}$. We define $\Sigma_h^j$ to be the arc of $H_{\varphi(z_j)}$ from $a_{j-1}$ to $a_{j}$ (where we identify $a_{-1}=a_n$), and define the Epstein curve for $\varphi$ to be $$\Sigma_h:= \Sigma_h^1\cup \dots\cup \Sigma_h^n.$$ 

See Figure~\ref{fig:piecewise_mobius} for an example. We parametrize $\Sigma_j$ by arclength, oriented from $a_{j-1}$ to $a_j$, and let $T_j$ denote its unit tangent vector. Let $N_j$ be the outward pointing unit normal vector to $H_{\varphi(z_j)}\supset \Sigma_h^j$, and let $\beta_j$ be the angle between $N_j,N_{j+1}$ at $a_j$ (we describe these in Lemma~\ref{lem:angle_sum} below). We define the signed arclength $\dd \ell_j$ along $\Sigma_h^j$ to be standard arclength times the sign of the orientation of the basis $(T_j,N_j)$. 
In particular, if $a_{j}$ is clockwise after $a_{j-1}$ along $H_{\varphi(z)}$ then the sign is positive (respectively if $a_j$ is counterclockwise after $a_{j-1}$ along $H_{\varphi(z)}$ then the sign is negative).

\begin{figure}
    \centering
    \includegraphics[width=0.4\linewidth]{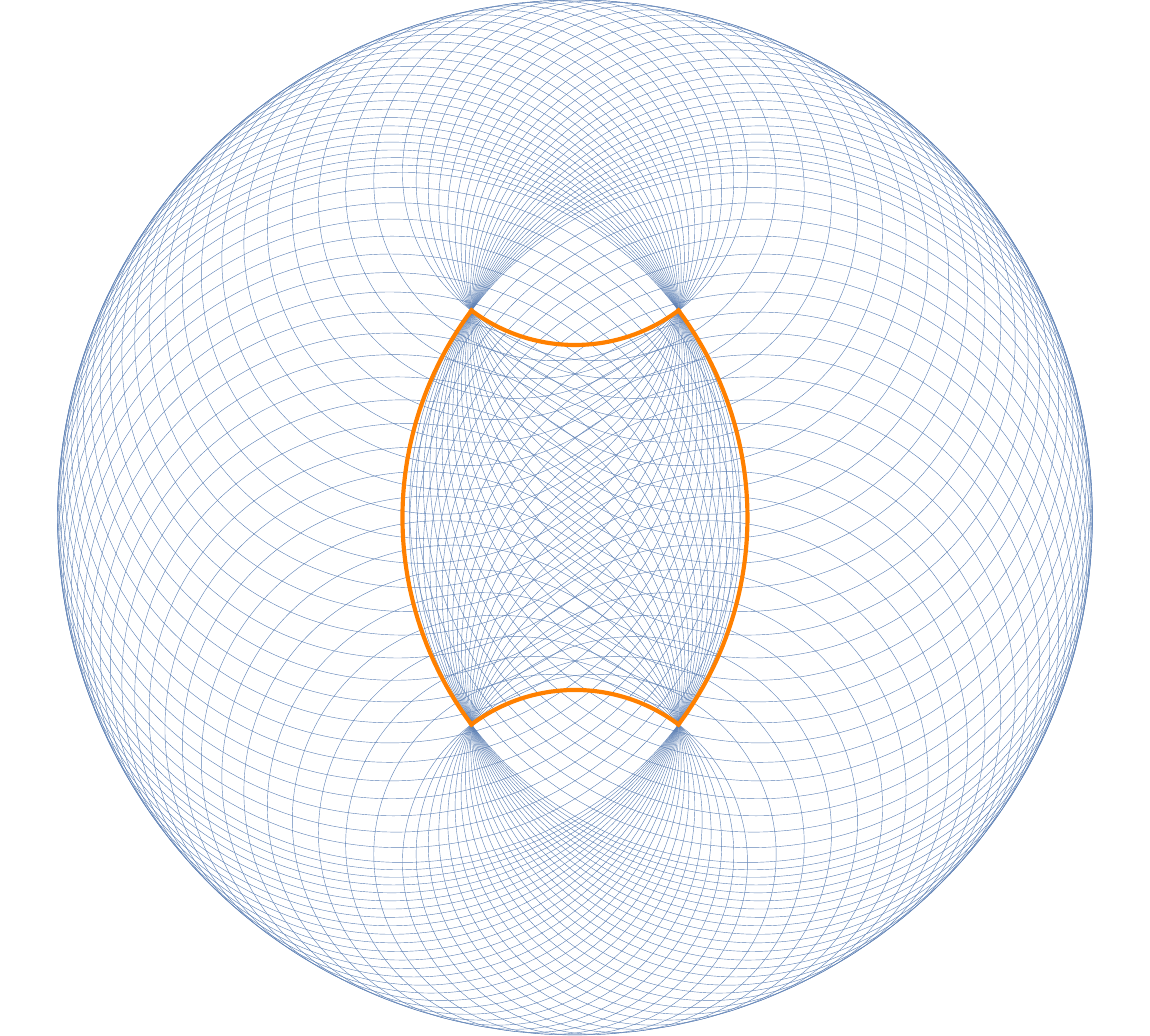}\includegraphics[width=0.4\linewidth]{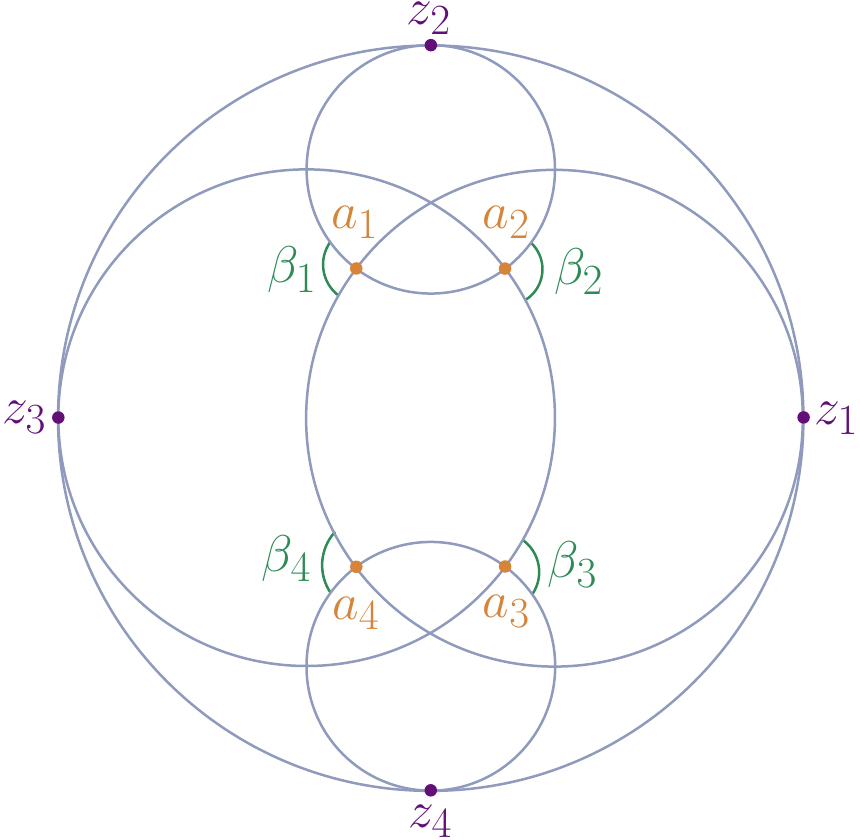}
    \caption{Let $h=\varphi_*\dd \theta$, $\varphi$ piecewise M\"obius with breakpoints $1,\ii,-1,-\ii$. Here are its horocycles in increments of $\pi/25$ and completed Epstein curve $\Sigma_h$ (left) and its horocycles at $z_j=1,\ii,-1,-\ii$, the intersection points $a_j = \alpha_j(0)$, and the angles $\beta_j$ between the outward normal vectors to the horocycles at $a_j$ (right).}
    \label{fig:piecewise_mobius}
\end{figure}

We now give a geometric interpretation of the ``jumps'' $\lambda_j$ at the breakpoints of $\varphi$, which appear in Lemma~\ref{lem:schwarzian_piecewise}, as the signed length of the horocyclic arcs $\Sigma_h^j$. 

\begin{lem}\label{lem:horocycle_arclengths}
    Let $\dd \ell_j$ denote the signed arclength measure along $\Sigma_h^j$. The total signed length is $$L(\Sigma_h^j) = \int_{\Sigma_h^j} \dd \ell_j =\lambda_j.$$ 
\end{lem}
\begin{proof}
We compute the signed hyperbolic distance along $H_{\varphi(z_j)}$ from $a_{j-1}=\alpha_{j-1}(0)$ to $a_{j}=\alpha_j(0)$; note that this is M\"obius invariant. Define $\alpha:= \alpha_{j-1}^{-1}\circ \alpha_j$. By Lemma~\ref{lem:naturality}, it therefore suffices to compute the signed hyperbolic distance along $H_{z_j}^0$ from $0$ to $\alpha(0)$. By construction, $\alpha(z_j)=z_j$, $\alpha'(z_j) = 1$, and $-\ii z_j \alpha''(z_j) = \lambda_j$. These three conditions imply that $\alpha$ is determined uniquely by $z_j,\lambda_j$ and given by the formula:
\begin{align*}
     \alpha(z) = \frac{(2\ii+\lambda_j)z_j z-z_j^2 \lambda_j}{\lambda_j z + (2\ii -\lambda_j) z_j}.
\end{align*}
Let $\mathfrak c_j :\m D \to \m H$ be the Cayley map
\begin{align*}
    \mathfrak c_j(z) = \ii \frac{1+z/z_j}{1-z/z_j}.
\end{align*}
This Cayley map sends $H_{z_j}^0$ to the horocycle at $\infty$ through $\ii=\mathfrak c_j(0)$, i.e., the horizontal line $\{y=1\}$, and $\Sigma_h^j$ to the segment of the horizontal line from $\ii$ to $\mathfrak c_j(\alpha(0))$. Thus,
\begin{align*}
    L(\Sigma_j) = -\text{Re}(\mathfrak c_j(\alpha(0))).
\end{align*}
The sign comes from considering orientation; the outward normal vector to the horocycle at $\infty$ is parallel to $(0,-1)$, so the signed arclength is positive when the horocycle is oriented clockwise. We compute 
\begin{align*}
     \mathfrak c_j\circ \alpha(z) = \frac{-z(\lambda_j+\ii)-z_j(\ii -\lambda_j)}{z-z_j}
\end{align*}
from which we see that $\mathfrak c_j(\alpha(0)) = \ii - \lambda_j$. Therefore, $L(\Sigma_h^j) = \lambda_j$. 
\end{proof}

\begin{lem}\label{lem:angle_sum}
    Let $\varphi$ be $C^1$ piecewise M\"obius with breakpoints $z_1,\dots,z_n$. The horocycles $H_{\varphi(z_j)},H_{\varphi(z_{j+1})}$ intersect at $a_j$, and the angle $\beta_j$ between their outward pointing normal vectors $N_j,N_{j+1}$ at $a_{j}$ (see Figure~\ref{fig:piecewise_mobius}) equals the Lebesgue measure of the circular arc $[z_j, z_{j+1}]$.
\end{lem}
\begin{proof}
	Let $\alpha_j=\varphi\mid_{[z_j,z_{j+1}]} \in \PSU(1,1)$. Since $\varphi$ is $C^1$, $\varphi'(z_j) = \alpha_j'(z_j)$ and $\varphi'(z_{j+1})=\alpha_j'(z_{j+1})$. Hence, by Lemma~\ref{lem:naturality}, $H_{\varphi(z_j)} = \alpha_j(H_{z_j}^0)$ and $H_{\varphi(z_{j+1})} = \alpha_j(H_{z_{j+1}}^0)$. Since $H_{z_j}^0, H_{z_{j+1}}^0$ intersect at $0$, these intersect at $a_j=\alpha_j(0)$. Since $\alpha_j$ is conformal, the angle of intersection between $N_j,N_{j+1}$ at $a_j$ is the same as the angle of between the normal vectors $N_j^0,N_{j+1}^0$ to $H_{z_j}^0, H_{z_{j+1}}^0$ respectively at $0$. These are the horocycles based at $z_j,z_{j+1}$ through $0$ respectively, and hence the angle between their outward normal vectors at $0$ is given by the arclength of $[z_j,z_{j+1}]$. 
\end{proof}

We combine the above with the Gauss--Bonnet theorem to obtain an analog of Theorem~\ref{thm:Ep_Sch} for $C^1$ piecewise M\"obius circle diffeomorphisms.
\begin{cor}
    Let $\varphi$ be a $C^1$ piecewise M\"obius circle diffeomorphism with breakpoints $z_1,\dots,z_n\subset \m S^1$ in counterclockwise order, and define the metric $h = \varphi_* \dd \theta$ and its completed Epstein curve $\Sigma_h$ as above. Then 
    \begin{align*}
    	\Isch(\varphi) = L(\Sigma_h) = -A(\Sigma_h). 
    \end{align*}
\end{cor}
\begin{proof}
	It follows from Lemma~\ref{lem:horocycle_arclengths} that $\Isch(\varphi) = L(\Sigma_h)$. For each $j$, let $\beta_j$ be the angle between the outward normal vectors $N_j,N_{j+1}$ to $H_{\varphi(z_{j})}$, $H_{\varphi(z_{j+1})}$ respectively as in Lemma~\ref{lem:angle_sum}. By the Gauss--Bonnet theorem, 
	\begin{align*}
		-A(\Sigma_h) + \sum_{j=1}^n\int_{\Sigma_h^j} k \, \dd \ell_j + \sum_{j=1}^n \beta_j = 2\pi.
	\end{align*}
By Lemma~\ref{lem:angle_sum}, $\sum_{j=1}^n \beta_j = 2\pi$. For each $j$, $\Sigma_j$ is a piece of a horocycle with outward pointing normal vector. In Example~\ref{ex:horocycle_2} we saw that the horocycle with \textit{inward} pointing normal vector has $k\equiv 1$. Since $\Sigma_h^j$ is oriented with the outward pointing normal vector, for all $j$, $\Sigma_h^j$ has geodesic curvature $k\equiv -1$. Hence altogether, $\sum_{j=1}^n \int_{\Sigma_h^j} k \, \dd \ell_j =-L(\Sigma_h)$, and rearranging the equation completes the proof. 
\end{proof}
   \begin{figure}
        \centering
        \includegraphics[width=0.4\linewidth]{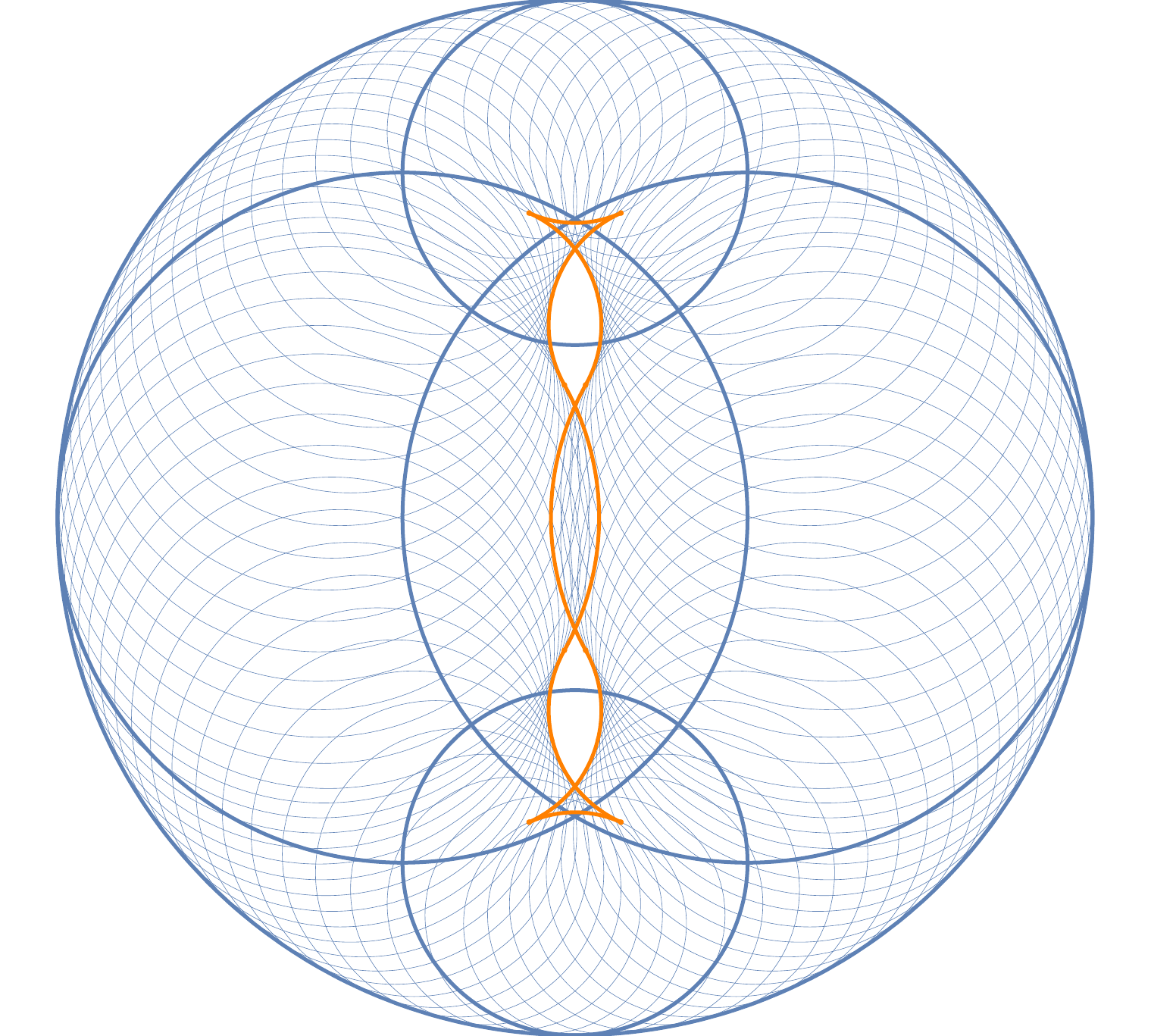}
        \caption{The orange curve here is the completed Epstein curve for $h_t=e^{t} h$ for $t=3/5$, where $h$ is as in Figure~\ref{fig:piecewise_mobius}. The orange dots mark the transitions between horocycle arcs and pieces of the image of $\Ep_{h_t}$. The four darker blue horocycles are for the original metric $h$ at the breakpoints, and the lighter horocycles correspond to the metric $h_t$.}
        \label{fig:piecewise_mobius_scaled}
    \end{figure}
\begin{remark}
    One can also consider the family of metrics $h_t = e^t h$, $h=\varphi_* \dd \theta$ for $\varphi$ which is $C^1$ piecewise M\"obius. Let $H_{\varphi(z)}^t$ denote the horocycles for $h_t$, which are the same as $h$ by flowing inward by a distance of $t$ along geodesics normal to the horocycle. 
    For any $t>0$, $\Ep_{h_t}$ is a union of circular arcs, where each arc $\Ep_{h_t}^j$ is obtained by flowing along geodesics towards $\m S^1$ by distance $t$ from the point $a_j=\alpha_j(0)$ in $\Ep_h$. The endpoints of $\Ep_{h_t}^j$ lie on $H_{\varphi(z_j)}^t,H_{\varphi(z_{j+1})}^t$. 
    We define $\Sigma_{h_t}$ as before to be the completion of $\Ep_{h_t}$ by adding the arcs on the horocycles based at the breakpoints, which are also the images of the horocycle arcs in $\S_h$ flowed by distance $t$. See Figure~\ref{fig:piecewise_mobius_scaled}. 
    We also obtain that for all $t >0$, $\S_{h_t}$ has only $0$ or $2\pi$ angles.
    \end{remark}    
    

\appendix

\section{Examples of Epstein curves}\label{sec:examples}

Here, we include some examples of the horocycles produced by Epstein's construction, as well as the associated Epstein curve and foliation. These are made for the measure $h=\ee^{\sigma(\theta)} \, \dd \theta$. Mathematica code, which can be used to generate other examples, is below.

\begin{ex} 
Let $\alpha(z) = (z+a)/(1+\overline{a} z)$, where $a=\ee^{\ii \pi/3}/3$, which is a M\"obius transformation that sends $0$ to $a$. Below are the horocycles and Epstein curves for the metrics $h_t = e^t \alpha_* \dd \theta$, where $t$ ranges from $0$ to $1.4$ in intervals of $0.2$. Note that $\Ep_{h_0}$ is the single point $\alpha(0) = a$. The Epstein foliation $(\Ep_{h_t})_{t\geq 0}$ consists of circles centered at $a$ with radius $t$. This can be seen from Example~\ref{ex:Ep_point} and Lemma~\ref{lem:naturality}.
    \begin{center}
\includegraphics[width=\textwidth]{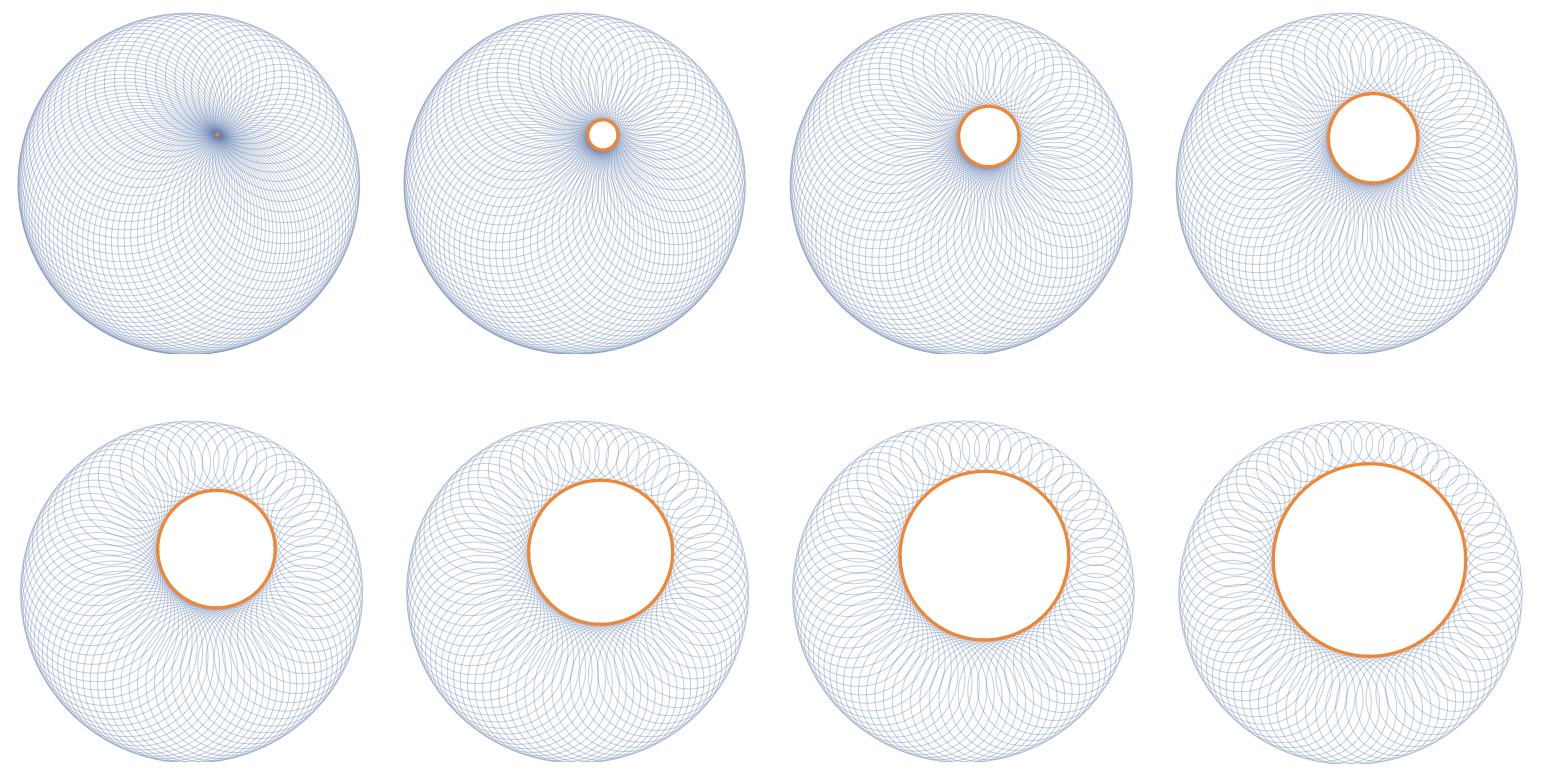}
    \end{center}
\end{ex}

\begin{ex}
    Let $\varphi(\theta) = \frac{1}{2}\sin(\theta) + \theta$. This is a diffeomorphism of the circle. Below are the horocycles and Epstein curve for $h=\varphi_*\dd \theta$ and the Epstein curves for $h_t= e^t h$, $t\in [0,5]$.
    \begin{center}
         \includegraphics[width=0.75\textwidth]{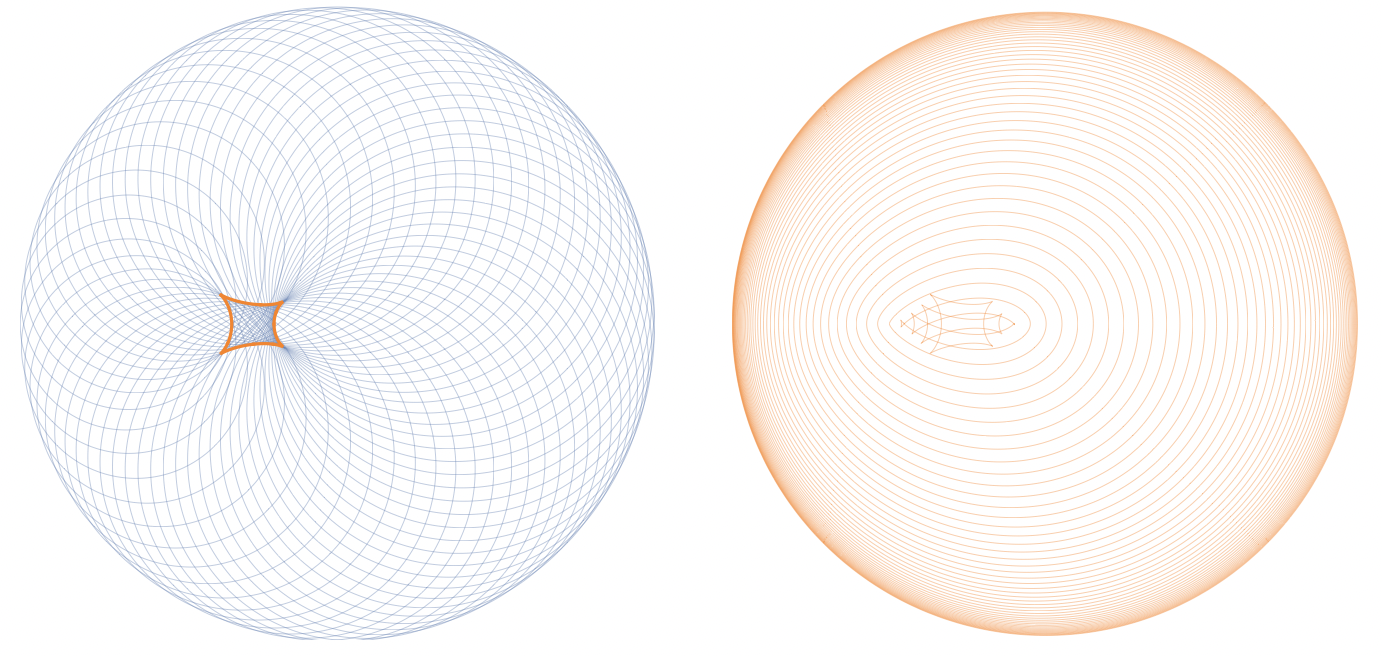}
    \end{center}
    Let $\phi=\varphi^{-1}$; this is also a diffeomorphism of the circle. Below are the horocycles and Epstein curve for $\hat{h}=\phi_*\dd \theta$ and the Epstein curves for $\hat{h}_t=e^t \hat{h}$, $t\in [0,5]$.
    \begin{center}
        \includegraphics[width=0.75\textwidth]{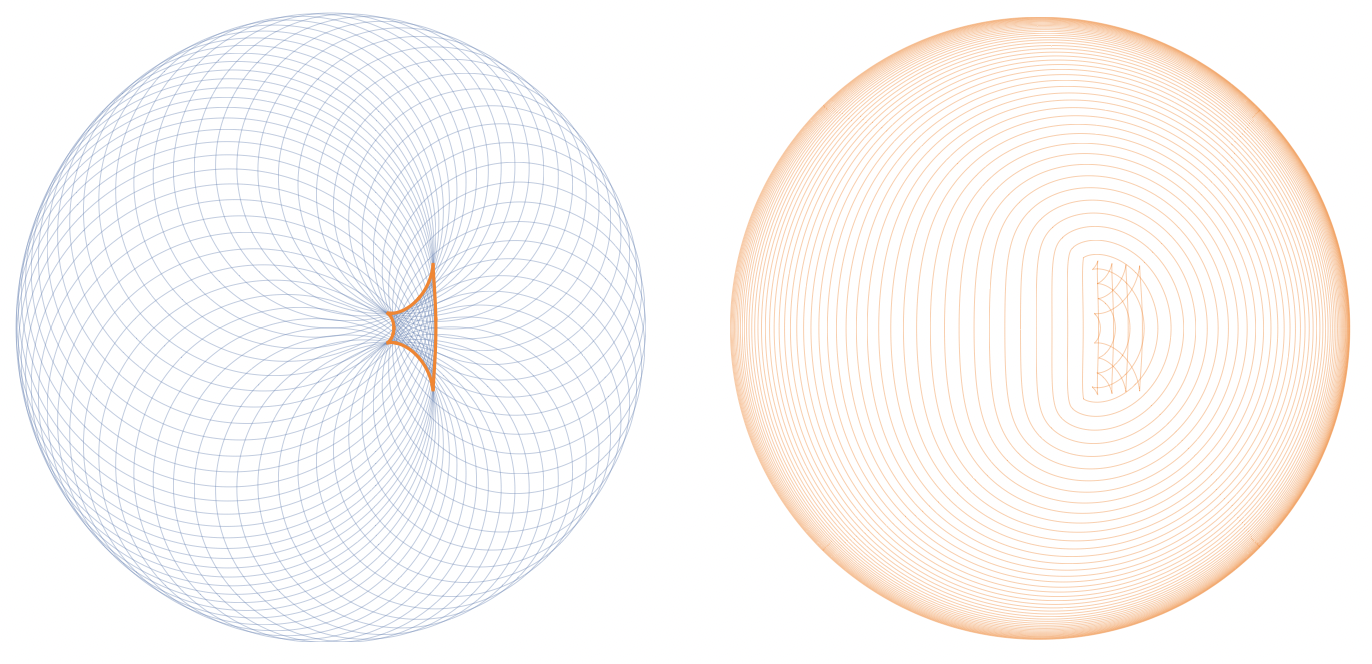}
    \end{center}
    Note that $\Ep_h$ and $\Ep_{\hat h}$ both have non-immersed points as required by Proposition~\ref{prop:non-immersed}.
\end{ex}

\begin{ex}
    Let $h=\exp(\cos^2(\theta/2))\, \dd \theta$. The total length of this metric is $\approx 8.19$, in particular it is greater than $2\pi$.
    \begin{center}
        \includegraphics[width=0.76\textwidth]{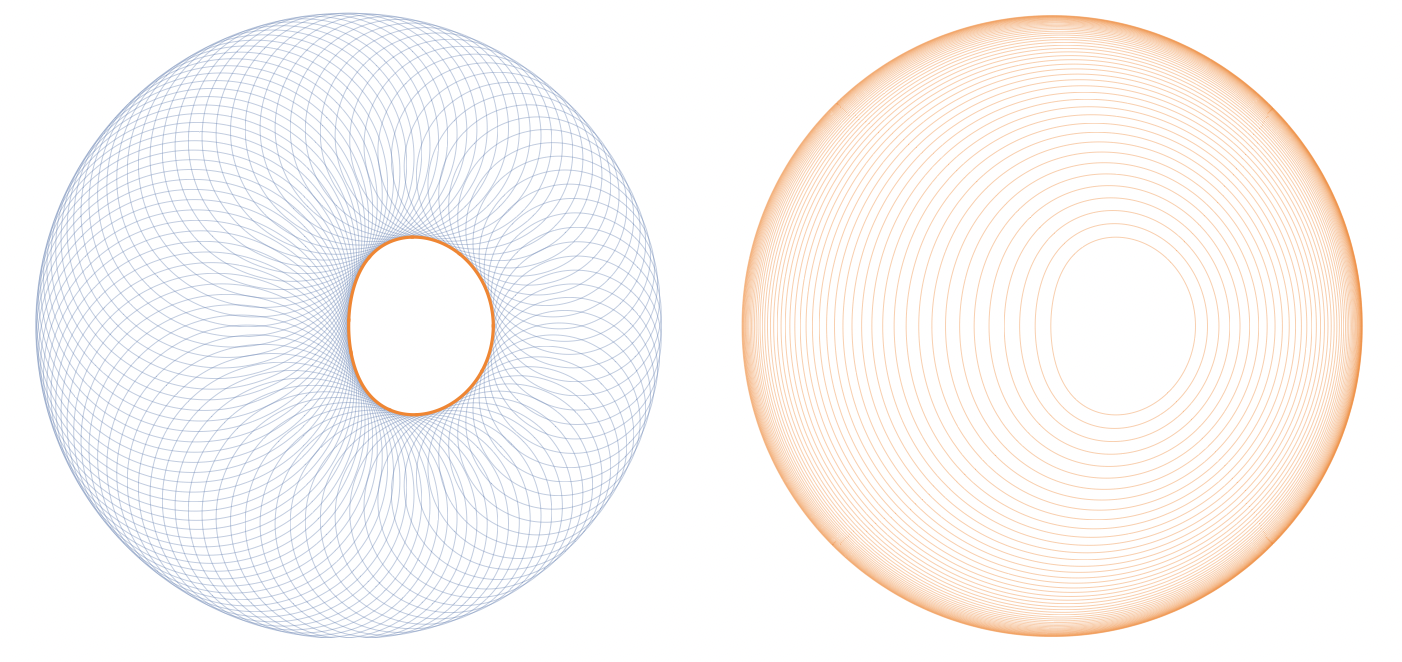}
    \end{center}
\end{ex}

\begin{ex}\label{ex:z3_conjugated}
    Let $h=6\sin^2(\theta)(5+3\cos(2\theta))^{-1}\, \dd \theta = |\phi'(\theta)| \dd \theta$, where $\phi(z) = (1 + 3 z^2)/(3 z + z^3)$ for $z\in \m S^1$. In particular $\phi=\mathfrak{c}^{-1}\circ f \circ \mathfrak{c}$ where $f:\m R\to \m R:x\mapsto x^3$ and $\mathfrak{c}(z) = \ii \frac{1-z}{1+z}$ is the Cayley map. Note that $\int h = 2\pi$, but $\phi$ is not a diffeomorphism and $h$ is degenerate at $-1, 1$. 
    The Epstein curve $\Ep_h$ contains the points $1$ and $-1$, has infinite length, and encloses infinite area.
    \begin{center}
        \includegraphics[width=0.39\textwidth]{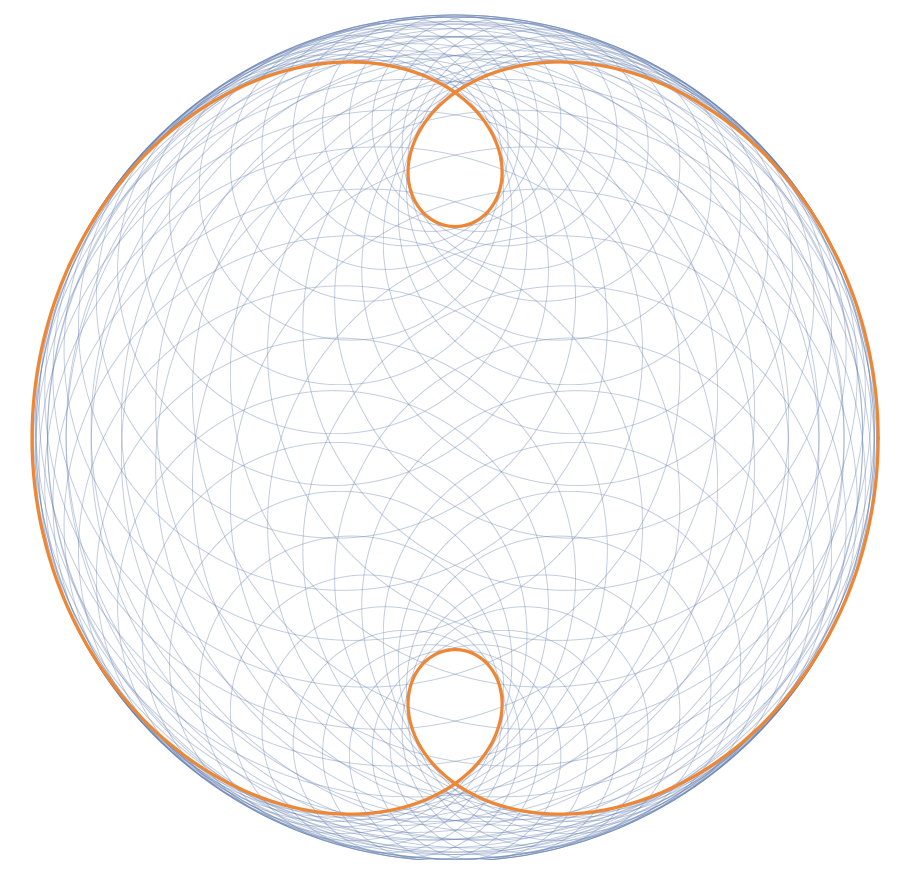}
    \end{center}
    For $a\in[0,1]$, we can define the metrics $h^a=((1-a)6\sin^2(\theta)(5+3\cos(2\theta))^{-1} +a)\dd \theta$. These all have $\int h^a=2\pi$ and are non-degenerate for $a>0$. Below, we plot the horocycles and Epstein curves for $h^a$ for $a=0.1,0.3,0.5,0.7.$
    \begin{center}
        \includegraphics[width=\textwidth]{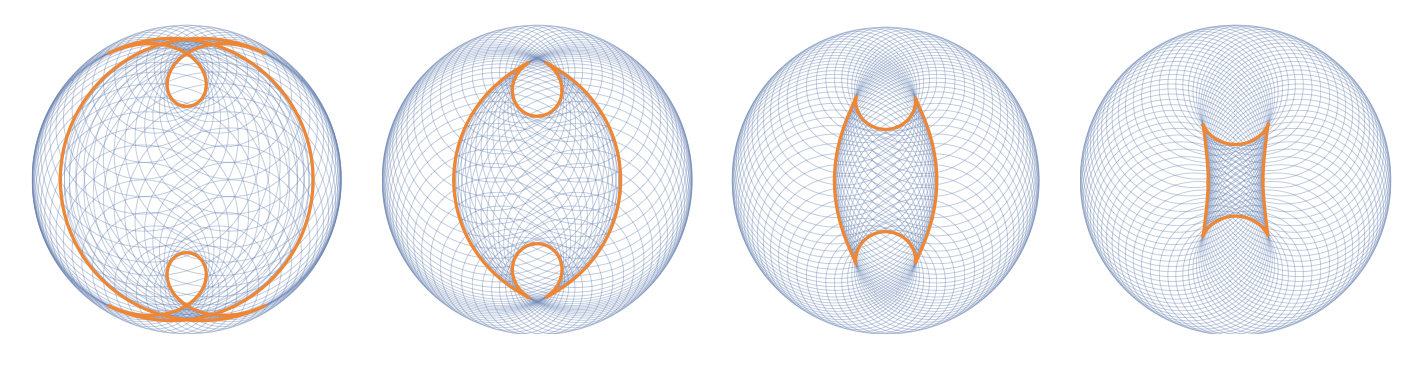}
    \end{center}
\end{ex}

For any $C^2$ metric $h=\ee^{\sigma(\theta)} \dd \theta$, these pictures are generated by plotting in Mathematica with the formulas from Section~\ref{subsec:eps_1}. The first gives the equation for the horocycle $H_{\ee^{\ii \theta}}$ as a function of $x\in[0,2\pi]$, and the second gives the equation for $\Ep_h$. 

\begin{verbatim}
  horocycle[\[Sigma]_, \[Theta]_] :=  
 E^(\[Sigma])/(1 + E^(\[Sigma])) E^(I \[Theta])  + 
  1/(1 + E^(\[Sigma])) E^(I x)
\end{verbatim}

\begin{verbatim}
epstein[\[Sigma]_, \[Theta]_] := (D[\[Sigma], \[Theta]]^2 + 
      E^(2 \[Sigma]) - 
      1)/(D[\[Sigma], \[Theta]]^2 + (E^(\[Sigma]) + 
         1)^2) E^(I \[Theta]) + 
  2 D[\[Sigma], \[Theta]]/(D[\[Sigma], \[Theta]]^2 + (E^(\[Sigma]) + 
         1)^2) I E^(I \[Theta])
\end{verbatim}

To plot, e.g., $\Ep_h$ (in orange) and the horocycles at $\ee^{\ii \theta}$ for $\theta$ a multiple of $\pi/40$ (in blue), one can plug the appropriate $\sigma(\theta)$ into the following. 

\begin{verbatim}
Show[{ParametricPlot[#, {x, 0, 2 Pi}, PlotRange -> {-1, 1}, 
     PlotStyle -> Thin, Axes -> False] & /@ 
   Table[ReIm[horocycle[\[Sigma], \[Theta]]], {\[Theta], 0, 2 Pi, 
     Pi/40}], 
  ParametricPlot[
   Evaluate[ReIm[epstein[\[Sigma], \[Theta]]]], {\[Theta], 0, 2 Pi}, 
   Axes -> False, PlotStyle -> Orange]}]
\end{verbatim}

The pictures of Epstein curves $(\Ep_{h_t})_t$ for scaled metrics $h_t = \ee^t \ee^{\sigma} \, \dd \theta$ can be made similarly by plotting the Epstein curves for $\sigma(\theta) + t$, e.g., for $t\in [0,5]$ in increments of $0.1$. Note that removing the ``Show'' function at the beginning will generate a single plot for each Epstein curve instead of plotting them on top of each other in one picture.

\begin{verbatim}
    Show[Table[ParametricPlot[Evaluate[
    ReIm[epstein[\[Sigma] + t, \[Theta]]]], {\[Theta],0, 2 Pi},
    Axes -> False, PlotStyle -> {Orange, Thin}, 
    PlotRange -> {-1, 1}], {t, 0, 5, 0.1}]]
\end{verbatim}

\bibliographystyle{plain}
\bibliography{ref}

\end{document}